\documentclass[prodmode,fleqn,acmtoplas]{acmsmall}
\usepackage{amsfonts}
\usepackage{amsmath}
\usepackage{amssymb}
\usepackage{xspace}
\usepackage{xcolor}
\usepackage{soul}
\usepackage{xspace}
\usepackage{url}
\usepackage{algorithm}
\usepackage{algpseudocode}
\usepackage{comment}
\usepackage{booktabs}
\usepackage{longtable}
\usepackage{tabularx}
\usepackage[mathscr]{euscript}
\usepackage{tikz}
\usetikzlibrary{arrows.meta}

\newdef{casestudy}{Case study}

\excludecomment{1}

\excludecomment{2}

\excludecomment{3}

\excludecomment{10}

\catcode`\|=\active
\catcode`\@=\active
\def@{\mathbin{\bullet}}
\def|{\mid}

\def\@p#1{\mathrel{\ooalign{\hfil$\mapstochar\mkern
      5mu$\hfil\cr$#1$}}}
\def \pfun      {\@p\fun}
\let \fun       \rightarrow

\newcommand{\comp}{\circ}
\newcommand{\lfun}{\longrightarrow}
\newcommand{\imp}{\implies}

\newcommand{\defs}{\triangleq}
\newcommand{\num}{\mathbb{Z}}
\newcommand{\Set}{\mathit{set}}
\newcommand{\set}[2]{\{#1 \sqcup #2\}}
\renewcommand{\Cup}{\mathit{un}}
\newcommand{\Ncup}{\mathit{nun}}
\newcommand{\Pfun}{\mathit{pfun}}
\newcommand{\disj}{\parallel}
\newcommand{\Disj}{\mathit{disj}}
\newcommand{\Ndisj}{\not\disj}

\newcommand{\Diff}{\mathit{diff}}
\renewcommand{\Cap}{\mathit{inters}}
\newcommand{\Ncap}{\mathit{ninters}}
\newcommand{\Ndiff}{\mathit{ndiff}}
\newcommand{\Apply}{\mathit{apply}}
\newcommand{\nat}{\mathbb{N}}
\renewcommand{\mod}{\mathbin{\mathrm{mod}}}

\newcommand{\true}{\mathit{true}}
\newcommand{\false}{\mathit{false}}

\newcommand{\Comp}{\mathit{comp}}

\newcommand{\Is}{\mathbin{\mathsf{is}}}
\newcommand{\Neq}{\mathbin{\mathsf{neq}}}
\newcommand{\In}{\mathbin{\mathsf{in}}}
\newcommand{\Nin}{\mathbin{\mathsf{nin}}}
\newcommand{\Div}{\mathbin{\mathsf{div}}}
\newcommand{\Mod}{\mathbin{\mathsf{mod}}}
\newcommand{\Ris}{\mathsf{ris}}
\newcommand{\Forall}{\mathsf{foreach}}
\newcommand{\Int}{\mathsf{int}}
\renewcommand{\And}{\mathbin{\&}}
\newcommand{\Or}{\mathbin{\mathsf{or}}}

\newcommand{\Ncomp}{\mathit{ncomp}}

\newcommand{\setlog}{$\{log\}$\xspace}
\newcommand{\CLPSET}{CLP($\mathcal{SET}$)\xspace}
\newcommand{\CLPFD}{CLP($\mathcal{FD}$)\xspace}

\newcommand{\CLPRIS}{$\LRIS$\xspace}
\newcommand{\SATRIS}{\mathit{SAT}_\mathcal{RIS}}
\newcommand{\SATX}{\mathit{SAT}_\mathcal{X}}
\newcommand{\LRIS}{\mathcal{L}_\mathcal{RIS}}
\newcommand{\LX}{\mathcal{L}_\mathcal{X}}

\newcommand{\CRIS}{\mathcal{C}_\mathcal{RIS}}
\newcommand{\TRIS}{\mathcal{T}_\mathcal{RIS}}
\newcommand{\TX}{\mathcal{T}_\mathcal{X}}
\newcommand{\Var}{\mathcal{V}}

\newcommand{\sU}{\mathsf{X}}
\newcommand{\sSet}{\mathsf{Set}}

\newcommand{\Ur}{\mathcal{X}}
\newcommand{\FUr}{\mathcal{F}_\Ur}
\newcommand{\sS}{\mathsf{S}}
\newcommand{\sB}{\mathsf{Bool}}
\newcommand{\FRIS}{\Phi_\mathcal{RIS}}
\newcommand{\FX}{\Phi_\mathcal{X}}

\newcommand{\PiSB}{\Pi_\mathcal{S}}
\newcommand{\isx}{\mathit{isX}}

\newcommand{\iS}{\mathcal{R}}
\newcommand{\iF}[1]{(#1)^\iS}

\newcommand{\plus}{\mathbin{\scriptstyle\sqcup}}
\newcommand{\ww}{\{\cdot \plus \cdot\}}
\newcommand{\e}{\emptyset}

\newcommand{\cross}{\mathsf{\:x\:}}
\newcommand{\arr}{\texttt{\:->\:}}

\newcommand{\risnopattern}[3]{\{ #1 : #2 | #3\}}

\newcommand{\risnocp}[2]{\{ #1 | #2\}}
\newcommand{\riss}[3]{\{ #1 | #2 @ #3\}}

\newcommand{\defris}[1]{\riss{#1}{\Fpv}{\Ppv}}

\newcommand{\defrisinit}[0]{\defris{\set{d}{D}}}

\newcommand{\svf}{\mathcal{S}}
\newcommand{\cvf}{\mathcal{C}}
\newcommand{\idom}{\mathcal{D}}
\newcommand{\flt}{\phi}   
\newcommand{\ptt}{u}      
\newcommand{\ct}{c}       
\newcommand{\vp}{\vec{p}}
\newcommand{\vv}{\vec{v}}

\newcommand{\Fpv}{\phi}
\newcommand{\Ppv}{u}
\newcommand{\Gpv}{\gamma}
\newcommand{\Qpv}{v}

\newcommand{\F}{\Fpv}
\renewcommand{\P}{\Ppv}
\newcommand{\Q}{\Qpv}
\newcommand{\G}{\Gpv}
\newcommand{\Fd}{\Fpv(d)}
\newcommand{\Fz}{\Fpv(x)}

\newcommand{\Gz}{G(x,\vv)}
\newcommand{\Pd}{\Ppv(d)}
\newcommand{\Ge}{G(e,\vv)}
\newcommand{\Qe}{Q(e,\vv)}
\newcommand{\Fdd}{\Fpv(d)}
\newcommand{\nFdd}{\Fpv(d)}
\newcommand{\Pdd}{\Ppv(d)}

\newcommand{\Pz}{\Ppv(x)}

\newcommand{\Qz}{Q(x,\vv)}
\newcommand{\bi}[1]{\textbf{\textit{#1}}}

\newcommand{\why}[1]{\tag*{{\footnotesize [by #1]}}}

\newcommand{\by}[1]{{\footnotesize [by #1]}}

\renewcommand{\iff}{\Leftrightarrow}

\newcommand{\q}{\text{\normalfont\'{}}}
\newcommand{\ql}{\hspace{5pt}\text{\normalfont\'{}}}
\newcommand{\qr}{\text{\normalfont\'{}}\hspace{5pt}}

\newcommand{\Applysl}{\mathsf{apply}}
\newcommand{\True}{\mathsf{true}}
\newcommand{\False}{\mathsf{false}}
\newcommand{\Notinsl}{\mathbin{\mathsf{nin}}}
\newcommand{\Cupsl}{\mathsf{un}}
\newcommand{\Ncupsl}{\mathsf{un}}
\newcommand{\Disjsl}{\mathsf{disj}}
\newcommand{\Ndisjsl}{\mathsf{ndisj}}
\newcommand{\Neqsl}{\mathbin{\mathsf{neq}}}

\newcommand{\type}{\mathcal{T}}
\newcommand{\stype}{\mathsf{S}}
\newcommand{\ptype}{\mathsf{P}}
\newcommand{\utype}{\mathsf{U}}

\newcommand{\hlg}[1]{\sethlcolor{green}\hl{#1}}

\definecolor{formula}{gray}{0.9}

\title{Automated Reasoning with Restricted Intensional Sets}

\author{%
MAXIMILIANO CRISTI\'{A} \affil{Universidad Nacional de Rosario and CIFASIS, Argentina}
GIANFRANCO ROSSI \affil{Universit\`a di Parma, Italy}}

\begin{abstract}
Intensional sets, i.e., sets given by a property rather than by
enumerating elements, are widely recognized as a key feature to describe
complex problems (see, e.g., specification languages such as B and Z).
Notwithstanding, very few tools exist
supporting high-level automated reasoning on general formulas involving
intensional sets. In this paper we present a decision procedure for a
first-order logic language offering both extensional and (a restricted form of)
intensional sets (RIS). RIS are introduced as first-class citizens of the 
language and set-theoretical operators on RIS are dealt with as
constraints. 
Syntactic restrictions on 
RIS guarantee that the denoted sets are finite, though unbounded. The language
of RIS, called $\LRIS$, is parametric with respect to any first-order theory
$\mathcal{X}$ providing at least equality and a decision procedure 
for $\mathcal{X}$-formulas. In particular, we consider the instance of $\LRIS$
when $\mathcal{X}$ is the theory of hereditarily finite sets
and binary relations. 
We also present a working implementation of this instance as part of the
\setlog tool and we show through a number of examples and two case studies
that, although RIS are a subclass of general intensional sets, they are still
sufficiently expressive as to encode and solve many interesting problems.
Finally, an extensive empirical evaluation provides evidence that the tool can
be used in practice.
\end{abstract}

\acmformat{Maximiliano Cristi\'{a} and Gianfranco Rossi. 2019. Automated
Reasoning with Restricted Intensional Sets}

\begin{document}
\begin{bottomstuff}
Authors' addresses:
M. Cristi\'a, Universidad Nacional de Rosario and CIFASIS, Pellegrini 250, 2000 Rosario, Argentina; email: \url{cristia@cifasis-conicet.gov.ar}
; G. Rossi, Universit\`a di Parma, Dip. di Matematica, Via M.
D'Azeglio, 85/A, 43100 Parma, Italy; email: \url{gianfranco.rossi@unipr.it}
\end{bottomstuff}

\maketitle

\section{Introduction}\label{sec:intro}

In the practice of mathematics very often a set is denoted by providing a
property that the elements must satisfy, rather than by explicitly enumerating
all its elements. This is usually achieved by the well-known mathematical
notation
\[
\{\,x\,|\:\varphi(x)\,\}
\]
where $x$ is a variable (the \emph{control variable}) and $\varphi$
is a first-order formula containing $x$\/. The logical meaning of
such notation is the set $S$ such that $\forall x (x \in  S \iff
\varphi(x))$, that is, the set of all instances of $x$ for which
$\varphi(x)$ holds. Sets defined by properties are also known as
\emph{set comprehensions} or \emph{set abstractions} or
\emph{intensional defined sets}; hereafter, we will refer to such
sets simply as \emph{intensional sets}.

Intensional sets are widely recognized as a key feature to describe
complex problems, possibly leading to more readable and compact
programs than those based on conventional data abstractions. As a
matter of fact, various specification or modeling languages provide
intensional sets as first-class entities. For example, some form of
intensional sets are available in notations such as Z
\cite{Woodcock00} and B \cite{schneider2001b}, Alloy
\cite{Jackson00}, MiniZinc \cite{DBLP:conf/cp/NethercoteSBBDT07},
and ProB \cite{Leuschel00}.  Also, a few
programming languages support intensional sets. Among them, SETL
\cite{DBLP:books/daglib/0067831}, Python, and the logic programming
language G\"odel \cite{DBLP:books/daglib/0095081}.

However, as far as we know, none of these proposals
provides the ability to perform high-level automated reasoning on
general formulas involving intensional sets. No direct support for
reasoning about intensional sets seem to be readily available even
in state-of-the-art satisfiability solvers, as noted, for instance,
by Lam and Cervesato with reference to SMT solvers
\cite{DBLP:conf/smt/LamC14}. Generally speaking, such reasoning
capabilities would be of great interest to several communities, such
as programming, constraint solving, automated theorem proving, and
formal verification.

Some form of automated reasoning about intensional
sets is provided by (Constraint) Logic Programming (CLP) languages
with sets, such as LDL \cite{DBLP:journals/jlp/BeeriNST91} and CLP(SET) \cite{Dovier00}.
The processing of intensional
sets in these proposals is based on the availability of a
\emph{set-grouping} mechanism capable of collecting into an
extensional set all the elements satisfying the property
characterizing the given intensional definition. Restrictions are
often put on the form of the collected sets, such as non-emptiness
and/or groundness. Allowing more general forms of set-grouping,
however, would require the ability to deal with logical negation in
a rather non-trivial way (see
\cite{DBLP:conf/iclp/BruscoliDPR94,DBLP:journals/ngc/DovierPR01} for
an analysis of the relationship between intensional sets and
negation).

Actually, set-grouping is not always necessary to deal with
intensional sets, and sometimes it is not desirable, at all. For
instance, given the formula $t \in \{x | \varphi\}$, one could check
whether it is satisfiable or not by simply checking satisfiability
of $\varphi[x \mapsto t]$, i.e., of the instance of $\varphi$ which
is obtained by substituting $x$ by $t$.

In this paper, we present a complete constraint solver which can act as a
decision procedure for an important fragment of a first-order logic language
offering both extensional and intensional sets. Intensional sets are introduced
as first-class citizens of the logic language, and set-theoretical
operators on intensional sets, such as membership, union, etc., are dealt with
as \emph{constraints} in that language. Complex formulas involving intensional
sets are processed and solved by a suitable constraint solver by using a sort
of \emph{lazy partial evaluation}. That is, an intensional set $S$ is treated
as a block until it is necessary to identify one of its elements, say $d$. When
that happens, $S$ is transformed into an extensional set of the form $\{d\}
\cup R_S$, where $R_S$ is the ``rest'' of $S$. At this point, classic set
constraint rewriting (in particular set unification) is applied.

A similar approach, based on intensional set constraint solving, has
been proposed in \cite{DBLP:conf/iclp/DovierPR03}. Differently from
that work, however, we avoid problems arising if general intensional
sets are allowed (namely, the use of logical negation and possibly
infinite sets), by considering a narrower form of intensional sets,
called \emph{Restricted Intensional Sets} (RIS). RIS have similar
syntax and semantics to the set comprehensions available in the
formal specification language Z, i.e.,
\[
\{ x : D | \flt @ \ptt(x) \}
\]
where $D$ is a set, $\flt$ is a formula, and $\ptt$ is a term
containing $x$. The intuitive semantics of $\{ x : D | \flt @
\ptt(x) \}$ is ``the set of terms $\ptt(x)$ such that $x$ belongs to
$D$ and $\flt$ holds for $x$''.\footnote{In this notation, as in Z,
$x:D$ is interpreted as $x \in D$.} $\flt$ is a quantifier-free
formula over a first-order theory $\Ur$, for which we assume a
complete satisfiability solver is available. Moreover, RIS have the
restriction that $D$ must be a \emph{finite set}. This fact, along
with a few restrictions on variables occurring in $\flt$ and $\ptt$,
guarantees that the RIS is a finite set, given that it is at most as
large as $D$. It is important to note that, although RIS are
guaranteed to denote finite sets, nonetheless, RIS can be not
completely specified. In particular, as the domain can be a variable
or a partially specified set, RIS are finite but \emph{unbounded}.

In previous work, we have proposed a constraint language dealing with RIS,
called $\LRIS$ \cite{DBLP:conf/cade/CristiaR17}. $\LRIS$ formulas are Boolean
combinations of constraints representing set equality and set membership and
whose set terms can be both extensional sets and RIS. Furthermore, $\LRIS$ is
parametric w.r.t. the language of $\Ur$, in the sense that set elements and the
formulas inside RIS are $\Ur$-terms and quantifier-free $\Ur$-formulas,
respectively. In particular, $\Ur$ can be the theory of sets and binary
relations described in \cite{Cristia2019}. Besides, $\LRIS$ is endowed with a
complete solver, called $\SATRIS$, which can act as a decision procedure for an
important fragment of $\LRIS$ formulas (provided a decision procedure for $\Ur$
is available). In this paper we extend our previous work on RIS in several
ways:
\begin{enumerate}
\item The Boolean algebra of sets is now supported; this is achieved by extending
the collection of primitive set constraints admitting RIS to union ($\cup$) and
disjointness ($\disj$).
\item The decision procedure for $\LRIS$, $\SATRIS$,
 and its Prolog implementation are
extended accordingly.
\item Formulas inside RIS can be existentially quantified $\Ur$-formulas as
long as the quantified variables are special parameters of functional
predicates.
\item An extensive empirical evaluation of these extensions, based on problems
drawn from the TPTP library, is also made.
\end{enumerate}

Though RIS are \emph{restricted}, we aim to show that in spite of
these restrictions $\LRIS$ is still a very expressive language. In
particular, we will show that it allows \emph{Restricted Universal
Quantifiers} (RUQ) on \emph{finite} domains to be expressed as
$\LRIS$ formulas, thus allowing an important class of quantified
formulas to lay inside the decision procedure. Besides encoding RUQ,
RIS provide a sort of second-order language of sets as they allow to
iterate over sets of sets and are useful to express partial
functions, as well as a programming facility similar to list/set
comprehensions available in programming languages.

The paper is organized as follows. In Section \ref{lris} the syntax and
semantics of $\LRIS$ are introduced; in particular Section \ref{informal}
presents some basic examples of formulas involving RIS terms. Section
\ref{solver} lists and discusses the rewrite rules of $\SATRIS$, in particular
those for set union and disjointness. In Section \ref{deci} we precisely
characterize the admissible $\LRIS$ formulas involving RIS, and we prove that
$\SATRIS$ is a decision procedure for such formulas. Section \ref{sec:uses}
shows two important applications for RIS: RUQ and partial functions. The
extension allowing existentially quantified formulas inside a RIS term is
motivated and discussed in Section \ref{sec:extension}. Section \ref{impl}
presents the implementation of our approach as part of the \setlog tool
(pronounced 'setlog') \cite{setlog} where two case studies are also discussed:
the first one presents an automated proof of a non-trivial security property,
while the second  elaborates on the possibility to iterate over collections of
sets and shows how \setlog can be used as a programming tool and as a
verification tool. The results of the empirical evaluation of \setlog are
presented in Section \ref{experiments}. A comparison of our approach with other
works and our conclusions are presented in Sections \ref{relwork} and
\ref{concl}, respectively.

\section{Formal Syntax and Semantics}\label{lris}

This section describes the syntax and semantics of the language of Restricted
Intensional Sets, $\LRIS$. A gentle, informal introduction is provided in
Section \ref{informal}.

$\LRIS$ is a first-order predicate language with terms of two sorts: terms
designating sets and terms designating ur-elements. The latter are provided by
an external first-order theory $\Ur$ (i.e., $\LRIS$ is parametric with respect
to $\Ur$). $\Ur$ must include: a class $\Phi_\Ur$ of admissible $\Ur$-formulas
based on a set of function symbols $\FUr$ and a set of predicate symbols
$\Pi_\Ur$ (providing at least equality); an interpretation structure
$\mathcal{I}_\Ur$ with domain $\idom_\sU$ and interpretation function
${(\cdot)}^{\mathcal{I}_\Ur}$; and a decision procedure $\SATX$ for
$\Ur$-formulas. When useful, we will write $\LRIS(\Ur)$ to denote the instance
of $\LRIS$ based on theory $\Ur$.
On the other hand, $\LRIS$ provides special set constructors, and a handful of
reserved predicate symbols endowed with a pre-designated set-theoretic meaning.
Set constructors are used to construct both RIS and extensional sets. Set
elements are the objects provided by $\Ur$, which are manipulated through the
primitive operators that $\Ur$ offers. Hence, $\LRIS$ sets represent
\emph{untyped unbounded finite hybrid sets}, i.e., unbounded finite sets whose
elements are of arbitrary sorts. $\LRIS$ formulas are
built in the usual way by using conjunction, disjunction and negation of atomic
formulas. A number of complex operators (in the form of predicates) are defined
as $\LRIS$ formulas, thus making it simpler for the user to write complex
formulas.

\subsection{Syntax}\label{syntax}

Syntax is defined primarily by giving the signature upon which terms and
formulas of the language are built.

\begin{definition}[Signature]\label{signature}
The signature $\Sigma_\mathcal{RIS}$ of $\LRIS$ is a triple $\langle
\mathcal{F},\Pi,\Var\rangle$ where:
\begin{itemize}
\item $\mathcal{F}$ is the set of function symbols, partitioned as
$\mathcal{F} = \mathcal{F}_\mathcal{S}  \cup \FUr$, where
$\mathcal{F}_\mathcal{S}$ contains $\e$, $\ww$ and $\{\cdot:\cdot | \cdot\, @\,
\cdot\}$, while $\FUr$ contains the function symbols provided by the theory
$\Ur$ (at least, a constant and the binary function
symbol $(\cdot,\cdot)$).

\item $\Pi$ is the set of \emph{primitive} predicate symbols, partitioned as
$\Pi = \Pi_\mathcal{S} \cup \Pi_\mathcal{T} \cup \Pi_\Ur$, where
$\Pi_\mathcal{S}
\defs \{=_\mathcal{S}, \neq_\mathcal{S}, \in_\mathcal{S}, \not\in_\mathcal{S},
\Cup_\mathcal{S}, \disj_\mathcal{S}\}$, $\Pi_\mathcal{T}
\defs \{\Set, \isx\}$, while $\Pi_\Ur$
contains the predicate symbols provided by the theory $\Ur$ (at least $=_\Ur$).
%

\item $\Var$ is a denumerable set of variables, partitioned as
$\Var = \Var_\mathcal{S} \cup \Var_\Ur$. \qed
\end{itemize}
\end{definition}

Intuitively, $\e$ represents the empty set, $\{ t \plus A \}$ represents the
set $\{t\} \cup A$, and $\{\ct(\vec{x}):D | \flt(\vec{x}) @ \ptt(\vec{x})\}$,
where $\vec{x} \defs \langle x_1,\dots,x_n\rangle$, $n > 0$, is the vector of
all variables occurring in $\ct$, represents the set of all instances of $\ptt$
such that $\ct$ belongs to $D$ and $\flt$ holds. $=_\Ur$ is interpreted as
the identity in $\idom_\sU$, while $(\cdot,\cdot)$ will be used to represent
ordered pairs.

$\LRIS$ defines two sorts, $\sSet$ and $\sU$, which intuitively represent the
domain of set objects and the domain of non-set objects (or ur-elements). We
also assume a sort $\sB$ is available, representing the two-valued domain of
truth values $\{\mathsf{false},\mathsf{true}\}$.

To each variable in $\Var$ and each constant in $\mathcal{F}$ we associate a
sort $\mathsf{S}$, while to each function symbol in $\mathcal{F}$ of arity $n
\ge 1$ we associate a string $\mathsf{S}_1 \cross \dots \cross \mathsf{S}_n
\arr \mathsf{S}$, where $\mathsf{S}, \mathsf{S}_i \in \{\sSet,\sU\}$. Moreover,
if $t$ is a variable or a constant and $\mathsf{S}$ is the associated sort, or
$t$ is a term $h(t_1,\dots,t_n)$, $n \ge 1$, and $\mathsf{S}_1 \cross \dots
\cross \mathsf{S}_n \arr \mathsf{S}$ is the sort associated to $h$, then we say
that $t$ is of sort $\mathsf{S}$ and we write $t:\mathsf{S}$.

\begin{definition}[Sorts of function symbols]\label{d:sorts}
The sorts of the symbols defined in $\mathcal{F}$ are as follows:
\begin{gather*}
\e: \sSet \\
\ww: \sU \cross \sSet \arr \sSet \\
\{\cdot : \cdot | \cdot\, @\, \cdot\}: \sU \cross \sSet \cross \sB \cross \sU \arr \sSet \\
s:\overbrace{\sU \cross \dots\cross \sU}^{n_s} \arr \sU\text{, if $s \in \FUr$,
for some $n_s \geq 0$} \\
v:\sSet\text{, if $v \in \Var_\mathcal{S}$} \\
v:\sU\text{, if $v \in \Var_\Ur$}
\end{gather*}
   \qed
\end{definition}

In view of the intended interpretation, terms of sort $\sSet$ are called
\emph{set terms}; in particular, set terms of the form $\ww$ are
\emph{extensional set terms}, whereas set terms of the form $\{\cdot : \cdot |
\cdot\, @\, \cdot\}$ are \emph{RIS terms}\footnote{The form of RIS terms is
borrowed from the form of set comprehension expressions available in Z.};
variable set terms are simply called \emph{set
variables}.

Note that terms that constitute the elements of sets are all of sort $\sU$.
Also note that in a RIS term, $\{\ct:D | \flt @ \ptt\}$, $D$ (called \emph{domain})
is a set term, $\flt$ (called \emph{filter}) is a $\Ur$-formula, while $\ct$
(called \emph{control term}) and $\ptt$ (called \emph{pattern}) are
$\Ur$-terms.


Terms of $\LRIS$---called \emph{$\mathcal{RIS}$-terms}---are built from symbols
in $\mathcal{F}$ and $\mathcal{V}$ as follows.

\begin{definition}\label{RIS-terms}
[$\mathcal{RIS}$-terms] \label{d:terms} Let $\TRIS^0$ be the set of
terms generated by the following grammar:
\begin{gather*}
\TRIS^0 ::= \mathit{Elem} \hspace{2pt}|\hspace{2pt} \mathit{Set} \\
\mathit{Elem}::= \TX \hspace{2pt}|\hspace{2pt} \Var_\Ur \\
\begin{split}
\mathit{Set} ::=
  & \q\e\qr \\
  & \hspace{2pt}|\hspace{2pt} \mathit{Ris} \\
  & \hspace{2pt}|\hspace{2pt}
      \q\{\qr \mathit{Elem}
           \ql\hspace{-2pt}\plus\hspace{-2pt}\qr \mathit{Set} \ql\}\q \\
  & \hspace{2pt}|\hspace{2pt} \Var_\mathcal{S}
\end{split} \\
\mathit{Ctrl} ::=
  \Var_\Ur
  \hspace{2pt}| \ql(\q~~\mathit{Ctrl}~~\ql,\q~~\mathit{Ctrl} \ql)\q \\
\mathit{Pattern} ::=
  \Var_\Ur
  \hspace{2pt}| \ql(\q~~\mathit{Ctrl}~~\ql,\q~~\TX \ql)\q \\
\mathit{Ris} ::=
  \ql\{\qr \mathit{Ctrl} \ql:\qr \mathit{Set}
  \ql\hspace{-2pt}|\hspace{-2pt}\qr \FX \ql@\qr \mathit{Pattern} \ql\}\qr
\end{gather*}
where $\TX$ and $\FX$ represent the set of non-variable $\Ur$-terms and the set
of $\Ur$-formulas built using symbols in $\mathcal{F}_\Ur$ and $\Pi_\Ur$,
respectively.

The set of \emph{$\mathcal{RIS}$-terms}, denoted by $\TRIS$, is the maximal
subset of $\TRIS^0$ complying with the sorts as given in Definition
\ref{d:sorts} and respecting the following restriction on RIS terms: if $c$ is
the control term of a RIS, then its pattern can be either $c$ or $(c,t)$.
\qed
\end{definition}

The special form of control terms and patterns will be
precisely motivated and discussed in Section \ref{sec:extension}.

Sets denoted by both extensional set terms and RIS terms can be \emph{partially
specified} because elements and sets can be variables. In particular, RIS can
have a variable domain.

\begin{definition}[Variable-RIS]
A RIS term is a \emph{variable-RIS} if its domain is a variable or
(recursively) a variable-RIS; otherwise it is a \emph{non-variable RIS}. \qed
\end{definition}

Variable-RIS will be of special interest since they can be easily turned
into the empty set by substituting their domains by the empty set.

\begin{definition}[Sorts of predicate symbols]\label{d:pred_sorts}
The sorts of the predicate symbols in $\Pi_\mathcal{S} \cup
\Pi_\mathcal{T}$ are as follows (symbols $=$, $\neq$, $\in$, $\notin$ and
$\disj$ are infix; all other symbols are prefix):
\begin{gather*}
=_\mathcal{S}, \neq_\mathcal{S}: \sSet \cross \sSet \arr \sB \\
\in_\mathcal{S}, \not\in_\mathcal{S}: \sU \cross \sSet  \arr \sB \\
\Cup_\mathcal{S}: \sSet \cross \sSet \cross \sSet \arr \sB \\
\disj_\mathcal{S}: \sSet \cross \sSet \arr \sB \\
\Set, \isx: \sSet \cup \sU  \arr \sB
\end{gather*}
   \qed
\end{definition}

\begin{definition}[$\mathcal{RIS}$-constraints]\label{primitive-constraint}
A (primitive) \emph{$\mathcal{RIS}$-constraint} is any atom built from symbols
in $\Pi_\mathcal{S} \cup \Pi_\mathcal{T}$ complying with the sorts as given in
Definition \ref{d:sorts}. The set of $\mathcal{RIS}$-constraints is denoted by
$\CRIS$. $\mathcal{RIS}$-constraints based on a symbol $\pi$ (resp., on a set
of symbols $\Pi$) are called \emph{$\pi$-constraints} (resp.,
\emph{$\Pi$-constraints}). \qed
\end{definition}

Formulas of $\LRIS$---called \emph{$\mathcal{RIS}$-formulas}---are built
from $\mathcal{RIS}$-constraints and $\Ur$-formulas as follows.

\begin{definition}[$\mathcal{RIS}$-formulas]\label{formula}
The set of $\mathcal{RIS}$-formulas, denoted by $\FRIS$, is given by the
following grammar.
\begin{gather*}
\FRIS ::= \true | \CRIS | \FRIS \land \FRIS |
\FRIS \lor \FRIS | \mathit{\FX}
\end{gather*}
$\mathcal{RIS}$-formulas not containing any $\Ur$-formula are called
\emph{pure $\mathcal{RIS}$-formulas}.
 \qed
\end{definition}

\begin{remark}[Notation]\label{r:notation}
The following notational conventions are used throughout this paper.
$s,t,u,a,b,c,d$ (possibly subscripted) stand for arbitrary
terms of sort $\sU$; while $A,B,C,D$ stand for arbitrary
terms of sort $\sSet$ (either extensional or intensional, variable or not).
When it is useful to be more specific, we will use
$\bar{A},\bar{B},\bar{C},\bar{D}$ to represent either variables of sort $\sSet$
or variable-RIS; $X,Y,Z,N$ to represent variables of sort $\sSet$
that are not variable-RIS; and $x,y,z,n$ to represent variables of
sort $\sU$. $\Phi, \Gamma, \Lambda$ stand for $\mathcal{RIS}$-formulas, while
$\phi, \gamma, \lambda$ stand for $\Ur$-formulas; $p,q,r$ stand for atomic
formulas/constraints.

Moreover, a number of special notational conventions are used for the sake of
conciseness. Specifically, we will write $\{t_1 \plus \{t_2 \plus \cdots \{ t_n
\plus A\}\cdots\}\}$ (resp., $\{t_1 \plus \{t_2 \plus \cdots \{ t_n \plus
\e\}\cdots\}\}$) as $\{t_1,t_2,\dots,t_n \plus A\}$ (resp.,
$\{t_1,t_2,\dots,t_n\}$). We will also use the notation $[m,n]$, $m$ and $n$
integer constants, as a shorthand for $\{m,m+1,\dots,n\}$. As concerns RIS,
when the pattern is the control term and the filter is $\true$, they can be
omitted (as in Z), although one must be present. \qed
\end{remark}

\subsection{Semantics}

Sorts and symbols in $\Sigma_\mathcal{RIS}$ are interpreted according to
the interpretation structure $\iS = \langle \idom,\iF{\cdot}\rangle$, where
$\idom$ and $\iF{\cdot}$ are defined as follows.

\begin{definition} [Interpretation domain] \label{def:int_dom}
The interpretation domain $\idom$ is partitioned as $\idom = \idom_\sSet \cup \idom_\sU$ where:
\begin{itemize}
\item $\idom_\sSet$: the collection of all finite sets built from elements in
$\idom_\sU$

\item $\idom_\sU$: a collection of any other objects (not in $\idom_\sSet$).
\end{itemize}
\end{definition}

\begin{definition} [Interpretation function] \label{def:int_funct}
The interpretation function $\iF{\cdot}$ for symbols in $\Sigma_\mathcal{RIS}$
is defined as follows:

\begin{itemize}
\item Each sort $\sS \in \{\sU,\sSet\}$ is mapped to the domain $\idom_\sS$.

\item $\iS$ coincides with $\mathcal{I}_\Ur$ for symbols in $\FUr$ and $\Pi_\Ur$.

\item For each sort $\sS$, each variable $x$ of sort $\sS$ is mapped to
      an element $x^\iS$ in $\idom_\sS$.

\item The constant and function symbols in $\mathcal{F}_\mathcal{S}$ are
interpreted as follows:
  \begin{itemize}
  \item $\e$ is interpreted as the empty set $\emptyset$;
  \item $\{ t \plus A \}$ is interpreted as the set $\{t^\iS\} \cup A^\iS$;
  \item Let $\vv$ be a vector of free variables and $\vec{x}$ the vector
  of variables occurring in $\ct$, then the set
     $\{\ct(\vec{x}):D | \flt(\vec{x},\vv) @ \ptt(\vec{x},\vv)\}$
  is interpreted as the set
\[
\{y : \exists \vec{x} (\ct(\vec{x}) \in_\Ur D \land \flt(\vec{x},\vv) \land y
=_\Ur \ptt(\vec{x},\vv))\}
\]

Note that in RIS terms, $\vec{x}$ are bound variables whose scope is the RIS
itself, while $\vv$ are free variables possibly occurring in the formula where
the RIS is participating in.
  \end{itemize}

\item The predicate symbols in $\PiSB$ are interpreted as follows:
  \begin{itemize}
   \item $A =_\mathcal{S} B$ is interpreted as $A^\iS = B^\iS$,
   where $=$ is the identity relation in $\idom_\sSet$;
   \item $t \in_\mathcal{S} A$ is interpreted as $t^\iS \in A^\iS$,
    where $\in$ is the set membership relation in $D_\sSet$;
   \item $\Cup_\mathcal{S}(A,B,C)$ is interpreted as $A^\iS \cup B^\iS = C^\iS$,
    where $\cup$ is the set union operator in $D_\sSet$;
   \item $A \disj_\mathcal{S} B$ is interpreted as $A^\iS \cap B^\iS = \emptyset$,
    where $\cap$ is the set intersection operator in $\idom_\sSet$;
   \item $\isx(t)$ is interpreted as $t^\iS \in \idom_\sU$;
   \item $\Set(t)$ is interpreted as $t^\iS \in \idom_\sSet$;
   \item $A \neq B$ and $t \notin A$ are
   interpreted as $\lnot (A^\iS = B^\iS)$ and $\lnot (t^\iS \in A^\iS)$,
   respectively.
   \end{itemize}

\end{itemize}
\end{definition}

Equality between extensional set terms is
regulated by the following two equational axioms \cite{Dovier00}:
\begin{gather}
\{x, x \plus A\} = \{x \plus A\} \tag{$Ab$} \label{Ab} \\
\{x, y \plus A\} = \{y, x \plus A\} \tag{$C\ell$} \label{Cl}
\end{gather}
which state that duplicates in a set term do not matter (\emph{Absorption
property}) and that the order of elements in a set term is irrelevant
(\emph{Commutativity on the left}), respectively. These two properties capture
the intuitive idea that, for instance, the set terms $\{1,2\}$, $\{2,1\}$, and
$\{1,2,1\}$ all denote the same set. In other words, duplicates do not occur in
a set, but they may occur in the set term that denotes it.
On the other hand, equality between $\Ur$-terms is assumed to be managed by the
parameter theory $\Ur$.

\subsection{Examples of $\LRIS$ formulas using RIS}\label{informal}

In this section we present some simple examples of $\LRIS$ formulas, in
particular those where RIS terms play a major role, with the objective to help
the reader to understand the notation. A deeper discussion of the applicability
of $\LRIS$ can be found in Sections \ref{sec:uses} and \ref{impl}.

For the sake of presentation, in coming examples, we will assume that the
language of $\Ur$, $\LX$, provides the constant, function and predicate symbols
of the theory of the integer numbers. Moreover, we will
write $=$ (resp., $\neq, \in, \Cup$) in place of $=_\Ur$ and
$=_\mathcal{S}$ (resp., $\in_\Ur, \neq_\Ur, \Cup_\Ur$ and $\in_\mathcal{S},
\neq_\mathcal{S}, \Cup_\mathcal{S}$) whenever is clear from context.

\begin{example}\label{ex:first}
The following is an atomic $\mathcal{RIS}$-formula:
\[
\{x:[-2,2] | x \mod 2 = 0 @ x\} = \{-2,0,2\}
\]
This formula is clearly
satisfiable with respect to the intended interpretation $\iS$. The RIS term can be written more
compactly by omitting the pattern as it is the control variable:
\[
\{x:[-2,2] | x \mod 2 = 0\} = \{-2,0,2\}
\]
\qed
\end{example}

\begin{example}
Set membership is also a viable operation on RIS:
\[
(5,y) \in \{x:D | \true @ (x,x*x)\}
\]
where $D$ is a variable. This formula is satisfiable with respect to $\iS$
provided $y$ is assigned the value $25$ and $D$ is bound to any set containing
$5$. The RIS can also be written as $(5,y) \in \{x:D @ (x,x*x)\}$ because the
filter is $\true$. Note that the control term and the pattern verify the
condition stated in Definition \ref{RIS-terms}.
 \qed
\end{example}

The negation of a set membership constraint, as all negations in $\LRIS$, is
written by using the corresponding predicate symbol and not the logical negation
which, anyway, is not part of $\LRIS$.

\begin{example}
A $\mathcal{RIS}$-formula containing negative constraints:
\[
(5,0) \notin \{(x,y):\{z \plus Z\} | y \neq 0 @ (x,y)\}
\]
where $z$ and $Z$ are free variables. This formula is satisfiable
with respect to $\iS$ for any $z$ and $Z$. As above, the pattern
can be omitted. Observe that in $y \neq 0$
inequality corresponds to that of the theory $\Ur$. \qed
\end{example}

\begin{example}\label{ex:second}
The following is a non-atomic $\mathcal{RIS}$-formula involving RIS terms:
\[
A = \{x:D | x \neq 0\} \land \Cup(A,B,C) \land A \disj C \land A \neq \emptyset
\]
This formula turns out to be unsatisfiable with respect to $\iS$ since
$\Cup(A,B,C) \land A \disj C$ is satisfiable only if $A = \emptyset$, while the
input formula constraints $A$ to be distinct from $\emptyset$.
\qed
\end{example}

Finally, note that we allow RIS terms to be the set part of extensional set
terms, e.g., $\{z \plus \{x:A | x \neq y\}\}$, as well as to be the domain of
other RIS. $\LRIS$ also defines two constraints that are mainly used internally
by the solver, namely $\Set(t)$ and $\isx(t)$. Each one asserts that its
parameter is of sort $\sSet$ and $\sU$, respectively.

\subsection{\label{expressiveness}{Derived Constraints and Expressiveness}}

Dovier et al. \cite{Dovier00} proved that the collection of predicate symbols
in $\PiSB$ is sufficient to define constraints implementing other common set
operators, such as $\cap$ and $\subseteq$.

Specifically, let us consider the following predicate symbols: $\subseteq,
\Cap, \Diff$, along with their interpretations: $A \subseteq B$ is interpreted
as $A^\iS \subseteq B^\iS$; $\Cap(A,B,C)$ is interpreted as $C^\iS = A^\iS \cap
B^\iS$; $\Diff(A,B,C)$ is interpreted as $C^\iS = A^\iS \backslash B^\iS$.
Dovier et al. \cite{Dovier00} prove that these predicates can be made available
in $\LRIS$ without having to add them to the collection of its primitive
constraints. As an example, $A \subseteq B$ can be defined by the $\LRIS$
formula $\Cup(A,B,B)$.

With a little abuse of terminology, we say that atoms based on these
predicates are \emph{derived constraints}. Whenever a formula contains a
derived constraint, the constraint is replaced by its definition turning the
given formula into a $\LRIS$ formula. Precisely, if $\Phi$ is the ${\cal
RIS}$-formula defining the constraint $c$, then $c$ is replaced by $\Phi$ and
the solver checks satisfiability of $\Phi$ to determine satisfiability of $c$.
Thus, we can completely ignore the presence of derived constraints in the
subsequent discussion about constraint solving and formal properties of our
solver (namely, soundness and termination).

The negated versions of set operators can be introduced as derived constraints,
as well. Specifically, if $p$ is a predicate symbol in $\Pi$, a new
predicate $p'$ is associated to $p$ to obtain a positive form $p'(\vec t)$ for
literals $\neg p(\vec t)$. Then, the definition of $p'$ for $p \in
\{\Cup,\disj\}$ can be given as a ${\cal RIS}$-formula, hence, as a derived
constraint. The derived constraint for $\lnot\cup$ and $\lnot\disj$ (called
$\Ncup$ and $\Ndisj$, respectively) are shown in \cite{Dovier00}.
For example, $\lnot A \cup B = C$ is introduced as:
\begin{equation}\label{e:nun}
\Ncup(A,B,C) \defs
     (n \in C \land n \notin A \land n \notin B)
     \lor (n \in A \land n \notin C)
     \lor (n \in B \land n \notin C)
\end{equation}

The same approach can be used to define  the derived constraint for
$\lnot\subseteq$, $\lnot\cap$ and  $\lnot\backslash$, called $\not\subseteq$,
$\Ncap$ and $\Ndiff$, respectively. With a little abuse of terminology, we will
refer to atoms based on these predicates as \emph{negative constraints}. Thanks
to the availability of negative constraints, classical negation is not strictly
necessary in $\LRIS$.

As concerns expressiveness, now $\LRIS$ supports equality, set membership,
union, disjointness, intersection, difference and all their negations. Hence,
the Boolean algebra of sets is fully supported.

\section{A Solver for $\LRIS$}\label{solver}

In this section we present a constraint solver for $\LRIS$, called $\SATRIS$.
The solver provides a collection of rewrite rules for rewriting $\LRIS$
formulas that are proved to be a decision procedure for a large class of
$\LRIS$ formulas.
As already observed, however, checking the satisfiability of
$\mathcal{RIS}$-formulas depends on the existence of a decision procedure for
$\Ur$-formulas (i.e., formulas over $\LX$).

\subsection{The solver}

$\SATRIS$ is a rewriting system whose global organization is shown in Algorithm
\ref{glob}, where $\textsf{STEP}$ is the core of the algorithm.

\textsf{sort\_infer} is used to automatically add $\Pi_\mathcal{T}$-constraints
to the input formula $\Phi$ to force arguments of $\mathcal{RIS}$-constraints
in $\Phi$ to be of the proper sorts (see Remark \ref{rsort} below).
\textsf{sort\_infer} is called twice in Algorithm \ref{glob}: first, at the
beginning of the Algorithm, and second, within the procedure $\textsf{STEP}$
for the constraints that are generated during constraint processing.
\textsf{sort\_check} checks $\Pi_\mathcal{T}$-constraints occurring in $\Phi$:
if they are satisfiable, then $\Phi$ is returned unchanged; otherwise, $\Phi$
is rewritten to $\false$.

\algtext*{EndIf}
\algrenewtext{EndProcedure}{\textbf{return} }
\begin{algorithm}
\begin{tabular}{lr}
\begin{minipage}{.55\textwidth}
\begin{algorithmic}[0]
\Procedure{$\mathsf{STEP}$}{$\Phi$}
\State \hspace{5mm}
       \textsf{for all} $\pi \in \Pi_\mathcal{S} \cup \Pi_\mathcal{T}:
            \Phi \gets \mathsf{rw}_{\pi}(\Phi)$;
\State \hspace{5mm}
       $\Phi \gets \mathsf{sort\_check}(\mathsf{sort\_infer}(\Phi))$
\EndProcedure $\Phi$
\end{algorithmic}
\begin{algorithmic}[0]
 \Procedure{$\mathsf{rw}_\pi$}{$\Phi$}
 \If{$\false \in \Phi$}
 \State\Return{$\false$}
 \Else
 \Repeat
 \State \textsf{select a $\pi$-constraint $c$ in $\Phi$}
 \State \textsf{apply any applicable rule to $c$}
 \Until{\textsf{no rule applies to any $\pi$-constraint }}
 \EndIf
 \EndProcedure $\Phi$
\end{algorithmic}
\end{minipage}
&
\begin{minipage}{.35\textwidth}
\begin{algorithmic}[0]
\Procedure{$\SATRIS$}{$\Phi$}
 \State $\Phi \gets \textsf{sort\_infer}(\Phi)$
 \Repeat
 \State $\Phi' \gets \Phi$
 \Repeat
   \State $\Phi'' \gets \Phi$
   \State $\Phi \gets \textsf{STEP}(\Phi)$
 \Until{$\Phi = \Phi''$}
 \State $\Phi \gets \textsf{remove\_neq}(\Phi)$
 \Until{$\Phi = \Phi'$}
 \State $\Phi \textbf{ is } \Phi_\mathcal{S} \land \Phi_\Ur$
 \State $\Phi \gets \Phi_\mathcal{S} \land \SATX(\Phi_\Ur)$
\EndProcedure $\Phi$
\end{algorithmic}
\end{minipage}
\end{tabular}
\caption{\label{glob}The $\SATRIS$ solver. $\Phi$ is the input formula.}
\end{algorithm}

\textsf{remove\_neq} deals with the elimination of $\neq$-constraints involving
variables of sort $\sSet$, possibly occurring in
the formula $\Phi$ at the end of the innermost loop of $\SATRIS$.
Its motivation and definition will be made evident later in Section
\ref{removeneq}.

\textsf{STEP} applies specialized rewriting procedures to the current formula
$\Phi$ and returns either $\false$ or the modified formula. Each rewriting
procedure applies a few non-de\-ter\-mi\-ni\-stic rewrite rules which reduce
the syntactic complexity of $\mathcal{RIS}$-constraints of one kind. Procedure
$\mathsf{rw}_\pi$ in Algorithm \ref{glob} represents the rewriting procedure
for ($\Pi_\mathcal{S} \cup \Pi_\mathcal{T}$)-constraints. The execution of
$\textsf{STEP}$ is iterated until a fixpoint is reached---i.e., the formula
cannot be simplified any further. \textsf{STEP} returns $\false$ whenever (at
least) one of the procedures in it rewrites $\Phi$ to $\false$. In this case, a
fixpoint is immediately detected, since $\textsf{STEP}(\false)$ returns
$\false$. The rewrite rules for $\mathcal{RIS}$-constraints involving only
extensional set terms are described in Section \ref{rulesext}, while those for
$\mathcal{RIS}$-constraints involving both extensional and RIS terms are
detailed in Section \ref{ruleris}.

$\SATX$ is the constraint solver for $\Ur$-formulas. The formula $\Phi$ can be
seen, without loss of generality, as $\Phi_\mathcal{S} \land \Phi_\Ur$, where
$\Phi_\mathcal{S}$ is a pure $\mathcal{RIS}$-formula and $\Phi_\Ur$ is a
$\Ur$-formula. $\SATX$ is applied only to the $\Phi_\Ur$ conjunct of $\Phi$.
Note that, conversely, \textsf{STEP} rewrites only $\mathcal{RIS}$-constraints,
while it leaves all other atoms unchanged. Nonetheless,
as the rewrite rules show, $\SATRIS$ generates $\Ur$-formulas that are
conjoined to $\Phi_\Ur$ so they are later solved by $\SATX$.

As we will show in Section \ref{deci}, when
all the non-deterministic computations of $\SATRIS(\Phi)$ terminate with
$\false$, then we can conclude that $\Phi$ is unsatisfiable; otherwise, when
all the computations terminate and at least one of them does not return
$\false$, then we can conclude that $\Phi$ is satisfiable and each solution of
the formulas returned by $\SATRIS$ is a solution of $\Phi$, and vice versa.

\begin{remark}\label{rsort}
$\LRIS$ does not provide variable declarations. The sort of variables are
enforced by adding suitable \emph{sort constraints} to the formula to be
processed. Sort constraints are automatically added by the solver.
Specifically, a constraint $\Set(y)$ (resp., $\isx(y)$) is added for each
variable $y$ which is required to be of sort $\sSet$ (resp., $\sU$). For
example, given $X = \{y \plus A\}$, \textsf{sort\_infer} conjoins the sort
constraints $\Set(X)$, $\isx(y)$ and $\Set(A)$. If the set of function and
predicate symbols of $\mathcal{RIS}$ and $\Ur$ are disjoint, there is a unique
sort constraint for each variable in the formula.
\end{remark}

The rewrite rules used by $\SATRIS$ are defined as follows.

\begin{definition}[Rewrite rules]\label{d:rw_rules}
If $\pi$ is a symbol in $\Pi_\mathcal{S} \cup \Pi_\mathcal{T}$ and $p$ is a
$\mathcal{RIS}$-constraint based on $\pi$, then a \emph{rewrite rule for
$\pi$-constraints} is a rule of the form $p \lfun \Phi_1 \lor \dots \lor
\Phi_n$, where $\Phi_i$, $i \ge 1$, are $\mathcal{RIS}$-formulas. Each
atom matching $p$ is non-deterministically rewritten to
one of the $\Phi_i$. Variables appearing in the right-hand side but not in the
left-hand side are assumed to be fresh variables, implicitly existentially
quantified over each $\Phi_i$. \qed
\end{definition}

A \emph{rewriting procedure} for $\pi$-constraints consists of the collection of
all the rewrite rules for $\pi$-constraints. For each rewriting procedure,
\textsf{STEP} selects rules in the order they are listed in the figures below.
The first rule whose left-hand side matches the input $\pi$-constraint is used
to rewrite it. If no rule applies, then the input constraint is left unchanged
(i.e., it is irreducible).

\subsection{\label{rules}Rewrite rules for extensional set terms}\label{rulesext}

When $\LRIS$ formulas contain only extensional set terms, all the operators
defined in $\Pi_\mathcal{S}$ are dealt with the rewrite rules given in \cite{Dovier00}.
In particular, set equality between extensional set terms is implemented by
\emph{$(Ab)(C\ell)$-set unification} \cite{Dovier03}. One of the key rewrite
rules of set unification is the following, which implements axioms \eqref{Ab}
and \eqref{Cl} (adapted from \cite{Dovier00}):
\begin{alignat}{2}\label{e:eseteq}
\{t_1 \plus A\} & = \{t_2 \plus B\} \lfun  \notag\\
  & t_1 =_\Ur t_2 \land A = B \quad            & \lor\quad  & t_1 =_\Ur t_2
                                               \land \{t_1 \plus A\} = B \quad\qquad \lor \\
  & t_1 =_\Ur t_2 \land A = \{t_2 \plus B\} \quad & \lor\quad  & A = \{t_2 \plus N\}
                                               \land \{t_1 \plus N\} = B \notag
\end{alignat}
where $A$ and $B$ are extensional sets, $t_1$ and $t_2$ are arbitrary $\Ur$-terms, $N$ is a new variable (of sort $\sSet$), and
$=_\Ur$ is the equality provided by the theory $\Ur$. This means that every
time $\LRIS$ finds an atom such as the left-hand side
of rule \eqref{e:eseteq}, it attempts to find a solution for it in four
different ways. In some cases one or more will fail (i.e., they will be
$\false$) but in general $\LRIS$ will compute all the four solutions.

In this paper we will extend the set unification algorithm of
\cite{Dovier03} to deal with also RIS terms.


Given that the rules for extensional set terms has been extensively studied by
the authors
\cite{Dovier00,Dovier03,DBLP:journals/tplp/CristiaRF15,Cristia2019},
they are not shown here, and can be found in \cite{calculusBR}.

\subsection{Rewrite rules for RIS}\label{ruleris}

When a $\LRIS$ formula includes RIS terms, the rewrite rules for extensional
sets are extended in a rather natural way. The following is an example to
intuitively show how $\SATRIS$ works when RIS are involved.
\begin{example}\label{ex:rewriting}
If $S$ is a variable and $a,b$ are constants belonging to $\FUr$, then
$\SATRIS$ can determine the satisfiability of:
\[
\{(x,y):\{(a,1) \plus S\} | x \in \{a,b\}\} = \{(a,1)\}
\]
by proceeding as follows (the arrow indicates a rewriting step; some steps are
simplified):
\begin{gather*}
\{(x,y):\{(a,1) \plus S\} | x \in \{a,b\}\} = \{(a,1)\} \\
\fun (x,y) = (a,1)
     \land x \in \{a,b\}
     \land \{(a,1) \plus \{(x,y):S | x \in \{a,b\}\}\} = \{(a,1)\} \\
\quad{}\lor (x,y) = (a,1)
     \land x \notin \{a,b\}
     \land \{(x,y):S | x \in \{a,b\}\} = \{(a,1)\} \\
\fun x = a \land y = 1 \land (x = a \lor x = b) \\
\qquad{}\land (\{(x,y):S | x \in \{a,b\}\} = \emptyset
              \lor \{(x,y):S | x \in \{a,b\}\} = \{(a,1)\}) \\
\quad{}\lor \false \\
\fun x = a \land y = 1 \\
\quad{}\land (S = \emptyset
              \lor S = \{(a,1) \plus N\} \land (a,1) \notin N \land
              \{(x,y):N | x \in \{a,b\}\} = \emptyset)
\end{gather*}
where $N$ is a fresh variable. At this point, Algorithm \ref{glob}
stops thus returning a disjunction of two $\LRIS$ formulas:
\begin{gather*}
x = a \land y = 1 \land S = \emptyset \\
\lor x = a \land y = 1 \land
     S = \{(a,1) \plus N\} \land (a,1) \notin N \land \{(x,y):N | x \in \{a,b\}\} = \emptyset
\end{gather*}
which are surely satisfiable. The first one because all variables are bound to
terms and the second one because:
\[
(a,1) \notin N \land \{(x,y):N | x \in \{a,b\}\} = \emptyset
\]
is satisfiable by substituting $N$ by the empty set.
\qed
\end{example}

Hence, intuitively, the key idea behind the extension of the rewriting rules
for extensional set terms to RIS terms is a sort of \emph{lazy partial
evaluation} of RIS. That is, a RIS term is treated as a block until it is
necessary to identify one of its elements. When that happens, the RIS is
transformed into an extensional set whose element part is the identified
element and whose set part is the rest of the RIS. At this point, classic set
constraint rewriting (in particular set unification) can be applied. More
precisely, when a term such as $\{x:\{d \plus D\} | \flt @
\ptt\}$ is processed, two cases are considered:
\begin{itemize}
\item $x = d$ and $\flt$ holds for $d$, in which case the
RIS is rewritten as the extensional set $\{\ptt(d) \plus \{x:D | \flt @
\ptt\}\}$; and
\item $x = d$ but $\flt$ does not hold for $d$, i.e., $\neg\flt(d)$ holds,
in which case the RIS is rewritten as $\{x:D | \flt @ \ptt\}$.
\end{itemize}

The main rewrite rules for $\Pi_\mathcal{S}$-constraints dealing with RIS terms
are given in Figures \ref{f:eq}-\ref{f:disj} and are discussed below; rules for
special cases are presented in Appendix \ref{ap:rules}. In order to make the
presentation more accessible: \emph{a}) the rules are given for RIS whose
domain is not another RIS, in particular, the domain of a variable-RIS is a
single variable; and \emph{b}) the control term of RIS is always a
\emph{variable}. The generalization to cases in which these restrictions are
removed is discussed in Appendix \ref{ap:rules}. In these figures, a RIS of the
form $\{x:D | \flt @ \ptt\}$ is abbreviated as $\{D | \flt @ \ptt\}$. Besides,
$\Fpv(d)$, $\Gpv(d)$, $\Ppv(d)$ and $\Qpv(d)$ are shorthands for $\Fpv[x \mapsto d]$, $\Gpv[x \mapsto d]$,
$\Ppv[x \mapsto d]$ and $\Qpv[x \mapsto d]$, respectively, where $[x \mapsto
d]$ represents variable substitution. In all
rules, $\varnothing$ represents either $\emptyset$ or a RIS with empty domain
(e.g., $\{\emptyset | \flt @ \ptt\}$). Recall the notation used across the
paper given in Remark \ref{r:notation}.

\subsubsection{Equality}


Figure \ref{f:eq} lists the main rewrite rules applied by
\textsf{STEP} to deal with constraints of the form $A = B$ and $A \neq B$, where
$A$ and $B$ are $\LRIS$ set terms and at least one of them is a RIS term.

Equality between a RIS and an extensional set is governed by rules
\eqref{e:empty1}--\eqref{e:v=e}. In particular:
\begin{itemize}
\item Rule \eqref{e:empty2} deals with the case in which a RIS with a non-empty
domain must be equal to the empty set. It turns out that to force a RIS $\{D |
\Fpv @ \Ppv\}$ to be empty it is enough that the filter $\flt$ is false for all
elements in $D$, i.e., $\forall x \in D: \lnot \flt$ (see Proposition
\ref{prop:2} in the appendix). This restricted universal quantification is
conveniently implemented through recursion, by extracting one element $d$ at a
time from the RIS domain.

\item Rule \eqref{e:eset} deals with the case illustrated by
Example \ref{ex:rewriting}, i.e., equality between a non-variable RIS and
any other non-empty set (including variables and RIS).

\item Rule \eqref{e:v=e} deals with equality between a variable-RIS and an
extensional set. The intuition behind this rule is as follows. Given that $\{t
\plus A\}$ is not empty, then $\bar D$ must be not empty in which case it is
equal to $\{n \plus N\}$ for some $n$ and $N$. Furthermore, $n$ must satisfy
$\F$ and $\ptt(n)$ must be equal to $t$. As the first element of $\{t \plus
A\}$ belongs to the RIS, then the rest of the RIS must be equal to $A$. It is
not necessary to consider the case where $\lnot \flt(n)$, as in rule
\eqref{e:eset}, because $n$ is a fresh variable.
\end{itemize}

Other rules deal with the symmetric cases (e.g., $B = \{\{d \plus D\} | \F @
\P\}$, $\{t \plus A\} = \{\bar D | \F @ \P\}$), and with special cases where
some variables are shared between the two terms to be compared (e.g.,  $\{\bar
D | \F @ \P\} = \{t \plus \bar D\}$), see Appendix \ref{ap:rules}. The
cases not considered by any rule, namely, equality between a variable and a
variable-RIS, between a variable-RIS and the empty set, and between two
variable-RIS, are dealt with as irreducible (see Section \ref{solved}).

Negated equality is governed by rule \eqref{e:neq}. Note that this rule serves
for all combinations of set terms and it is based on the definition of set
inequality. In fact, $D$ can be either a variable, or a variable-RIS, or a set
term of the form $\{d \plus D\}$, while $A$ can be any term, including another
RIS. In this case, the problem is handed over to the rules for set membership.

\begin{figure}
\hrule
\begin{gather}
\{\e | \F @ \P\} = \varnothing \fun \true  \label{e:empty1} \\[2mm]
\{\e | \F @ \P\} = \{t \plus A\} \fun \false  \label{e:false} \\[2mm]
\{\{d \plus D\} | \F @ \P\} = \varnothing \fun
  \lnot \nFdd \land \{D | \F @ \P\} = \e
  \label{e:empty2} \\[2mm]
\begin{split}
& \text{If $B$ is any set term except $\varnothing$:} \\
& \{\{d \plus D\} | \F @ \P\} = B \fun \\
& (\Fdd \land \{\Pdd \plus \{D | \F @ \P\}\} = B)
 \lor (\lnot \nFdd \land \{D | \F @ \P\} = B)
\end{split}  \label{e:eset} \\[2mm]
\begin{split}
& \{\bar D | \F @ \P\} = \{t \plus A\} \fun \\
& \bar D = \{n \plus N\} \land \flt(n) \land t =_\mathcal{X} \ptt(n) \land \{N
| \F @ \P\} = A
\end{split}  \label{e:v=e} \\[2mm]
\begin{split}
& \{D | \F @ \P\} \neq A  \fun \\
   & (n \in \{D | \F @ \P\} \land n \notin A)
    \lor (n \notin \{D | \F @ \P\} \land n \in A)
\end{split}  \label{e:neq}
\end{gather}
\hrule \caption{\label{f:eq}Rewrite rules for $A = B$ and $R \neq B$; $A$ or
$B$ RIS terms}
\end{figure}

\subsubsection{Membership}

Rules dealing  with constraints of the form $t \in A$ and $t \notin A$, where
$t$ is a $\LX$ term and $A$ is a RIS term, are listed in Figure \ref{f:in}. The
case $t \notin \bar A$ where $\bar A$ is a variable-RIS is dealt with as
irreducible (Section \ref{solved}), while constraints of the form $t \in A$ are
eliminated in all cases.

\begin{figure}
 \hrule
\begin{gather}
t \in \{\e | \F @ \P\} \lfun \false \label{in:e}\\[2mm]
t \in \{ D | \F @ \P\} \lfun
  n \in  D \land \F(n) \land t =_\Ur \P(n) \label{in:V} \\[2mm]
%
%
t \notin \{\e | \F @ \P\} \lfun \true  \label{nin:e} \\[2mm]
\begin{split}
& t \notin \{\{d \plus D\} | \F @ \P\} \lfun \\
& (\Fdd \land t \neq_\Ur \Pdd \land t \notin \{D | \F @ \P\}) \lor (\lnot \nFdd
\land t \notin \{D | \F @ \P\})
\end{split} \label{nin:nV}
\end{gather}
\hrule \caption{\label{f:in}Rewrite rules for $t \in A$ and $t \notin A$; $A$
RIS term}
\end{figure}

\subsubsection{Union}

The rules for the union operator are given in Figure \ref{f:un}.

\begin{itemize}
\item Rules \eqref{un:empty}-\eqref{un:empty3} extend
similar rules dealing with extensional set terms
to the case where one argument is a RIS whose domain is
the empty set.
In these cases, $\Cup$-constraints are rewritten to equality constraints
and rules of Figure \ref{f:eq} are called into play.

\item Rules \eqref{un:ext1}-\eqref{un:ext} are the same used for extensional
set terms; the use of set unification, however, forces arguments which are RIS
terms to be possibly converted into extensional sets.

\item Finally, rule \eqref{un:ris} is added to deal with constraints where at
least one argument is a non variable-RIS and cannot be dealt with by the
previous rules. For example, rule \eqref{un:ext} cannot be used when the last
argument is a RIS, then rule \eqref{un:ris} is used instead.

In rule \eqref{un:ris}, each non variable-RIS is rewritten into either an
extensional set or a new RIS with one element less in its domain, plus the
constraints that assert or negate the filter. This is formalized by functions
$\svf$ and $\cvf$.
\end{itemize}

The cases not considered in Figure \ref{f:un}, that is, when the three arguments
are either variables or variables-RIS and the first two are not the same
variable, are dealt with as irreducible (Section \ref{solved}).

The rule for $\Ncup$ is the standard rule for extensional set terms (cf.
rule \eqref{e:nun}).

\begin{figure}
\hrule
\begin{gather}
\Cup(X,X,B) \rightarrow X = B \label{un:equalvars} \\[2mm]
\Cup(A,B,\varnothing) \rightarrow A = \emptyset \land B = \emptyset \label{un:empty} \\[2mm]
\Cup(\varnothing, A, B) \rightarrow B = A \label{un:empty2} \\[2mm]
\Cup(A, \varnothing, B) \rightarrow B = A \label{un:empty3} \\[2mm]
\begin{split}
\Cup&(\set{t}{C}, A, \bar{B}) \rightarrow \\
    & \set{t}{C} = \set{t}{N_1} \land t \notin N_1 \land \bar{B} = \set{t}{N} \\
    & \land (t \notin A \land \Cup(N_1, A, N)
             \lor A = \set{t}{N_2} \land t \notin N_2 \land \Cup(N_1, N_2, N))
\end{split} \label{un:ext1} \\[2mm]
%
\begin{split}
\Cup&(A, \set{t}{C}, \bar{B}) \rightarrow \\
    & \set{t}{C} = \set{t}{N_1} \land t \notin N_1 \land \bar{B}= \set{t}{N} \\
    & \land (t \notin A \land \Cup(N_1, A, N)
             \lor A = \set{t}{N_2} \land t \notin N_2 \land \Cup(N_1, N_2, N))
\end{split} \label{un:ext2} \\[2mm]
\begin{split}
\Cup&(A, B, \set{t}{C}) \rightarrow \\
    & \set{t}{C} = \set{t}{N} \\
    & \land (A = \set{t}{N_1} \land t \notin N_1 \land \Cup(N_1, B, N) \\
    & \qquad \lor B = \set{t}{N_1} \land t \notin N_1 \land \Cup(A, N_1, N) \\
    & \qquad \lor A = \set{t}{N_1} \land t \notin N_1
         \land B = \set{t}{N_2} \land t \notin N_2 \land \Cup(N_1, N_2, N))
\end{split} \label{un:ext} \\[2mm]
\text{If at least one of $A,B,C$ is not a variable nor a variable-RIS:} \notag \\
\Cup(A, B, C) \rightarrow
    \Cup(\svf(A), \svf(B), \svf(C))
    \land \cvf(A) \land \cvf(B) \land \cvf(C) \label{un:ris} \\
    \quad\text{where $\svf$ is a set-valued function
               and $\cvf$ is a constraint-valued function} \notag \\
\quad
\svf(\sigma) =
\begin{cases}
N_\sigma & \text{if $\sigma \equiv\defris{\set{d}{D}}$} \notag \\
\sigma & \text{otherwise}
\end{cases} \\[2mm]
\quad
\cvf(\sigma) =
\begin{cases}
N_\sigma = (\set{\Pdd}{\defris{D}} \land \Fdd)       & \text{if $\sigma \equiv\defris{\set{d}{D}}$} \notag \\
\quad{}\lor (N_\sigma = \defris{D} \land \lnot \Fdd) & \\
\true & \text{otherwise}
\end{cases}
\end{gather}
\hrule
\caption{\label{f:un}Rewrite rules for $\Cup(A,B,C)$; $A$, $B$ or $C$ a RIS term}
\end{figure}

\subsubsection{Disjointness}

The rules for the $\disj$ operator are listed in Figure \ref{f:disj}. As with
the other operators, some of the standard rules for extensional set terms are
extended to deal with RIS and new rules are added. In particular:

\begin{itemize}
\item Rule \eqref{disj:rempty} is simply extended to RIS whose domain is
the empty set.

\item Rule \eqref{disj:rris} is a new rule dealing with the case when the second
argument is a non variable-RIS. As can be seen, two cases are considered
depending on whether the filter holds or not for one of the elements of the
domain. In the first case, the rules for $\notin$ are called and then a
recursion is started; in the second case only the recursive call is made.
\end{itemize}

Note that $A \disj B$ is not further rewritten if $A$ and $B$ are distinct
variables or variable-RIS. The rule for $\Ndisj$ is the standard rule for
extensional set terms (see \cite{calculusBR}).

\begin{figure}
\hrule
\begin{gather}
X \disj X \rightarrow X = \e \label{disj:id} \\[2mm]
A \disj \varnothing \rightarrow \true \label{disj:rempty} \\[2mm]
\varnothing \disj A \rightarrow \true \label{disj:lempty} \\[2mm]
A \disj \set{t}{B} \rightarrow t \notin A \land A \disj B \label{disj:risext} \\[2mm]
\set{t}{B} \disj A \rightarrow t \notin A \land A \disj B \label{disj:extris} \\[2mm]
\begin{split}
A \disj & \riss{\set{d}{D}}{\Fpv}{\Ppv} \rightarrow \\
        & \Fdd \land \Pdd \notin A \land A \disj \riss{D}{\Fpv}{\Ppv}
          \lor \lnot\Fdd \land A \disj \riss{D}{\Fpv}{\Ppv}
\end{split} \label{disj:rris} \\[2mm]
\begin{split}
\{\{d &\plus D\} | \Fpv @ \Ppv\} \disj A \rightarrow \\
      & \Fdd \land \Pdd \notin A \land \riss{D}{\Fpv}{\Ppv} \disj A
        \lor \lnot \Fdd \land \riss{D}{\Fpv}{\Ppv} \disj A
\end{split} \label{disj:rrisSym}
\end{gather}
\hrule
\caption{\label{f:disj}Rewrite rules for $A \disj B$; $A$ or $B$ a RIS term}
\end{figure}

\subsection{Procedure \textsf{remove\_neq}}\label{removeneq}

The ${\cal RIS}$-formula returned by Algorithm \ref{glob} when \textsf{STEP}
reaches a fixpoint is not guaranteed to be satisfiable due to the presence of
inequalities involving set variables. For example, the formula $D \neq \e \land \{x:D | x=x\} =
\e$ is dealt with by $\SATRIS$ as irreducible but it is clearly
unsatisfiable. In order to guarantee satisfiability, all such inequalities must
be removed from the final formula returned by $\SATRIS$. This is performed (see
Algorithm \ref{glob}) by executing the procedure \textsf{remove\_neq}, which
applies the rewrite rule described by the generic rule scheme of Figure
\ref{fig:rules_neq_elim}. Basically, this rule exploits set extensionality to
state that non-equal sets can be distinguished by asserting that a fresh element
($n$) belongs to one but not to the other.

\begin{figure}
\hrule \vspace{8pt}
 If $A \in \Var_{S}$, $t : \sU$ and $\Phi$ is the input formula then:
\begin{equation}
\begin{split}
&\text{If $A$ is an argument of a $\Cup$-constraint or the domain of a variable-RIS}\\
&\text{occurring as the argument of either a $=$ or $\notin$ or $\Cup$ constraint in $\Phi$:}\\
& A \neq t \rightarrow (n \in A \land n \notin t) \lor (n \in t \land n \notin
A)
 \label{rule:neq_elim}
\end{split} 
\end{equation}
 \hrule
 \caption{Rule scheme for $\neq$-constraint elimination rules}
\label{fig:rules_neq_elim}
\end{figure}

\begin{example}
Given the formula $D \neq \e \land \{x:D | \flt @ \ptt\} = \e$,
\textsf{remove\_neq} rewrites $D \neq \e$ as $n \in D$, where $n$ is a fresh
variable. In turn, $n \in D$ is rewritten as $D = \{n \plus N\}$ for another
fresh variable $N$. Finally, the whole formula is rewritten as $D = \{n \plus
N\} \land \{x:\{n \plus N\} | \flt @ \ptt\} = \e$, which fires one of the rules
given in Section \ref{rules}. This rewriting chain is fired only because $D$ is
the domain of a RIS; otherwise \textsf{remove\_neq} does nothing with $D \neq
\e$.
 \qed
\end{example}


\section{Decidability of \CLPRIS Formulas}\label{deci}

In this section we show that $\SATRIS$ can serve as a decision procedure for a
large fragment of $\LRIS$. To this end, first we precisely characterize the
admissible $\LRIS$ formulas involving RIS; then we prove that: \emph{(a)} when
$\SATRIS$ terminates the answer is either $\false$ or a disjunction of $\LRIS$
formulas of a very particular form; \emph{(b)} that disjunction is always
satisfiable, and a solution can be trivially found; \emph{(c)} the set of
solutions of the computed formula is the same of the input formula; and
\emph{(d)} $\SATRIS$ terminates for every admissible formula. Detailed proofs
are given in Appendix \ref{ap:proofs}.

\subsection{Admissible formulas} \label{ssec:admissible}

RIS are intended to denote finite sets. For this reason, RIS have the
restriction that the RIS domain must be a finite---though unbounded---set.
However, this is not enough to guarantee RIS finiteness.

\begin{example}\label{ex:infinite_set}
The $\LRIS$ formula
\begin{equation}\label{eq:infinite_set}
A = \{1 \plus \{x : A | \true @ (x,y)\}
\end{equation}
admits only the solution binding $A$ to the infinite set $\{1,(1,y),((1,y),y),$
$(((1,y),y),y),\dots\}$. If $\SATRIS$ is requested to solve this formula, it
will loop forever trying to build the extensional set term representing this
infinite set.
\qed
\end{example}

In the rest of this section, we will analyze this kind of ``cyclic''
definitions, introducing a few restrictions on $\mathcal{RIS}$-formulas that
prevent them from defining infinite sets.

First of all, observe that we have assumed that the set $\Var_\mathcal{S}$ of
variables ranging over $\mathcal{RIS}$-terms (i.e., set variables) and the set
$\Var_\Ur$ of variables ranging over $\Ur$-terms are disjoint sets. Since RIS
filters are $\Ur$-formulas, then set variables cannot occur in them. This fact
prevents us from creating \emph{recursively defined} RIS of the form $S = \{\ct
: D | \flt(S) @ \ptt\}$.  Moreover, since RIS patterns, as well as elements of
RIS domains, are $\Ur$-terms, then no cyclic definitions of the form  $S =
\{\ct : D | \flt @ \ptt(S)\}$ or $S = \{\ct : \{t_1,\dots,S,\dots,t_n \plus A\}
| \flt @ \ptt\}$ can be written in $\LRIS$.

Hence, the only possible way to create a cyclic definition involving a RIS is
through the set variable possibly occurring as the set part of the RIS domain,
e.g., $S = \{\ct : \{t_1,\dots,t_n \plus S\} | \flt @ \ptt\}$. When the pattern
is the same as the control term, an equation such as $S = \{\ct :
\{t_1,\dots,t_n \plus S\} | \flt\}$ admits the trivial solution in which $S$ is
bound to the (finite) set containing all $t_i$, $i=1,\dots,n$, such that
$\flt(t_i)$ holds. For example,  $S = \{x : \{-1,0,1 \plus S\} | x \neq 0\}$
has the solution $S = \{-1,1 \plus N\}$, where $N$ is a fresh variable. This is
no longer true when the pattern is not the control term, for instance as in
formula \eqref{eq:infinite_set}. In other words, when the pattern is the
control term, then the RIS is a subset of the domain, while in all other cases
this relation does not hold.

As a consequence, processing a formula such as \eqref{eq:infinite_set} may
cause $\SATRIS$ to enter into an infinite loop. In order to allow $\SATRIS$ to
be a decision procedure for $\LRIS$, formulas such as \eqref{eq:infinite_set}
should be discarded.

In what follows, we provide sufficient (syntactic) conditions characterizing a
sub-language of $\LRIS$ for which $\SATRIS$ will be
proved to terminate and so to be a decision procedure for
that sub-language.

First, we define a transformation $\tau$ of $\mathcal{RIS}$-formulas that
allows us to restrict our attention to a single kind of constraints. Let $\Phi
= \Phi_\mathcal{S} \land \Phi_\Ur$ be the input formula, where
$\Phi_\mathcal{S}$ is a pure $\mathcal{RIS}$-formula and $\Phi_\Ur$ is a
$\Ur$-formula; $\Phi_\Ur$ is removed from $\Phi$ and so we only consider its
pure RIS part. Without loss of generality, $\Phi_\mathcal{S}$ can be seen as
$\Phi_1 \lor \dots \lor \Phi_n$, where the $\Phi_i$'s are conjunctions of
primitive $\mathcal{RIS}$-constraints (i.e., all derived constraints have been
replaced by their definitions and the corresponding DNF has been built). Then,
each $\Phi_i$ is transformed into $\Phi_i'$ as follows:
\begin{itemize}
\item constraints of the form $\Cup(A,B,C)$, where $A,B,C$ are either variables
or variable-RIS whose innermost domain variables are $D_A$, $D_B$, $D_C$, and
none of these domain variables occur elsewhere in $\Phi_i$, are removed from
the formula
\item constraints of the form $A = B$, where neither is $\varnothing$, are
rewritten into $\Cup(A,B,B) \land \Cup(B,A,A)$
\item If one of the arguments of a $\Cup$ constraint is of the form
$\{x_1,\dots, x_n \plus B\}$ then it is replaced by a new variable, $N$, and
$\Cup(\{x_1,\dots, x_n\},B,N)$ is conjoined to the formula.
\item all the $\neq$, $\in$, $\notin$, $\disj$, and $\Set$ constraints,
and all the remaining $=$ constraints, are removed from the formula.
\end{itemize}
Hence, $\tau(\Phi) = \Phi_1' \lor \dots \lor \Phi_n'$, where each $\Phi_i'$ is
a conjunction of $\Cup$-constraints.

The following function allows to classify set terms occurring as arguments of
$\Cup$-constraints.

\begin{definition}
Let $\type$ be a function that takes a set term $T$ and returns an element in
$\{\stype,\ptype,\utype\}$, where $\ptype$ depends on one argument belonging to
$\{\stype,\ptype,\utype\}$ and $\utype$ depends on two arguments belonging to
$\{\stype,\ptype,\utype\}$. For each constraint of the form $\Cup(A,B,C)$ the
function $\type(T)$ is defined as follows (note that the definition of $\type$
depends on the position of the argument in the constraint):
\begin{enumerate}
\item If $T$ is $C$, then:
$\type(C) = \utype(\type(A),\type(B))$
\item If $T$ is either $A$ or $B$, then:
\begin{gather}
\type(\e) = \stype  \\
\type(\{\cdot \plus V\}) = \type(V) \\
\type(\risnopattern{c}{D}{F}) = \type(D)  \\
\type(\riss{c : D}{F}{P}) = \ptype(\type(D)), c \not\equiv P
\end{gather}
and $\type(T)$ remains undefined when $T$ is a variable.
\end{enumerate} \qed
\end{definition}

\begin{example}\label{ex:type_function}
If $\Phi_i$ is $\riss{x : D}{F}{(x,y)} \subseteq D$, then $\Phi_i'$ is
\[
\Cup(\riss{x : D}{F}{(x,y)},D,D)
\]
and the $\type$ function for this constraint is:
\begin{gather*}
\type(D) = \utype(\type(D),\type(\riss{x : D}{F}{(x,y)})) \\
\iff \type(D) = \utype(\type(D),\ptype(\type(D)))
\end{gather*}
where the computation of $\type$ stops because $D$ is a variable. \qed
\end{example}

\begin{definition}
$\ptype^*(D)$ denotes a $\ptype$ that at some point depends on variable $D$.
\qed
\end{definition}

\begin{definition}[Admissible $\mathcal{RIS}$-formula]\label{d:admissible_formula}
Let $\Phi$ be a $\mathcal{RIS}$-formula not in solved form and $\mathcal{E}$ be
the collection of equalities computed by (recursively) applying the $\type$
function to all the $\Cup$-constraints in $\tau(\Phi)$ and performing all
possible term substitutions. Then $\Phi$ is \emph{non-admissible} iff
$\mathcal{E}$ contains at least one equality of the form $X = \utype(Y,Z)$ such
that:
\begin{itemize}
\item If $X$ depends on $\ptype^*(D)$, for some variable $D$, then $Y$ or $Z$
does not depend on $\ptype^*(D)$; and
\item If $Y$ or $Z$ depend on $\ptype^*(D)$, for some variable $D$, then $X$
does not depend on $\ptype^*(D)$.
\end{itemize}
All other $\mathcal{RIS}$-formulas are \emph{admissible}. \qed
\end{definition}

\begin{example}
Given the formula $\Phi$
\[
\riss{x : D}{F}{(x,y)} \subseteq D \land D \neq \e
\]
then $\tau(\Phi)$ is
\[
\Cup(\riss{x : D}{F}{(x,y)},D,D)
\]
that is exactly the formula of Example \ref{ex:type_function}. So $\Phi$ is
classified as non-admissible. Conversely, if $\Phi$ is just $\riss{x :
D}{F}{(x,y)} \subseteq D$, that is, $D$ is a variable not occurring elsewhere
in $\Phi$, then $\mathcal{E}$ is empty and $\Phi$ is classified as admissible.
\qed
\end{example}

\begin{example}
Given the formula $\Phi$
\[
\riss{x : D}{F}{(x,y)} \subseteq \risnopattern{x}{A}{G} \land A \subseteq D
\land D \neq \e
\]
$\tau(\Phi)$ is
\[
\Cup(\riss{x : D}{F}{(x,y)},\risnopattern{x}{A}{G},\risnopattern{x}{A}{G})
\land \Cup(A,D,D)
\]
Then, the collection $\mathcal{E}$ for $\Phi$ is
\begin{gather*}
\{\type(\risnopattern{x}{A}{G}) = \utype(\type(\riss{x :
D}{F}{(x,y)}),\type(\risnopattern{x}{A}{G})),
\type(D) = \utype(\type(A),\type(D))\} \\
\iff \\
\{\type(A) = \utype(\ptype(\type(D)),\type(A)),
\type(D) = \utype(\type(A),\type(D))\} \\
\iff \why{substitution} \\
\{\type(A) = \utype(\ptype(\utype(\type(A),\type(D))),\type(A)), \type(D) =
\utype(\type(A),\type(D))\}
\end{gather*}
So $\Phi$ is classified as non-admissible. Given a formula similar to the
previous one, but where the second RIS is a set of pairs, i.e.,
\[
\riss{x : D}{F}{(x,y)} \subseteq \riss{h : A}{G}{(h,w)} \land A \subseteq D
\land D \neq \e
\]
then the final collection $\mathcal{E}$ for this formula is
\[ \{\ptype(\type(A))
= \utype(\ptype(\utype(\type(A),\type(D))),\ptype(\type(A))), \type(D) =
\utype(\type(A),\type(D))\}
\]
so it is classified as admissible. \qed
\end{example}

From the above
definitions, it is evident that, if the given formula $\Phi$ does not contain
any RIS term, or if all RIS terms possibly occurring in it have pattern
identical to its control term, then $\Phi$ is surely classified as admissible.

Intuitively, non-admissible formulas are those where a variable $A$ is the
domain of a RIS representing a function and, at the same, time $A$ is either a
sub or a superset of that function. For example, the
non-admissible formula $\riss{x : D}{\true}{(x,y)} \subseteq D \land D \neq
\e$ implies that if $a \in D$ then $(a,n) \in D$ and
so $((a,n),n_1) \in D$ and so forth, thus generating an infinite $\Ur$-term. In
other words, non-admissible formulas, in a way or another, require to compare a
function and its domain through the subset relation. Note, however, that $\mathcal{RIS}$-formulas that are classified
as non-admissible are, in a sense, quite unusual formulas. Clearly, in most
typed formalisms, such formulas would not type-check (e.g., in B and Z
\cite{Spivey00,Abrial00}). Therefore, missing them from the decision procedure
implemented by $\SATRIS$ does not seem to be a real lack.

\subsection{\label{solved}Satisfiability of solved form}

As stated in the previous section, the formula $\Phi$ handled by $\SATRIS$ can
be written as $\Phi_\mathcal{S} \land \Phi_\Ur$ where $\Phi_\mathcal{S}$ is a
pure $\mathcal{RIS}$-formula and $\Phi_\Ur$ is a $\Ur$-formula. Right before
Algorithm \ref{glob} calls $\SATX$, $\Phi_\mathcal{S}$ is in a particular form
referred to as \emph{solved form}. This fact can be easily proved by inspecting
the rewrite rules given in Section \ref{rules}. The constraints in solved form
are selected as to allow trivial verification of satisfiability of the formula
as a whole.

\begin{definition}[Solved form]\label{def:solved}
Let $\Phi_\mathcal{S}$ be a pure admissible $\mathcal{RIS}$-formula; let
$\bar C$, $\bar D$ and $\bar E$ be either variables of sort $\sSet$ or
variable-RIS, $X$ and $Y$ be variables of sort $\sSet$ but not variable-RIS,
and $x$ a variable of sort $\sU$; let $t$ be any term of sort $\sU$, and $S$
any term of sort $\sSet$ but not a RIS. An atom $p$ in $\Phi_\mathcal{S}$ is in
\emph{solved form} if it has one of the following forms:
\begin{enumerate}
\item $\true$

\item\label{i:icfirst} $X = S$ or $X = \{Y | \flt @ \ptt\}$,
and $X$ does not occur in $S$ nor in $\Phi_\mathcal{S} \setminus \{p\}$

\item\label{i:solvedRIS1} $\{X | \flt @ \ptt\} = \varnothing$ or
$\varnothing = \{X | \flt @ \ptt\}$

\item\label{i:solvedRIS2} $\{X | \flt_1 @ \ptt_1\} = \{Y | \flt_2 @ \ptt_2\}$.

\item \label{i:solvedneq} $X \neq S$, and $X$ does not occur in $S$ nor as
the domain of a RIS which is the argument of a $=$ or $\notin$ or $\Cup$
constraint in $\Phi_\mathcal{S}$ \footnote{This is guaranteed by procedure
\textsf{remove\_neq} (see Section \ref{solver}).}

\item \label{i:solvednin} $t \notin \bar D$

\item $\Cup(\bar C,\bar D,\bar E)$, and if $\bar C, \bar D \in \Var_{S}$
then $\bar C \not\equiv \bar D$

\item $\bar C \disj \bar D$, and if $\bar C,\bar D \in \Var_{S}$
then $\bar C \not\equiv \bar D$

\item $\Set(X)$, and no constraint $\isx(X)$ is in $\Phi_\mathcal{S}$

\item $\isx(X)$
\end{enumerate}
$\Phi_\mathcal{S}$ is in solved form if all its atoms
are in solved form. \qed
\end{definition}

\begin{example}
The following are $\LRIS$ atoms in solved form,
occurring in a formula $\Phi$ (where $X$, $D$ and $D_i$ are variables):
\begin{itemize}
\item $X = \{x : D | x \neq 0\}$ and $X$ does not occur
elsewhere in $\Phi$ (note that $X$ and $D$ can be the same variable)
\item $1 \notin \{x : D | x \neq 0\}$
\item $\{x : D_1 | x \mod 2 = 0 @ (x,x)\} = \{x : D_2 | x > 0 @ (x,x+2)\}$
\item $\Cup(X,\riss{D_1}{F}{P},\riss{D_2}{G}{Q})$ and, for any $t$,
there are no constraints $D_1 \neq t$ nor $D_2 \neq t$ in $\Phi$.
\qed
\end{itemize}
\end{example}

Right before Algorithm \ref{glob} calls $\SATX$, $\Phi_\mathcal{S}$ is either
$\false$ or it is in solved form, and in this case it is
satisfiable.

\begin{theorem}[Satisfiability of solved form]\label{satisf}
Any (pure) $\mathcal{RIS}$-formula in solved form is satisfiable w.r.t. the
interpretation structure of $\LRIS$.
\end{theorem}
\begin{proof}[sketch]
Basically, the proof of this theorem uses the fact that, given a (pure)
$\mathcal{RIS}$-formula $\Phi$ in solved form, it is possible to guarantee the
existence of a successful assignment of values to all variables of $\Phi$ using
pure sets only (in particular, the empty set for set variables), with the
exception of the variables $X$ occurring in atoms of the form $X = t$. \qed
\end{proof}

Moreover, if $\Phi_\mathcal{S}$ is not $\false$, the satisfiability of $\Phi$
depends only on $\Phi_\Ur$.

\begin{theorem}[Satisfiability of $\Phi_\mathcal{S} \land \Phi_\Ur$] \label{satisf2}
Let $\Phi$ be $\Phi_\mathcal{S} \land \Phi_\Ur$ right before Algorithm
\ref{glob} calls $\SATX$. Then either $\Phi_\mathcal{S}$ is $\false$ or the
satisfiability of $\Phi$ depends only on the satisfiability of $\Phi_\Ur$.
\end{theorem}
\begin{proof}[sketch]
The proof is based on the observation that assigning values to
variables of sort $\Ur$ possibly occurring in $\Phi_\mathcal{S}$ does not
affect the solved form of $\Phi_\mathcal{S}$, i.e., $\Phi_\mathcal{S}$ remains
in solved form, hence satisfiable, disregarding of the values assigned to the
$\Ur$ variables.
\qed
\end{proof}

\subsection{Equisatisfiability}

In order to prove that Algorithm \ref{glob} is correct and complete,
we prove that it preserves the set of solutions of the
input formula.

\begin{theorem}[Equisatisfiability]\label{equisatisfiable}
Let $\Phi$ be an admissible $\mathcal{RIS}$-formula and
$\{\Phi_i\}_{i=1}^{n}$ be the collection of $\mathcal{RIS}$-formulas returned
by $\SATRIS(\Phi)$. Then $\bigvee_{i=1}^{n} \Phi_i$ is equisatisfiable w.r.t.
$\Phi$, that is, every possible solution\footnote{More precisely, each solution
of $\Phi$ expanded to the variables occurring in $\Phi_i$ but not in $\Phi$, so
to account for the possible fresh variables introduced into $\Phi_i$.} of
$\Phi$ is a solution of one of $\{\Phi_i\}_{i=1}^{n}$ and, vice versa, every
solution of one of these formulas is a solution for $\Phi$.
\end{theorem}

\begin{proof}[sketch]
The proof rests on a series of lemmas each showing that the set of solutions of
left and right-hand sides of each rewrite rule is the same. For all cases
dealing with extensional set terms the proofs can be found in \cite{Dovier00}.
The proofs of equisatisfiability for the rules shown in Section \ref{solver}
and in Appendix \ref{ap:rules} can be found in Appendix \ref{ap:proofs}.  \qed
\end{proof}

Observe that Theorems \ref{satisf} and \ref{equisatisfiable} imply that
$\SATRIS$, not only determines whether the input formula
is satisfiable or not but also, when the input formula is satisfiable, it
computes a finite representation of all the (possibly infinitely many)
solutions.

\subsection{\label{equi}Termination}

Termination of $\SATRIS$ can be proved for admissible formulas as is stated by
the following theorem.

\begin{theorem}[Termination]\label{termination-glob}
The $\SATRIS$ procedure can be implemented in such a way it terminates for
every input admissible $\mathcal{RIS}$-formula $\Phi$.
\end{theorem}

The termination of $\SATRIS$ and the finiteness of the number of
non-determi\-nistic choices generated during its computation guarantee the
finiteness of the number of $\mathcal{RIS}$-formulas non-deterministically
returned by $\SATRIS$. Therefore, $\SATRIS$\ applied to an admissible
$\mathcal{RIS}$-formula $\Phi$ always terminates, returning either $\false$ or
a finite collection of satisfiable $\mathcal{RIS}$-formulas in solved form.

Thanks to Theorems \ref{satisf}, \ref{equisatisfiable} and
\ref{termination-glob} we can conclude that, given an admissible
$\mathcal{RIS}$-formula $\Phi$, $\Phi$ is satisfiable with respect to the
intended interpretation structure $\iS$ if and only if there is a
non-deterministic choice in $\SATRIS(\Phi)$ that returns a
$\mathcal{RIS}$-formula $\Phi_\mathcal{S} \land \Phi_\Ur$ where
$\Phi_\mathcal{S}$ is in solved form---i.e., there is a choice that does not
terminate with $\false$.
Hence, $\SATRIS$ is a decision procedure for testing satisfiability of
admissible $\mathcal{RIS}$-formulas.

\begin{remark}
Definition \ref{d:admissible_formula} gives only \emph{sufficient} conditions.
In fact, not all formulas classified as non-admissible are indeed formulas that
our solver cannot deal with. Given the formula:
\begin{equation}\label{eq:1}
\riss{x : A}{\false}{(x,y)} \subseteq A \land Y \in A
\end{equation}
any set $A$ satisfying $Y \in A$ is a solution of it. So, we should accept it.
However, according to Definition \ref{d:admissible_formula}, this formula is
classified as non-admissible. Note that the similar formula where the filter is
$\true$ admits only an infinite set solution and is (correctly) classified as
non-admissible.

Accepting or not formula \eqref{eq:1} depends on the satisfiability of the RIS
filter. Checking the satisfiability of the filter, however, cannot be done, in
general, by simple syntactic analysis, i.e., without running the solver on it.
Thus, when aiming at providing a syntactic characterization of admissible
formulas, we must classify formulas disregarding the form of the RIS filters
possibly occurring in them. Finer characterizations would be feasible,
however, considering special forms of RIS filters, such as $\false$ and
$\true$.
\qed
\end{remark}

\begin{remark}
In practice many interesting theories are undecidable and only semi-decision
procedures exist for them. On the other hand, in fact, all of the theoretical
results presented so far apply provided RIS' filters belong  to a decidable
\emph{fragment} of the considered parameter theory, and a solver for this
fragment is called from $\SATRIS$. Hence, the condition on the availability of
a decision procedure for \emph{all} $\Ur$-formulas can be often relaxed.
Instead, the existence of some algorithm capable of deciding the satisfiability
of a significant fragment of $\Ur$-formulas can be assumed. If such an
algorithm exists and the user writes formulas inside the corresponding
fragment, all of our results still apply. For these reason, in coming sections,
we sometimes give examples where $\Ur$ is not necessarily a decidable theory.
\end{remark}

\subsection{Complexity}

$\SATRIS$ strongly relies on set unification. Basically, rewrite rules dealing
with RIS ``extract'' one element at a time from the domain of a RIS by means of
set unification and construct the corresponding extensional set again through
set unification. Hence, complexity of our decision procedure strongly depends
on complexity of set unification. As observed in \cite{Dovier03}, the decision
problem for set unification is NP-complete. A simple proof of the NP-hardness
of this problem has been given in \cite{DBLP:journals/jlp/DovierOPR96}. The
proof is based on a reduction of 3-SAT to a set unification problem. Concerning
NP-completeness, the algorithm presented here clearly does not belong to NP
since it applies syntactic substitutions. Nevertheless, it is possible to
encode this algorithm using well-known techniques that avoid explicit
substitutions, maintaining a polynomial time complexity along each
non-deterministic branch of the computation.

Besides, $\SATRIS$ deals not only
with the decision problem for set unification but also with the associated
function problem (i.e., it can compute solutions for the problem at hand).
Since solving the function problem clearly implies solving the related decision
problem, the complexity of $\SATRIS$ can be no better than the complexity of
the decision problem for set unification. Finally, since $\SATRIS$ is
parametric w.r.t. $\SATX$, its complexity is at least the maximum between the
complexity of both.

\section{Encoding restricted universal quantifiers and partial functions} \label{sec:uses}

In this section we show how RIS can be used to encode restricted universal
quantifiers (Section \ref{ruq}) and partial functions (Section \ref{pf}).

\subsection{Restricted universal quantifiers}\label{ruq}

RIS can be used to encode restricted universal quantifiers. That is, if $A$ is
a finite set and $\Fpv$ is a formula of some theory $\Ur$, the formula:
\begin{equation}
\forall x \in A: \Fpv(x)
\end{equation}
is equivalent to the $\LRIS$  formula (see Proposition
\ref{prop:1} in the appendix):
\begin{equation}
A \subseteq \{x:A | \Fpv(x)\}
\end{equation}
where $\subseteq$ is a derived constraint based on $\Cup$ (see Section
\ref{expressiveness}).

This formula lies inside of the decision procedure of $\LRIS$. Therefore,
$\LRIS$ is able to reason about universally quantified formulas.

The following example illustrates a possible use of RIS to express a restricted
universal quantification.

\begin{example}\label{ex:quantif}
The formula $y \in S \land S \subseteq \{x:S | y \leq x\}$ states that $y$ is
the minimum of a set of integers $S$. If, for instance, $S = \{2,4,1,6\}$, then
$y$ is bound to $1$. The same RIS can be used in conjunction with a
partially specified $S$ to properly constraint its unknown elements. For
instance, if $S = \{2,4,z,6\}$, then the $\LRIS$ solver returns either $y = 2
\land z \geq 2$ or $y = z \land z \leq 2$, provided the underlying $\LX$ solver
is able to manage simple integer disequations. Similarly, if $R$ is a variable
and $S = \{1 \plus R\}$, then the first answer returned by $\LRIS$ will be $y =
1 \land R \subseteq \{x:R | 1 \leq x\}$, where $R$ represents the set of all
integer numbers less or equal to $1$.
\qed
\end{example}

Restricted universal quantifiers can be exploited, for instance, in program
verification. As a matter of fact, many programs operate on finite, unbounded
inputs and states. Hence, if the formal specification of such a program uses an
universally quantified formula, indeed it can be replaced with a restricted
universally quantified formula. This implies that if theory $\Ur$ is decidable,
then the formal verification of those programs that can be specified as
$\LRIS(\Ur)$ formulas, is decidable as well.

For the sake of convenience, we introduce the following derived constraint:
\begin{equation}
\Forall(x \in A,\Fpv) \defs A \subseteq \{x:A | \Fpv\}
\defs
  \Cup(A,\{x:A | \Fpv\},\{x:A | \Fpv\})
\end{equation}
Hereafter, we will use $\Forall$ in place of the equivalent
$\subseteq$-constraint.

\begin{example}[Linear integer algebra]
It is a known fact that linear algebra over $\num$ is decidible; let us call
this theory $\mathcal{Z}$. Then, any admissible $\LRIS$-formula including the
$\Forall$-constraints listed in Table \ref{t:forall} fall inside the decision
procedure of the restricted quantified theory $\LRIS(\mathcal{Z})$.
\end{example}

\begin{figure}
\begin{tabularx}{\textwidth}{Xl}
\toprule \\
All the even numbers in $A$ & $\Forall(x \in A,x \mod 2 = 0)$ \\[1mm]
Intersection of $B$ with the less-than relation & $\Forall((x,y) \in B,x \leq y)$ \\[1mm]
All the solutions of an equation with a free variable in
$C$ & $\Forall(y \in C,3x + 2y - 14 = 0)$ \\[1mm]
Intersection of $D$ with the successor function & $\Forall((x,y) \in D,y = x + 1)$ \\
\bottomrule
\end{tabularx}
\caption{\label{t:forall} Examples of decidable restricted quantified formulas
over linear integer algebra}
\end{figure}

\begin{remark}
Although the rules for $\Cup$ shown in Figure \ref{f:un} can deal with the
$\Forall$ predicate, we introduce a set of specialized rewrite rules to process
this specific kind of predicates more efficiently (see Figure \ref{f:forall}).
The key rule is \eqref{forall:iter} because it corresponds to the formula
$\Forall(x \in \set{t}{A},\Fpv(x))$, where $x$ is a variable. In turn, this
formula can be seen as an iterative program whose iteration variable is $x$,
the range of iteration is $\set{t}{A}$, and the body is $\Fpv$. In fact, rule
\eqref{forall:iter} basically iterates over $\set{t}{A}$ and evaluates $\Fpv$
for each element in that set. If one of these elements does not satisfy $\Fpv$
then the loop terminates immediately, otherwise it continues until the empty
set is found (i.e., rule \eqref{forall:empty}) or a variable is found (i.e.,
rule \eqref{forall:var}).

\begin{figure}
\hrule
\begin{gather}
\Cup(\varnothing,
     \risnocp{A}{\Fpv},\risnocp{A}{\Fpv})
 \rightarrow \true \label{forall:empty} \\[2mm]
\Cup(\bar{A},\risnocp{\bar{A}}{\Fpv},\risnocp{\bar{A}}{\Fpv})
  \rightarrow \text{irreducible} \label{forall:var} \\[2mm]
\Cup(\set{t}{A},
      \risnocp{\set{t}{A}}{\Fpv},\risnocp{\set{t}{A}}{\Fpv})
  \rightarrow
    \Fpv(t)
     \land \Cup(A,
                \risnocp{A}{\Fpv},
                \risnocp{A}{\Fpv})  \label{forall:iter}
\end{gather}
\hrule
\caption{\label{f:forall}Specialized rewrite rules to deal with
restricted universal quantifiers}
\end{figure}

\end{remark}

\subsection{Partial functions}\label{pf}

Another important use of RIS is to define \emph{(partial) functions} by giving
their domains and the expressions that define them. In general, a RIS of the
form
\begin{equation}
\{x:D | F @ (x,f(x))\}
\end{equation}
where $f$ is any $\LX$ function symbol, defines a partial function. Such a RIS
contains ordered pairs whose first components belong to $D$ which cannot have
duplicates (because it is a set). Then, if no two pairs share the same first
component, then the RIS is a function. Given that RIS are sets, then, in
$\LRIS$, functions are sets of ordered pairs as in Z and B. Therefore, through
standard set operators, functions can be evaluated, compared and point-wise
composed; and by means of constraint solving, the inverse of a function can
also be computed. The following examples illustrate these properties.

\begin{example}
The square of 5 can be calculated by: $(5,y) \in \{x:D @ (x,x*x)\}$, yielding
$y = 25$. The same RIS calculates the square root of a given number: $(x,36)
\in \{x:D @ (x,x*x)\}$, returning $x = 6$ and $x = -6$. Set membership can also
be used for the point-wise composition of functions. The function $f(x) = x^2 +
8$ can be evaluated on 5 as follows: $(5,y) \in \{x:D @ (x,x*x)\} \land (y,z)
\in \{e:E @ (e,e + 8)\}$ returning $y = 25$ and $z = 33$. \qed
\end{example}

\section{Extensions}\label{sec:extension}

The formula $\Phi$ of a (general) intensional set $\{x : \Phi(x)\}$ may depend
on existentially quantified variables, declared inside the set. For example, if
$R$ is a set of ordered pairs and $D$ is a set, then the subset of $R$ where all
the first components belong to $D$ can be denoted by
\begin{equation}\label{IS}
\{p : \exists x, y(x \in D \land (x,y) \in R \land p = (x,y))\}.
\end{equation}
We will refer to these existentially quantified variables as
\emph{parameters}.

Allowing parameters in RIS rises major problems when RIS have to be manipulated
through the rewrite rules considered in the previous section. In fact, if
$\vec{p}$ is the vector of parameters possibly occurring in a RIS, then literals
of the form $\lnot \nFdd$, occurring in the rules (e.g., \eqref{e:eset}), should
be replaced with the more complex universally quantified formula $\forall
\vec{p}(\lnot \flt(d,\vec{p}))$. This, in turn, would require that the theory
$\Ur$ is equipped with a solver able to deal with such kind of formulas.

In this section we describe two extensions to $\LRIS$ that increase its
expressiveness concerning the presence of parameters, without compromising
decidability. The first one deals with control terms (and patterns), while the
second one deals with special kinds of predicates occurring in RIS filters,
called \emph{functional predicates}.

Both extensions are supported by the implementation of $\LRIS$ within \setlog
(see Section \ref{impl}).

\subsection{Control terms and patterns}\label{sec:controlPatterns}

The first of these extensions is the possibility to use more general forms of
control terms and patterns in RIS.

As stated in Definition \ref{RIS-terms}, $\LRIS$ is
designed to allow for terms of the form $(x,y)$ in place of the quantified
control variable of set comprehensions. In effect, the
idea of allowing control terms (not just control variables) steams from
the observation that many uses of parameters can be avoided by a proper use of
the control term of a RIS.

\begin{proposition}\label{prop:3}
If $\vp$ is a vector of existentially quantified variables declared inside
an intensional set then:
\begin{gather*}
S = \{x:D_1 | \Fpv(x,\vp) \land \vp \in D_2 @ \Ppv(x,\vp)\} \\
\iff S = \{(x,\vp):D_1 \times D_2 | \Fpv((x,\vp)) @ \Ppv((x,\vp))\}
\end{gather*}
\end{proposition}

This result can be applied to the example mentioned above.

\begin{example}
The intensional set \eqref{IS} can be expressed with a
RIS (hence, without parameters) as follows: $\{(x,y):R | x \in D\}$. If $R$ is
for instance $\{(a,1),(b,2),(a,2)\}$ and $D$ is $\{a\}$, then the formula
$\{(x,y):R | x \in D\} = \{(a,1),(a,2)\}$ is (correctly) found to be
satisfiable by $\SATRIS$.
\end{example}

Therefore, it would be interesting to extend RIS to allow more general forms of
control terms without loosing completeness of the solver. In this respect, it
is important to note that the proof of Theorem \ref{equisatisfiable} relies on
a particular relation between the control term and the pattern of a RIS and a
property that patterns allowed in the language must verify. In order to state
these conditions we need the following definitions.

\begin{definition}[Bijective pattern]
Let $\{x:D | \flt(x) @ \ptt(x)\}$ be a RIS, then its pattern is \emph{bijective}
if $\ptt: \{x : x \in D \land \flt(x)\} \fun Y$ is a bijective function, where
$Y$ is the set of images of $\ptt$.
\end{definition}

\begin{definition}[Co-injective patterns]
Two patterns, $\P$ and $\Q$, are said to be \emph{co-injective} if for any $x$
and $y$, if $\P(x) = \Q(y)$ then $x = y$.
\end{definition}

Then, Theorem \ref{equisatisfiable} can be proved provided all patterns are
bijective and pairwise co-injective.

Through Definition \ref{RIS-terms}, $\LRIS$ is designed to guarantee that
patterns verify those conditions. In particular, when all the patterns are the
corresponding control terms, they verify the above conditions as they are
trivially bijective and pairwise co-injective. This is important because this
subclass of RIS is used to encode restricted universal quantifiers (see Section
\ref{ruq}).

The bijectivity of a pattern depends on the form of the control term. For
example, if the control term is $(x,y)$ and the pattern depends only on one
between $x$ and $y$, then it cannot be bijective because for a given $x$ or $y$
there are many $(x,y)$. Conversely, if $x$ is the control term, then any
pattern of the form $(x,\cdot)$ is bijective.

Besides the terms considered in Definition \ref{RIS-terms}, however, others can
be bijective patterns, although they might not be pairwise co-injective.

\begin{example}
\begin{itemize}
\item If $c$ is any control term then a pattern of the form
$(\cdot,c)$ is bijective. Further, these patterns are also pairwise
co-injective. However, if the language allows these patterns in
conjunction with patterns of the form $(c,\cdot)$, then pairwise
co-injectivity is lost.
\item If $x$ is the control term and $n$ is a constant, then $x+n$ is a
bijective pattern. However, these patterns are not always pairwise
co-injective.

\item If $x$ is the control term, then $x*x$ is not bijective as $x$ and
$-x$ have $x*x$ as image, while $(x,x*x)$ is indeed a bijective pattern,
allowed in $\LRIS$. \qed
\end{itemize}
\end{example}

The intuitive reason to ask for bijective patterns is that if $y$ belongs to a
RIS whose pattern, $\ptt$, is not bijective then there may be two or more
elements in the RIS domain, say $x_1$ and $x_2$, such that $\ptt(x_1) =
\ptt(x_2) = y$. If this is the case, then eliminating, say, $x_1$ from the
domain is not enough to eliminate $y$ from the RIS. And this makes it
difficult, for instance, to prove the equality between a variable-RIS and a set
(extensional or RIS) having at least one element.

In turn pairwise co-injectivity is necessary to solve equations such as:
\begin{equation}\label{eq:pairwise}
  \{x:X | \G @ \Q\} = \{t \plus \{y:X | \F @ \P\}\}
\end{equation}
in a finite number of iterations. Let's assume for a moment that patterns of
the form $x+n$ ($n$ constant) are allowed and \eqref{eq:pairwise} is
instantiated as follows:
\begin{equation}\label{eq:pairwise1}
  \{x:X | \true @ x+2\} = \{5 \plus \{y:X | \true @ y+8\}\}
\end{equation}
Over the integers, this equation is satisfiable only if $X$ equals
$\num$, as we ask for the non-empty image of two different lines
over the same domain to be equal. $\SATRIS$ would be unable to
conclude this in a finite number of iterations. Given that 5 must
belong to the l.h.s. set, then $3 \in X$. But this implies that 11
belongs to the r.h.s. and so, to keep the sets equal, it must belong
to the l.h.s., then $9 \in X$ and so we have an infinite loop. On
the other hand, if we ask whether the two \emph{lines} are equal or
not:
\begin{equation}\label{eq:pairwise2}
  \{x:X | \true @ (x,x+2)\} = \{(3,5) \plus \{y:X | \true @ (y,y+8)\}\}
\end{equation}
$\SATRIS$ is able to give the right answer. Given that $(3,5)$ must
belong to the l.h.s. set then $3 \in X$. But this implies that
$(3,11)$ belongs to the r.h.s. and so, to keep the sets equal, it
must belong to the l.h.s. but this is impossible as $(3,11)$ is not
a point in the line with equation $y = x+2$. The difference between
\eqref{eq:pairwise1} and \eqref{eq:pairwise2} is that in the former
patterns are not pairwise co-injective while in the latter they are.

Unfortunately, these properties cannot be easily syntactically assessed. Thus
we prefer to restrict $\LRIS$ by adopting a more restrictive definition of
admissible pattern that can be syntactically checked. From a more practical
point of view, however, other patterns could be admitted instead of those given
in Definition \ref{RIS-terms}, with the assumption that if they verify the
conditions stated above the result is surely safe; while if they do not, it is
not safe.

\subsection{Functional predicate symbols}\label{ssec:fps}

Although in general is undecidable to assert the satisfiability of formulas of
the form $\forall \vp(\lnot \Fpv(x,\vp))$, where $\Fpv$ is a formula of an
arbitrary theory $\mathcal{T}$, we have identified a non-trivial fragment of
$\mathcal{T}$ whose satisfiability can be decided. Intuitively, this fragment is
composed by those formulas where $\Fpv$ is of the form $\Fpv_p \land \Fpv_r$
such that $\forall \vp(\lnot (\Fpv_p(x,\vp) \land \Fpv_r(x,\vp)))$ is equivalent
to an existentially quantified formula.

\begin{definition}[Functional predicate]
\label{fps} Let $p \in \Pi_\Ur$ be a predicate symbol of $\Ur$ with arity $n
\geq 2$ . $p(x_1,\dots,x_{n-1},y)$ is said to be a \emph{functional predicate}
if given $x_1,\dots,x_{n-1}$ there exists \emph{at most} one $y$ such that
$p(x_1,\dots,x_{n-1},y)$ holds. \qed
\end{definition}

Functional predicates are a form of encoding functions in a logic language.
Indeed, if $p(x,y)$ is a functional predicate then it is equivalent to $p(x) =
y$ if $p$ is seen as a function. Further, $y$ can be seen as the name given to
the expression $p(x)$. In other words, we can say that $y$ is the
\emph{definition} of $p(x)$.

Note that, for any functional predicate $p$ and given $x_1,\dots,x_{n-1}$, it
may be the case that $p(x_1,\dots,x_{n-1},y)$ does not hold for all $y$. For
this reason we define the following.

\begin{definition}[Pre-condition of a functional predicate]
\label{prefps} Let $p(x_1,\dots,x_{n-1},y)$ be a functional predicate. Let
$\Fpv_q$ be an $\Ur$-formula based on symbols in $\Pi_\Ur \setminus \{p\}$, and
$\vec x$ a subset of $x_1,\dots,x_{n-1}$. $\Fpv_q$ is said to be the
\emph{pre-condition} of $p(x_1,\dots,x_{n-1},y)$ if and only if for any given
$x_1,\dots,x_{n-1}$:
\[
 \Fpv_q(\vec x) \iff \exists y (p(x_1,\dots,x_{n-1},y))
\]
\qed
\end{definition}

The following proposition characterizes a class of formulas including an
existential quantification whose negation is free of universal quantifiers. See
the proof in Appendix \ref{ap:proofs}.

\begin{proposition}\label{negparam}
Let $\Fpv_q$ be the precondition of a functional predicate
$p(x_1,\dots,x_{n-1},y)$. Let $\Fpv_r$ be an $\Ur$-formula, and $\vec x_q$ and
$\vec x_r$ subsets of $x_1,\dots,x_{n-1}$. Then given $x_1,\dots,x_{n-1}$ the
following holds:
\begin{gather*}
\forall y (\lnot (p(x_1,\dots,x_{n-1},y) \land \Fpv_r(\vec x_r,y) )) \\
\iff \lnot \Fpv_q(\vec x_q)
       \lor \Fpv_q(\vec x_q)
            \land \exists z (p(x_1,\dots,x_{n-1},z)
                             \land \lnot \Fpv_r(\vec x_r,z))
\end{gather*}
\qed
\end{proposition}

If a RIS has a filter of the form $\exists n (\Fpv(x,n))$ then its negation is
of the form $\forall n (\lnot \Fpv(x,n))$. However, if $\Fpv$ can be written as
$q(x) \land p(x,n) \land r(x,n)$ where $q$ is the pre-condition of the
functional predicate symbol $p$, then Proposition \ref{negparam} permits to
transform $\forall n (\lnot \Fpv(x,n))$ into a formula free of universal
quantifiers. In effect, we have:
\begin{gather*}
\forall n (\lnot \Fpv(d,n)) \\
\iff \forall n (\lnot (q(x) \land p(x,n) \land r(x,n)))
        \why{definition of $F$} \\
\iff \forall n (\lnot q(x) \lor \lnot (p(x,n) \land r(x,n)))
       \why{distributivity} \\
\iff \lnot q(x) \lor \forall n (\lnot (p(x,n) \land r(x,n)))
       \why{$q$ does not depend on $n$} \\
\iff \lnot q(x)
     \lor \lnot q(x)
     \lor \exists z (q(x) \land p(x,z) \land \lnot r(x,z))
       \why{Proposition \ref{negparam}} \\
\iff \lnot q(x)
     \lor \exists z (q(x) \land p(x,z) \land \lnot r(x,z))
\end{gather*}

Hence, this proposition applied to RIS' filters considerably enlarges the
fragment of decidable $\LRIS$ formulas, as is shown by the following example.

\begin{example}\label{ex:funcpred}
Recently Cristi\'a and Rossi extended \CLPSET, and its \setlog implementation,
to support binary relations and partial functions
\cite{DBLP:journals/tplp/CristiaRF15,Cristia2019}. In this theory, called
$\mathcal{BR}$, binary relations and functions are sets of ordered pairs, and a
number of relational operators are provided as constraints. Such theory has
been proved to be semi-decidable for formulas involving conjunctions,
disjunctions and negations of its constraints. Notwithstanding, in the
remaining of this paper, we
will consider the instance of $\LRIS$ based on the decidable fragment of
$\mathcal{BR}$ implemented in \setlog.

Among others, $\mathcal{BR}$ provides the following constraints, where $F, G,
H$ are binary relations or partial functions: $\Apply(F,x,y)$, whose
interpretation is $F(x) = y$; $\Comp(F,G,H)$ which is interpreted as the
composition of $F$ and $G$, i.e., $H = R \comp T$; and $\Pfun(F)$ which
constrains set $F$ to be a partial function.

$\Apply(F,x,y)$ is a functional predicate with precondition $\Pfun(F) \land
\Ncomp(\{(x,x)\},F,\e)$, for any $F$ and $x$, where $\Ncomp$ is $\lnot \Comp$
(notice that, if $F$ is a partial function, $\Ncomp(\{(x,x)\},F,\e)$ is true iff
$x$ belongs to the domain of $F$). Hence, the following $\LRIS$ formula is
decidable for any $F$, $x_1$, $x_2$ and $x_3$:
\begin{equation}\label{eq:apply}
F = \{(x_1,x_2),(x_1,x_3)\} \land \{y:\{1\} | \exists z (\Apply(F,y,z) \land z
\neq 0)\} = \e
\end{equation}
as the negation of $\exists z (\Apply(F,y,z) \land z \neq 0)$ can be turned
into the following decidable $\Ur$-formula, due to Proposition \ref{negparam}:
\[
\lnot (\Pfun(F) \land \Ncomp(\{(x,x)\},F,\e))
              \lor (\Pfun(F) \land \Ncomp(\{(x,x)\},F,\e) \land \Apply(F,y,n) \land n = 0)
\]
for some new $n$. Possible solutions for \eqref{eq:apply} are: $x_2 \neq x_3$
(i.e., $F$ is not a function), $x_1 \neq 1$ (i.e., $1$ does not belong to the
domain of $F$); $x_1 = 1 \land x_2 = 0 \land x_3 = 0$ (i.e., $F = \{(1,0)\}$).
\qed
\end{example}

\section{RIS in Practice}\label{impl}

RIS have been implemented in Prolog as an extension of \setlog \cite{setlog}, a
freely available implementation of $\mathcal{BR}$ \cite{Cristia2019} extended
with FD-constraints \cite{Palu:2003:IFD:888251.888272}. In this case, the
theory $\Ur$ is basically the theory of \emph{hereditarily finite hybrid sets
and binary relations}, augmented with that of \CLPFD, that is integer
arithmetic over finite domains. This theory provides the same function symbols
as $\LRIS$ for building extensional set terms (namely, $\e$ and $\ww$), along
with a set of predicate symbols including those of $\LRIS$, with the same
interpretation. In addition, $\mathcal{BR}$ provides predicate symbols such as
composition of binary relations, converse of binary relations, domain, range,
etc., and the usual function symbols representing operations over integer
numbers (e.g., $+,-,\Mod$, etc.), as well as the predicate symbols
$\mathsf{size}$, representing set cardinality, and $\leq$, representing the
order relation on the integers. One notable difference w.r.t. $\LRIS$ is that
set elements can be either finite sets or non-set elements of any sort (i.e.,
nested sets are allowed).

The theory underlying \setlog is endowed with a constraint solver which is
proved to be a semi-decision procedure for its formulas. More precisely, the
constraint solver is a decision procedure for the subclass of
$\mathcal{BR}$-formulas not involving (two particular forms of) relational
composition \cite{Cristia2019}.

Syntactic differences between the abstract syntax used in this paper and the
concrete syntax used in \setlog are made evident by the following examples.

\begin{example}
The $\mathcal{RIS}$-formula:
\[
\{5\} \in \{x:\{y \plus D\} | x \neq \emptyset
\land 5 \notin x @ x \}
\]
is written in \setlog as:
\[
\{5\} \In \Ris(X \In \{Y /
D\},X \Neq \{\} \And 5 \Nin X,X)
\]
where $\Ris$ is a function symbol whose
arguments are: \emph{i}) a constraint of the form $X \In A$ where $X$ is the
control term and $A$ the domain of the RIS; \emph{ii}) the filter given as a
\setlog formula; and \emph{iii}) the pattern given as a \setlog term. Filters
and patterns can be omitted as in $\LRIS$. Variables must start with an
uppercase letter; the set constructor symbols for both $\LRIS$ and \setlog
sets terms are written as $\{\cdot / \cdot\}$. If this formula is provided to
\setlog it answers $\mathsf{no}$ because the formula is unsatisfiable.
\qed
\end{example}

The following are more examples of RIS that can be written in \setlog.


\begin{example}\label{ex:noparam}
\mbox{}
\begin{itemize}
\item The multiples of $N$ belonging to $D$:
\[
\Ris(X \In D, 0 \Is X \Mod N)
\]
where $\Is$ is the Prolog built-in predicate that forces the evaluation
of the arithmetic expression at its right-hand side.

\item The sets belonging to $D$ containing a given set $A$:
\[
\Ris(S \In D, \mathsf{subset}(A,S))
\]

\item A function that maps integers to their squares:
\[
\Ris([X,Y] \In D, Y \Is X*X)
\]
where ordered pairs are written using $[\cdot,\cdot]$; note that the pattern
can be omitted since it is the same as the control term, that is $[X,Y]$.
\qed
\end{itemize}
\end{example}

Actually \setlog implements the extended version of $\LRIS$ described in
Section \ref{sec:extension}. In particular, in \setlog a RIS can contain
parameters, i.e., existentially quantified variables local to the RIS; and
functional predicates can be conveniently declared. Parameters
are listed as the second argument of the $\Ris$ term, while functional
predicates must be located just
after the pattern, as the last argument. Both parameters and functional
predicates are optional arguments of a RIS term.

\begin{example}
A function that maps sets to their cardinalities
provided they are greater than $1$:
\[
\Ris(S \In D,[C], C > 1,[S,C],\mathsf{size}(S,C))
\]
where $C$ is a parameter and $\mathsf{size}(S,C)$ is a functional
predicate, whose intuitive meaning is  $C = \lvert S \rvert$.
\qed
\end{example}

RIS patterns in \setlog can be any term (including $\{\cdot / \cdot\}$). If
they verify the conditions given in Section \ref{sec:controlPatterns}, then the
solver is guaranteed to be a decision procedure; otherwise this may be not the
case.

\begin{example}
The formula $\{x:[1,4] @ 2*x\} = \{2,4,6,8\}$ can be
written in \setlog as:
\[
\Ris(X \In \Int(1,4),\mathsf{true},2*X) = \{2,4,6,8\}
\]
where $\Int(1,4)$ is the \setlog syntax to denote the interval$[1,4]$. If this
formula is provided to \setlog it answers $\mathsf{yes}$. Note that $2*X$ is a
bijective pattern and since it is the only pattern in the formula is trivially
pairwise co-injective. Caution must be taken if this pattern is mixed with
others. \qed
\end{example}

When the domain of a RIS is at least partially specified it is also possible to
explicitly enumerate the elements of the set denoted by the RIS by means of the
$\Is$ operator.

\begin{example}
When
\[
\mathit{Even} \Is \Ris(X \In
\Int(1,100),\mathsf{true},2*X)
\]
is run on \setlog it immediately answers $Even
= \{2,4,\dots,200\}$.
\qed
\end{example}

In \setlog the language of the RIS and the language of the parameter theory
$\Ur$ are completely amalgamated. Thus, it is possible for example to use
literals of the latter in formulas of the former, as well as to share variables
of both. The following example exploits this feature.

\begin{example}
A formula to find out whether $N$ is prime or not:
\[
N > 1 \And \mathit{MD} \Is N \Div 2 \And \Ris(X \In \Int(2,\mathit{MD}), 0 \Is
N \Mod X) = \emptyset
\]
The idea is to check if the set of proper divisors of $N$ (i.e.,
$\{x:[2,\mathit{MD}] | 0 = N \mod x\}$) is empty or not. Then, if $N$ is bound
to, say, 20, \setlog answers \textsf{no}; but if it is bound to 101 it answers
$N = 101, \mathit{MD} = 50$.
\qed
\end{example}

\begin{remark}
In \setlog checks for detecting non-admissible formulas (see Definition
\ref{d:admissible_formula}) are limited to atomic formulas (i.e.,
$\mathcal{RIS}$-constraints). As a consequence, if the solver has to deal with
general non-admissible formulas, then there is a risk that it
will go into an infinite loop. As observed at the end of Section
\ref{ssec:admissible}, however, these formulas are quite ``unusual'', hence
this behavior is not perceived as a problem in practice.
\qed
\end{remark}

\subsection{Case studies}\label{casestudies}

In this subsection we present two case studies showing the capabilities of
\setlog concerning RIS. The first one shows how \setlog can be used as an
automated theorem prover when restricted universal quantifiers are involved. In
this case study a non-trivial security property of a security model is proved.
In the second case study RIS are used to specify a simplified version of the
\texttt{grep} program. There we show that RIS provide a sort of second-order
language because it is possible to iterate over sets of sets.

\begin{casestudy}[Bell-LaPadula's security condition]
Bell and LaPadula proposed a security model of an operating system enforcing a
security policy known as multi-level security \cite{BLP1,BLP2}. The model is a
state machine described with set theory and first-order logic. This model
verifies two state invariants, one of which is called \emph{security
condition}. This condition can be stated as follows:
\begin{equation}
\forall (p,f) \in proctbl (\mathit{scf}(f) \preceq scp(p))
\end{equation}
where $proctbl$ is the process table represented as a set of ordered pairs of
the form $(p,f)$ where $p$ is a process and $f$ a file opened in read
mode by $p$; and $\mathit{scf}$ and $scp$ are functions returning the
\emph{security class} of files and processes, respectively. A security class is
an ordered pair of the form $(l,C)$ where $l \in \nat$ is called \emph{security
level}, and $C$ is a \emph{set of categories}. Security classes are partially
ordered by the \emph{dominates} relation ($\preceq$) defined as: $(l_1,C_1)
\preceq (l_2,C_2) \iff l_1 \leq l_2 \land C_1 \subseteq C_2$.

Bell-LaPadula's security condition is expressible in \setlog using a $\Forall$
constraint as follows (recall Example \ref{ex:funcpred}):
\begin{equation}\label{blp1}
\begin{split}
\mathsf{sec}&\mathsf{Cond}(\mathit{Scf},Scp,Proctbl) \defs \\
& \Forall([P,F] \in Proctbl,
        \Applysl(\mathit{Scf},F,S_f)
        \And \Applysl(Scp,P,S_p)
        \And \mathsf{dominates}(S_f,S_p))
\end{split}
\end{equation}
where $S_f$ and $S_p$ are parameters (cf. Section \ref{ssec:fps}). Since
$\Applysl(G,X,Y) $ is a functional predicate, formula \eqref{blp1} lies inside of
the results given in Section \ref{ssec:fps}\footnote{Observe that, if $f$ and
$g$ are two functional predicates of arity 2, then we can introduce a new
predicate $h$, $h(x_1,x_2,w) \iff w = (n_1,n_2) \land f(x_1,n_1) \land
g(x_2,n_2)$, where $h(x_1,x_2,w)$ is trivially a functional predicate.}. In
turn, the $\mathsf{dominates}$ relation can be defined easily:
\begin{equation}
\mathsf{dominates}(S_1,S_2) \defs S_1 = [L_1,C_1] \And S_2 = [L_2,C_2]
\And L_1 \leq L_2 \And C_1 \subseteq C_2
\end{equation}

More importantly, although \setlog implements a semi-decision procedure for
formulas involving partial functions, it can be used to (automatically) prove
whether or not a given operation in Bell-LaPadula's model preserves the
security condition. For example, the operation describing process $P$
requesting file $F$ to be opened in read mode can be described with a \setlog
formula:
\begin{equation}
\begin{split}
\mathsf{openRead}&(\mathit{Scf},Scp,Proctbl,P,F,Proctbl') \defs \\
& [P,F] \notin Proctbl
  \And \Applysl(\mathit{Scf},F,S_f)
   \And \Applysl(Scp,P,S_p)  \And \mathsf{dominates}(S_f,S_p) \\
&  \And Proctbl' = \{[P,F] \plus Proctbl\}
\end{split}
\end{equation}
where $Proctbl'$ represents the value of $Proctbl$ in the next state (as is
usual in formal notations such as B and Z).

Hence, the proof obligation asserting that $\mathsf{openRead}$ preserves
$\mathsf{secCond}$ as an invariant is, informally:
\begin{gather*}
\text{$\mathsf{secCond}$ is true of $Proctbl$} \\
{}\land \text{$\mathsf{openRead}$ is called on $Proctbl$, thus returning $Proctbl'$} \\
{}\implies \text{$\mathsf{secCond}$ is true of $Proctbl'$}
\end{gather*}
However, if \setlog is going to be used to discharge this proof obligation, it
should be submitted in negated form:
\begin{gather}
\mathsf{secCond}(\mathit{Scf},Scp,Proctbl) \notag \\
{}\And \mathsf{openRead}(Proctbl,P,F,Proctbl') \label{eq:openreadpi} \\
{}\And \mathsf{nforall}([P,F] \In Proctbl',
        \Applysl(\mathit{Scf},F_1,S_f)
        \And \Applysl(Scp,P_1,S_p) \And \mathit{dominates}(S_f,S_p))
     \notag
\end{gather}
in which case in $\approx$ .5 s \setlog answers $\false$, as expected.

Besides, assume the specifier makes a mistake in $\mathsf{openRead}$ such that the
security condition is no longer preserved. For instance, (s)he forgets
$\mathsf{dominates}(S_f,S_p)$ as a pre-condition. When (s)he attempts to
discharge \eqref{eq:openreadpi}, \setlog will return a counterexample showing
why the security condition has been violated\footnote{Only the core of the
counterexample is shown.}:
\begin{gather*}
Scf = \{[F_1,[Lf,\{N_6 \plus N_5\}]] \plus N_4\}, \\
ScP = \{[P_1,[Lp,Cp]] \plus N_2\}, \\
Proctbl' = \{[P_1,F_1] \plus Proctbl\},  \\
\text{Constraint: } N_6 \notin Cp
\end{gather*}
That is, the security class of $P_1$ does not dominate the security class of
$F_1$ because $N_6$ does not belong to $Cp$.

The source code of this case study can be found in Appendix \ref{app:blp}.
\qed
\end{casestudy}

In $\LRIS$ it is possible to iterate or quantify over sets whose elements are
other sets---i.e., $\LRIS$ offers a sort of second-order language. Differently
from previous versions of \setlog, however, quantified set variables can be
also intensional sets (namely, RIS), not only extensional sets.
In other words, in $\LRIS$ intensional sets are real first-class citizens of the
constraint language. The following case study illustrates this feature of
$\LRIS$.

\begin{casestudy}[A simple \texttt{grep} program]
We can model a text file as a set of lines and each
line as a set of words, i.e., the file is modeled as a
collection of sets. With this representation it is easy to implement a basic
version of the usual $\mathsf{find}$ function using a formula based on set unification:
\begin{equation}
\mathsf{find}(File,Word) \defs File = \{\{Word / Line\} / Rest\}
\end{equation}

However, no formula based on set unification can describe the \verb+grep+
program which collects \emph{all} the lines of the file where a
given word is found. The reason is that it requires to inspect all the lines
of the file---i.e., all the sets of a set. For such program a RIS can do the
job:
\begin{equation}\label{eq:grep}
\mathsf{grep}(File,\mathit{Word},Result) \defs
  Result = \Ris(Line \In File, \mathit{Word} \In Line)
\end{equation}

Furthermore, by adding a parameter to the definition of $\mathsf{grep}$ we can
implement the \verb+-v+ option which reverses the search:
\begin{equation}
\begin{split}
\mathsf{grep}&(Opt,File,\mathit{Word},Result) \defs \\
    & Opt = \text{' '} \And Result = \Ris(Line \In File, \mathit{Word} \In Line) \\
    & \Or Opt = \text{'v'} \And Result = \Ris(Line \In File, \mathit{Word} \Nin Line)
\end{split}
\end{equation}

In this way $\mathsf{grep}$ is both a program \emph{and} a formula. This means
that we can use \setlog to run $\mathsf{grep}$ and to prove properties of it. For
instance, we can run:
\begin{equation}
\mathsf{grep(\text{' '},\{\{hello, world\},\{to, be, or, not, to, be\},\{i,said,hello,sir\}\},hello},R)
\end{equation}
and the answer will be:
\begin{equation}
R = \mathsf{\{\{i,said,hello,sir\},\{hello,world\}\}}
\end{equation}

Now we can use \setlog to prove the following two general properties of
$\mathsf{grep}$:
\begin{gather}
\mathsf{grep}(\text{' '},F,W,R_1) \land \mathsf{grep}(\text{'v'},F,W,R_2) \implies \Cupsl(R_1,R_2,F) \\
\mathsf{grep}(\text{' '},F,W,R_1) \land \mathsf{grep}(\text{'v'},F,W,R_2) \implies \Disjsl(R_1,R_2)
\end{gather}
which, as always, should be submitted to \setlog in negated form:
\begin{gather}
\mathsf{grep}(\text{' '},F,W,R_1) \And \mathsf{grep}(\text{'v'},F,W,R_2) \And \Ncupsl(R_1,R_2,F) \label{eq:grep_prop1} \\
\mathsf{grep}(\text{' '},F,W,R_1) \And \mathsf{grep}(\text{'v'},F,W,R_2) \And \Ndisjsl(R_1,R_2)
\end{gather}
in which case \setlog answers $\false$, as expected. \qed
\end{casestudy}

This is a significant extension w.r.t. previous versions of \setlog. In
fact, in those versions it was possible to give a definition such as
\eqref{eq:grep} by using general intensional sets; but, in that case, the
solver would fail to prove the unsatisfiability of a general formula such as
\eqref{eq:grep_prop1}, simply because it tries to replace the intensional
sets with the corresponding extensional definitions, instead of solving the
relevant constraints directly over the intensional sets, as it does when RIS
are used.

\section{Empirical evaluation}\label{experiments}

The goal of this empirical evaluation is to provide experimental evidence that
the \setlog implementation of $\SATRIS$ works in practice. To this end we
selected 176 problems from the SET collection of the TPTP library \cite{Sut09}.
The selected problems are those representing quantifier-free, first-order set
theory results involving conjunctions, disjunctions and negations of set
equality, membership, union, intersection, difference, disjointness and
complement that can be encoded in \setlog. The TPTP.SET problem collection has
been used to empirically evaluate previous versions of \setlog
\cite{Cristia2019}.

From these 176 formulas we generated two larger collections of formulas involving RIS:
the \textsc{Negative} collection, where formulas are expected to be
unsatisfiable; and the \textsc{Positive} collection, where formulas are
expected to be satisfiable, although a few of them are not.

The \textsc{Negative} collection was generated as follows. Let $\Theta$ be the
formula of one of the selected 176 problems. Let $\mathit{vars}(\Theta)$ be the
set of free variables of $\Theta$. Then $\Theta$ is transformed into:
\[
\Theta'  \bigwedge_{A \in \mathit{vars}(\Theta)}
   A = \Ris(X \In
D_A,[\mathscr{V}],\mathscr{F},\mathscr{P},\mathscr{C})
%
%
\]
where: $X$ and $D_A$ are variables; $\Theta' $ is obtained from $\Theta$ by
replacing every occurrence of $\{\}$ by $\Ris(X \In
\mathscr{D},[\;],\mathscr{F}_\emptyset)$; and $\mathscr{D}$,
$\mathscr{F}_\emptyset$, $\mathscr{V}$, $\mathscr{F}$, $\mathscr{P}$ and
$\mathscr{C}$ are instantiated as depicted in Table
\ref{tab:neggen}. This means that for each instance of $\mathscr{V}$, $\mathscr{F}$ and
$\mathscr{C}$, an instance of $\mathscr{P}$ and an instance of $\mathscr{D}$
and $\mathscr{F}_\emptyset$ are selected. Besides, note that for each $A \in
\mathit{vars}(\Theta)$ a different domain variable is generated (i.e., $D_A$).
In this way, the initial 176 formulas are transformed into 5813 formulas where
all variables and every occurrence of the empty set are RIS terms.

\begin{table}
\tbl{\label{tab:neggen}Domain, filter, pattern and functional predicates
used to generate RIS formulas. Recall that in \setlog the fifth argument of a RIS
term is used to specify the functional predicates possibly occurring in the
RIS.}{%
\begin{tabularx}{\textwidth}{XX}
\toprule
\multicolumn{2}{c}{\textsc{Instances of $\mathscr{V}$, $\mathscr{F}$
and $\mathscr{C}$}} \\
\multicolumn{2}{l}{$B_i$ and $D$ are new variables for each instance}
\\[1mm]
$\Ris(X \In D,[\;],\True,\mathscr{P},\True)$ &
$\Ris(X \In D,[\;],X \In B,\mathscr{P},\True)$ \\
$\Ris(X \In D,[\;],X \Notinsl B,\mathscr{P},\True)$ &
$\Ris(X \In D,[B],X \In B,\mathscr{P},\Cupsl(B_1,B_2,B))$ \\
$\Ris(X \In D,[B],X \Notinsl B,\mathscr{P},\Cupsl(B_1,B_2,B))$ &
$\Ris(X \In D,[\;],\Disjsl(X,B_1),\mathscr{P},\True)$ \\
$\Ris(X \In D,[\;],\Ndisjsl(X,B_1),\mathscr{P},\True)$ &
$\Ris(X \In D,[\;],X = B_1,\mathscr{P},\True)$ \\
$\Ris(X \In D,[\;],X \Neqsl B_1),\mathscr{P},\True)$ & \\
\midrule
\multicolumn{2}{c}{\textsc{Instances of $\mathscr{P}$}} \\
\multicolumn{2}{l}{$X$ is the control variable and $V$ is a new variable
for each instance} \\[1mm]
$X$ & $[X,V]$ \\
\midrule
\multicolumn{2}{c}{\textsc{Instances of $\mathscr{D}$ and
$\mathscr{F}_\emptyset$}} \\
\multicolumn{2}{l}{$B_i$ are new variables for each instance} \\[1mm]
$\Ris(X \In \{\},[\;],\False)$ &
$\Ris(X \In \{\},[\;],\True)$ \\
$\Ris(X \In \{B_1/B_2\},[\;],X \Neqsl B_1 \And X \Notinsl B_2)$ & \\
\bottomrule
\end{tabularx}
}
\end{table}

In turn, the \textsc{Positive} collection is generated as follows. First, one
$\mathcal{RIS}$-constraint of each of the 176 original formulas is negated
(specifically, it is replaced by its negated version). In general this turns
the formula from unsatisfiable to satisfiable. However, 27 formulas become
trivial and so were removed. Then, if $\Theta$ is one of the remaining 149
formulas is transformed as follows:
\[
\Theta' \bigwedge_{A \in \mathit{vars}(\Theta)} A = \Ris(X \In
D_A,[\mathscr{V}],\mathscr{F},\mathscr{P},\mathscr{C}) \land D_A \Neqsl \{\}
\]
where the same considerations of the \textsc{Negative} collection apply. The
last conjunct implies that every non-empty RIS has at least one element. In
this way the \textsc{Positive} collection has 4728 problems.

The union of the \textsc{Negative} and \textsc{Positive} collections makes a
10541 problems benchmark to evaluate the automated processing of RIS.

We run two experiments to empirically assess the effectiveness and efficiency
of \setlog on the automated processing of RIS. In these experiments we measure
the number of formulas that \setlog solves before a given timeout and the
average time in doing so.

The results of these experiments are summarized in Table \ref{tab:experim}. In
the first experiment the timeout is 2 seconds; in the second experiment the
timeout is 60 seconds but only the unsolved formulas of the first experiment
are considered. That is, the 976 `negative' formulas and the 586 `positive'
formulas used in the 60 seconds experiment are those that remain unsolved in
the 2 seconds experiment. The last row of the table shows a summary of both
experiments.

According to these figures, it can be said that augmenting the timeout from 2
seconds to 60 seconds produces a modest gain in the number of solved formulas,
but given that most of the formulas are solved in a few milliseconds, setting a
higher timeout would not be harmful. Finally, considering both experiments
\setlog solves 89\% of the benchmark in just below three quarters of a second in average.

\begin{table}
\tbl{\label{tab:experim}Summary of the empirical evaluation}{%
\begin{tabularx}{\textwidth}{Xrrrr}
\toprule \textsc{Collection} & \textsc{Problems} & \textsc{Solved} &
\textsc{Percentage} &
\textsc{Average Time} \\
\midrule \multicolumn{5}{c}{\textsc{First experiment (2 seconds timeout)}}
\\[1mm]
\textsc{Negative} & 5813 & 4837 & 83\% & 0.069 s \\
\textsc{Positive} & 4728 & 4142 & 88\% & 0.035 s \\
\midrule \multicolumn{5}{c}{\textsc{Second experiment (60 seconds timeout)}}
\\[1mm]
\textsc{Negative} & 976 & 266 & 27\% & 13.718 s \\
\textsc{Positive} & 586 & 173 & 30\% & 13.680 s \\
\midrule
\textsc{Summary}  & 10541 & 9418 & 89\% & 0.690 s \\
\bottomrule
\end{tabularx}
}
\end{table}

Although these results are good, there is room for improvements. One such
improvement is a more efficient implementation of derived constraints (cf.
Section \ref{expressiveness}). In a previous work \cite{Cristia2019}, we have
used the initial 176 formulas (that is, formulas where set variables are
\emph{not} bound to RIS terms) to assess a version of \setlog where
intersection, subset and difference are implemented as derived constraints. That
is, for instance, an $\Cap$ constraint is rewritten as a formula based on $\Cup$
and $\Disj$ instead of being processed by an \emph{ad-hoc} rewriting procedure.
When these formulas are run on that version of \setlog, it solves 94\% in 0.078
seconds in average; but when just intersection and subset are processed by
\emph{ad-hoc} rewriting procedures, \setlog solves all of them. In the current
version of \setlog, intersection, subset and difference on RIS terms are treated
as derived constraints. Based on the results obtained with the previous work, it
is reasonable to assume that extending the \emph{ad-hoc} rewriting procedures
for these constraints to RIS terms might yield significantly better results than
those reported in Table \ref{tab:experim}.

\subsection{\label{details}Technical details of the empirical evaluation}

The experiments were performed on a Latitude E7470 (06DC) with a 4 core
Intel(R) Core\texttrademark{} i7-6600U CPU at 2.60GHz with 8 Gb of main memory,
running Linux Ubuntu 18.04.3 (LTS) 64-bit with kernel 4.15.0-65-generic.
\setlog 4.9.6-17i over SWI-Prolog (multi-threaded, 64 bits, version 7.6.4) was
used during the experiments.

Each \setlog formula was run within the following Prolog program:
\begin{verbatim}
use_module(library(dialect/sicstus/timeout)).
consult('setlog.pl').
consult_lib.
set_prolog_flag(toplevel_print_options,[quoted(true),portray(true)]).
get_time(Tini).
time_out(setlog(<FORMULA>), <TIMEOUT>, _RES).
get_time(Tend).
\end{verbatim}
where \verb+<FORMULA>+ is replaced by each formula and \verb+<TIMEOUT>+ is
either 2000 or 60000 depending on the experiment. Each of these programs was
run from the command line as follows:
\begin{verbatim}
prolog -q < <PROG>
\end{verbatim}
The execution time for each problem is calculated as \texttt{Tend - Tini}.

The full data set containing all these Prolog programs can be downloaded from
\url{https://www.dropbox.com/s/d1ysiq3eji2xg9a/exTOPLAS.tar.gz?dl=0}.

\section{Related Work}\label{relwork}

As mentioned in Section \ref{sec:intro}, some form of intensional sets is
offered by the CLP language \CLPSET. Specifically, \CLPSET supports general
intensional sets by implementing set-grouping on top of the language itself
(i.e., not as a primitive feature of the language). Hence, formulas involving
intensional sets fall outside the scope of \CLPSET's decision procedure. As an
example, the formula $A \cap B \neq \{x : x \in A \land x \in B\}$, which is
written in \CLPSET as
\[
\mathsf{inters}(A,B,C) \And D = \{X : X \hspace{2pt}\mathsf{in}\hspace{2pt} A
\And X \hspace{2pt}\mathsf{in}\hspace{2pt} B\} \And C \hspace{2pt} \mathsf{neq}
\hspace{2pt} D
\]
is (wrongly) found to be $\true$ by the \CLPSET resolution
procedure. Conversely, the same formula but written using RIS is
(correctly) found to be unsatisfiable by $\SATRIS$.

A very general proposal providing real \emph{intensional set constraints},
where intensional set processing is embedded within a general set constraint
solver, is CLP($\{\mathcal{D}\}$) \cite{DBLP:conf/iclp/DovierPR03}.
CLP($\{\mathcal{D}\}$) is a CLP language offering arbitrarily nested
extensional and intensional sets of elements over a generic constraint domain
$\mathcal{D}$. No working implementation of this proposal, however, has been
developed. As observed in \cite{DBLP:conf/iclp/DovierPR03}, the presence of
undecidable constraints such as $\{x : p(x)\} = \{x : q(x)\}$ (where $p$ and
$q$ can have an infinite number of solutions) ``prevents us from developing a
parametric and complete solver''. Conversely, this problem can be
``approximated'' using RIS as $\{x:D_1 | p(x)\} = \{x:D_2 | q(x)\}$, $D_1$,
$D_2$ variables. For $\SATRIS$, this is a solved form formula admitting at
least one solution, namely $D_1 = D_2 = \e$; hence, it is simply returned
unchanged by the solver. Generally speaking, finding a fragment of intensional
sets that is both decidable and expressive is a key issue for the development
of an effective tool for reasoning with intensional sets. RIS, as presented
here, may be one step toward this goal.

The usefulness of ``a direct support for reasoning about set
comprehension'' in the context of SMT solvers has been also recently advocated
by Lam and Cervesato \cite{DBLP:conf/smt/LamC14}.
In their approach, however, no ad-hoc solver for intensional set constraints is
indeed developed; rather, the satisfiability problem for formulas featuring
intensional sets over a standard theory (e.g., linear integer arithmetic) is
reduced to solving satisfiability constraints over this same theory, extended
with an uninterpreted sort for sets and an uninterpreted binary predicate
encoding set membership.

A number of works in the area of Computable Set Theory (CST) have studied the
satisfiability problem of different fragments of set theory with quantifiers.
Clearly, the availability of quantifiers plus set constraints opens the door to
intensional sets. Cantone and Zarba \cite{DBLP:conf/tableaux/CantoneZ00}
introduce the language \textbf{2LS}($\mathcal{L}$) which is parametric w.r.t. a
first-order language $\mathcal{L}$. \textbf{2LS}($\mathcal{L}$) extends
$\mathcal{L}$ with set constants, functional symbols, standard Boolean set
operators, set membership and equality. The authors show that, if the
first-order theory underlying $\mathcal{L}$, $\mathcal{T}$, is ground-decidable
and the collection of ground terms of $\mathcal{T}$ is finite, then the
$\mathcal{T}$-satisfiability problem for $\exists\forall$-sentences of
\textbf{2LS}($\mathcal{L}$), i.e., formulas where all terms not involving any
set-theoretic symbol are \emph{ground}, is decidable.

In another work Cantone and Longo present the $\forall_{0,2}^{\pi}$ language
\cite{DBLP:journals/tcs/CantoneL14}. $\forall_{0,2}^{\pi}$ is a two-sorted
quantified fragment of set theory which allows restricted quantifiers of the
forms $(\forall x \in A)$, $(\exists x \in A)$, $(\forall (x,y) \in R)$,
$(\exists (x,y) \in R)$ and literals of the forms $x \in A$, $(x,y) \in R$, $A
= B$, $R = S$, where $A$ and $B$ are set variables (i.e., variables ranging
over sets) and $R$ and $S$ are relation variables (i.e., variables ranging over
binary relations). $\forall_{0,2}^{\pi}$-formulas can be conjunctions of RUQ
and conjunctions of Restricted Existential Quantifiers (REQ),
with some mild restrictions. This language is expressive enough as to describe
all the Boolean set operators and disequalities of relational operators such as
composition, domain and relational image (e.g., $R \comp S \subseteq T$ but not
$T \subseteq R \comp S$). Although $\forall_{0,2}^{\pi}$ is not parametric
w.r.t. a first-order language, certain forms of intensional sets can be
described as well. Cantone and Longo show that $\forall_{0,2}^{\pi}$ is
decidable.

One important difference between our work and those on CST is that the later is
mainly concerned with decidability results for fragments of set theory, while
the rewriting systems used in our works are also able to generate a finite
representation of all the (possibly infinitely many) solutions of each input
formula.

Several logics (e.g.,
\cite{DBLP:conf/vmcai/DragoiHVWZ14,DBLP:conf/lpar/VeanesS08,DBLP:conf/frocos/WiesPK09})
provide some forms of intensional sets. However, in some cases, for the formula
to be inside the decision procedure, the intensional sets must have a ground
domain; in others, set operators do not include set equality; and in others,
they present a semi-decision procedure. Handling intensional sets can be
related also to handling universal quantifiers in a logical setting, since
intensional sets ``hide'' a universal quantifier. Tools such as SMT solvers
deal with this kind of problems (see, e.g., \cite{Deharbe7} and
\cite{DBLP:conf/sas/BjornerMR13}), although in general they are complete only
in quite restricted cases \cite{DBLP:conf/cav/GeM09}. Recently, a language
admitting some forms of quantified formulas over sets was proven to be
decidable in the context of separation logic \cite{DBLP:conf/sofsem/GaoCW19}.

Our decision procedure finds models for formulas with \emph{finite} but
\emph{unbounded} domains, i.e., their cardinalities are not constrained by a
fixed value. The field of finite model finding faces a similar problem but
usually with \emph{bounded} domains. There are two basic styles of model
finding: the MACE-style in which the formula is transformed into a SAT problem
\cite{Claessen00}; and the SEM-style which uses constraint solving techniques
\cite{DBLP:conf/cade/ZhangZ96}. Our approach is closer to the SEM-style as it
is based on constraint programming. However, since both styles do not deal with
quantified domains as sets, then they cannot reduce the domain every time an
element is identified, as we do with RIS---for instance, in rule
\eqref{e:empty2}. Instead, they set a size for the domain and try to find a
model at most as large as that.

Ideas from finite model finding were taken as inspiration by Reynolds et al.
\cite{DBLP:conf/cav/ReynoldsTGK13} for handling universal quantifiers in SMT.
These authors propose to find finite models for infinite universally quantified
formulas by considering finite domains.
In particular, Reynolds et
al. make use of the cardinality operator for the sorts of quantified variables
and propose a solver for a theory based on this operator. Then, they make a
guess of the cardinality for a quantified sort and use the solver to try to
find a model there. In the default strategy, the initial guess is 1 and it is
incremented in 1. Note that our approach does not need a cardinality operator
because it operates directly over a theory of sets.

\section{Concluding Remarks}\label{concl}

In this paper we have shown a decision procedure for an expressive class
of intensional sets, called Restricted Intensional Sets (RIS). Key features of
this procedure are: it returns a finite representation of all possible
solutions of the input formula; it allows set elements to be variables; it is
parametric with respect to any first-order theory endowed with a decision
procedure; and it is implemented as part of the \setlog tool. Besides, we have
shown through a number of examples and two case studies that, although RIS are
a subclass of general intensional sets, they are still sufficiently expressive
as to encode and solve many interesting problems. Finally, an extensive
empirical evaluation provides evidence that the tool can be used in practice.

We foresee some promising lines of work concerning RIS:
\begin{itemize}
\item Some set-theoretic operators (e.g., intersection),
which are provided in $\LRIS$ as derived constraints, are first rewritten into
a $\LRIS$ formula and then $\SATRIS$ is applied to the resulting formula. This
threatens the efficiency of $\SATRIS$. Hence, a future work is to implement
specific rewriting procedures for some widely used constraints such as
intersection.
\item The relational operators available in $\mathcal{BR}$ could be extended to
allow for RIS as arguments.
\item As it is now, $\LRIS$ admits only RUQ of the form $\forall x \in A (\phi(x))$,
that is, a single bound variable and its domain. Extending $\LRIS$ to allow
multiple bound variables, each with its own domain, would require to admit
certain $\mathcal{RIS}$-formulas as RIS filters. For example, $\forall x \in A
(\forall y \in B (\phi(x,y)))$ can be encoded as:
\[
A \subseteq \{x:A | B \subseteq \{y:B | \phi(x,y)\}\}
\]
but this is not a $\mathcal{RIS}$-formula because the filter of the outermost
RIS is a $\mathcal{RIS}$-formula. Our intuition is that this would not only be
decidable but relatively efficient too, which, in the end, would be aligned
with some results of CST.

An alternative way of encoding RUQ with multiple variables is by allowing
the recently added Cartesian products \cite{DBLP:conf/RelMiCS/CristiaR18}
as RIS domains. We plan to assess both approaches.
\end{itemize}

\bibliographystyle{ACM-Reference-Format-Journals}
\bibliography{/home/mcristia/escritos/biblio}


\pagebreak
\appendix


\section{Rewrite rules}\label{ap:rules}

This section lists all the $\LRIS$ rewrite rules for $=$, $\neq$, $\in$ and
$\notin$ constraints not given in Section \ref{rules}, along with the
rules for $\Set$ and $\isx$ constraints; rules for $\Cup$ and $\Disj$
constraints are instead all shown in Section \ref{rules} and are not repeated
here. Many of the rules listed in this appendix are borrowed directly from
\cite{Dovier00}.

We adopt the following notational conventions:
$s,t,u$ (possibly subscripted) stand for arbitrary
$\Ur$-terms;
$A,B,C,D$ stand for arbitrary
$\mathcal{RIS}$-terms of sort $\sSet$ (either extensional or intensional,
variable or not);
$\bar D,\bar E$ represent either variables of sort $\sSet$ or variable-RIS;
$X,N$ are variables of sort $\sSet$ representing extensional sets (not RIS)
while $x$ is a variable of sort $\sU$; $\varnothing$ represents either
$\emptyset$ or a RIS with empty domain (e.g., $\{\emptyset | \flt @ \ptt\}$);
finally, $\textbf{\textit{X}}$ represents either the variable $X$ or a set term
containing $X$ as its innermost variable set part, i.e., $\{t_1,\dots,t_n
\sqcup X\}$. Note that, when $n=0$, $\{t_1,\dots,t_n \sqcup X\}$ is just
$X$.

In all rules, variables appearing in the right-hand side but not in the
left-hand side are assumed to be fresh variables.

Besides, recall that:
a) the rules are given for RIS whose domain is not another RIS (see Appendix
\ref{ap:domris} for further details);
b) the control term of RIS terms is assumed to be a
\emph{variable} in all cases (see Appendix \ref{ap:ctrlris} for
further details);
c) $\Fpv(d)$, $\Gpv(d)$, $\Ppv(d)$ and $\Qpv(d)$ are shorthands for
$\Fpv[x \mapsto d]$,
$\Gpv[x \mapsto d]$, $\Ppv[x \mapsto d]$ and $\Qpv[x \mapsto d]$, respectively,
where $[x \mapsto d]$ represents variable substitution; and
d) we use $=$ and $\in$ in place of $=_\mathcal{X}$ and $\in_\mathcal{X}$
whenever it is clear from the context.

\subsection*{Equality}

\setcounter{equation}{0}
\renewcommand{\theequation}{$=_{\arabic{equation}}$}
\begin{gather}
\varnothing = \varnothing  \lfun \true \label{e=e} 
\end{gather}
\begin{gather}
 X = X  \lfun \true
 \label{x=x} 
\end{gather}
\begin{gather}
\text{If $S \equiv X$ or $S \equiv \{c:\bi{X} | \F @ \P\}$ with $c \equiv \P$:} \notag \\
 X = \{t_0,\dots, t_k \plus S\} \lfun X = \{t_0,\dots, t_k \plus S[X \mapsto N]\}
 \label{special2} 
\end{gather}
%
%
\begin{gather}
\text{If $X$ occurs in other constraints in the input formula and
 $A \not \equiv \defrisinit$:} \notag\\
 X = A  \lfun  X = A \text{ and substitute $X$ by $A$ in the rest of the formula}
 \label{x=s} 
\end{gather}
\begin{gather}
 \varnothing = \{t \plus A\} \lfun \false
 \label{e=s} 
\end{gather}
\begin{gather}
\text{If ($R \equiv X$ and $S \equiv X$) or ($R \equiv
\{c:\bi{X} | \F @ \P\}$ and $S \equiv X$) or } \notag \\
\text{($R \equiv \{c:\bi{X} | \F @ \P\}$ and $S \equiv
\{d:X | \G @ \Q\}$) with $c \equiv \P$ and $d \equiv \Q$:} \notag \\
\begin{split}
 \{t_0, & \dots, t_m \plus R\} =
   \{s_0, \dots, s_k \plus S\} \lfun \\
  & t_0 = s_j \\
  &\quad{}\land \{t_0 \plus N_1\} = R
          \land t_0 \notin N_1
          \land \{t_0 \plus N_2\} = S
          \land t_0 \notin N_2 \\
  &\quad{}\land \{t_1, \dots, t_m \plus N_1\}
                = \{s_0, \dots, s_{j-1},s_{j+1}, \dots, s_k
                    \plus N_2\} \\
  &{}\lor t_0 = s_j \\
  &\quad{}\land t_0 \notin R
          \land t_0 \notin S \\
  &\quad{}\land \{t_1, \dots, t_m \plus R\}
                = \{s_0, \dots, s_{j-1},s_{j+1}, \dots, s_k
                    \plus S\} \\
  &{}\lor t_0 = s_j \\
  &\quad{}\land \{t_0 \plus N\} = R
          \land t_0 \notin N
          \land t_0 \notin S \\
  &\quad{}\land \{t_1, \dots, t_m \plus N\}
                = \{s_0, \dots, s_{j-1},s_{j+1}, \dots, s_k
                    \plus S\} \\
  &{}\lor t_0 = s_j \\
  &\quad{}\land t_0 \notin R
          \land \{s_j \plus N\} = S
          \land s_j \notin N \\
  &\quad{}\land \{t_1, \dots, t_m \plus R\}
                = \{s_0, \dots, s_{j-1},s_{j+1}, \dots, s_k
                    \plus N\} \\
  &{}\lor t_0 = s_j
          \land \{t_1, \dots, t_m \plus R\}
                = \{s_0, \dots, s_k \plus S\} \\
  &{}\lor t_0 = s_j
          \land \{t_0, \dots, t_m \plus R\}
                = \{s_0, \dots, s_{j-1},s_{j+1}, \dots, s_k
                    \plus S\} \\
  &{}\lor X = \{t_0 \plus N\} \\
  &{}\quad{}\land \text{\textbf{if} $S \equiv X$
                        \textbf{then} $\true$ \textbf{else} $\G(t_0)$} \\
  &{}\quad{}\land \{t_1, \dots, t_m \plus R[X \mapsto N]\}
                  = \{s_0, \dots, s_k \plus S[X \mapsto N]\}
 \end{split}
 \label{special3} 
\end{gather}
\begin{gather}
\text{If $R \equiv X$ and $S \equiv \{d:X | \G @ \Q\}$ with $d \equiv \Q$:} \notag \\
 \begin{split}
 \{t_0, \dots, t_m \plus R\} = \{s_0, \dots, s_k \plus S\} \lfun
        \{s_0, \dots, s_k \plus S\} = \{t_0, \dots, t_m \plus R\}
 \end{split}
\label{special5}
\end{gather}

\begin{gather}
\text{If $R \equiv \{c:\bi{X} | \F @ \P\}$ and $S \equiv \{d:X | \G @ \Q\}$
with $c \not\equiv \P$ and $d \not\equiv \Q$}: \notag \\
 \begin{split}
 \{t_0, & \dots, t_m \plus R\} =
   \{s_0, \dots, s_k \plus S\} \lfun \\
  & t_0 = s_j \\
  &\quad{}\land \{t_0 \plus N_1\} = R
          \land t_0 \notin N_1
          \land \{t_0 \plus N_2\} = S
          \land t_0 \notin N_2 \\
  &\quad{}\land \{t_1, \dots, t_m \plus N_1\}
                = \{s_0, \dots, s_{j-1},s_{j+1}, \dots, s_k
                    \plus N_2\} \\
  &{}\lor t_0 = s_j \\
  &\quad{}\land t_0 \notin R
          \land t_0 \notin S \\
  &\quad{}\land \{t_1, \dots, t_m \plus R\}
                = \{s_0, \dots, s_{j-1},s_{j+1}, \dots, s_k
                    \plus S\} \\
  &{}\lor t_0 = s_j \\
  &\quad{}\land \{t_0 \plus N\} = R
          \land t_0 \notin N
          \land t_0 \notin S \\
  &\quad{}\land \{t_1, \dots, t_m \plus N\}
                = \{s_0, \dots, s_{j-1},s_{j+1}, \dots, s_k
                    \plus S\} \\
  &{}\lor t_0 = s_j \\
  &\quad{}\land t_0 \notin R
          \land \{s_j \plus N\} = S
          \land s_j \notin N \\
  &\quad{}\land \{t_1, \dots, t_m \plus R\}
                = \{s_0, \dots, s_{j-1},s_{j+1}, \dots, s_k
                    \plus N\} \\
  &{}\lor t_0 = s_j
          \land \{t_1, \dots, t_m \plus R\}
                = \{s_0, \dots, s_k \plus S\} \\
  &{}\lor t_0 = s_j
          \land \{t_0, \dots, t_m \plus R\}
                = \{s_0, \dots, s_{j-1},s_{j+1}, \dots, s_k
                    \plus S\} \\
  &{}\lor X = \{n \plus N\} \land \Gpv(n) \land t_0 = \Qpv(n) \\
          &\quad{}\land (\neg\F(n)
                         \lor (\F(n)
                              \land \P(n) \notin \{s_0, \dots,s_k\} \land t_0 = \P(n))) \\
          &\quad{}\land \{t_1, \dots, t_m \plus R[X \mapsto N]\}
          = \{s_0, \dots, s_k \plus S[X \mapsto N]\}\} \\
  &{}\lor X = \{n \plus N\} \land \Gpv(n) \land t_0 = \Qpv(n)
          \land \F(n) \land \P(n) = s_j \\
          &\quad{}\land \{\P(n),t_1, \dots, t_m \plus R[X \mapsto N]\}
          = \{s_0, \dots, s_k \plus S[X \mapsto N]\}\} \\
 \end{split} \label{special6} 
\end{gather}
\begin{gather}
 \begin{split}
 \{t \plus & A\} = \{s \plus B\} \lfun \\
       & (t = s \land A = \{s \plus B\})  \lor (t = s \land \{s \plus A\} = B) \\
       & \lor (t = s \land A = B)
       \lor (A = \{s \plus N\} \land \{t \plus N\} = B)
 \end{split}
 \label{s1=s2} 
\end{gather}
\begin{gather}
 \text{\emph{(rule \eqref{e:empty2} of Fig. \ref{f:eq})}} \notag \\
 \{\{t \plus D\} | \F @ \P\} = \varnothing \lfun
 \lnot \Fpv(t) \land \{D | \F @ \P\} = \e
 \label{app.e:empty2} 
\end{gather}
\begin{gather}
 \text{\emph{(rule \eqref{e:eset} of Fig. \ref{f:eq})}} \notag \\
 \begin{split}
 & \{\{t \plus D\} | \F @ \P\} = B \lfun \\
 & (\Fpv(t) \land \{\Ppv(t) \plus \{D | \F @ \P\}\} = B)
 \lor (\lnot \Fpv(t) \land \{D | \F @ \P\} = B)
 \end{split}
 \label{app.e:eset} 
\end{gather}
\begin{gather}
\text{If $S \equiv X$ or $S \equiv \{c:\bi{X} | \F @ \P\}$ with
  $c \equiv \P$ and $d \equiv \Q$:} \notag \\
 \begin{split}
 \{d: & X  | \G @ \Q\} = \{t_0,t_1,\dots,t_k \plus S\} \lfun \\
      & X = \{t_0 \plus N\} \land \Gpv(t_0)
        \land \{d: N | \G @ \Q\}
              = \{t_1,\dots,t_k \plus S[X \mapsto N]\}
 \end{split}
 \label{e:v=e2}
\end{gather}
\begin{gather}
 \text{If $S \equiv \{c:\bi{X} | \F @ \P\}$ with
  $c \not\equiv \P$ and $d \not\equiv \Q$:} \notag \\
 \begin{split}
 \{d: & X  | \G @ \Q\} = \{t_0,t_1,\dots,t_k \plus S\} \lfun \\
      & X = \{n \plus N\} \land \Gpv(n) \land t_0 = \Qpv(n)
        \land ( \neg\F(n) \lor t_0 = \P(n)) \\
      & \land \{d:N | \G @ \Q\}
              = \{t_1,\dots,t_k \plus S[X \mapsto N]\}
 \end{split}
 \label{special4} 
\end{gather}
\begin{gather}
\text{\emph{(rule \eqref{e:v=e} of Fig. \ref{f:eq})}} \notag \\
 \begin{split}
 & \{\bar D | \F @ \P\} = \{t \plus A\} \lfun \\
 & \bar D = \{n \plus N\} \land \flt(n) \land t = \ptt(n) \land \{N
 | \F @ \P\} = A
 \end{split}
 \label{app.e:v=e}
\end{gather}

There is a symmetric rule for each of the following: $\eqref{special2}$,
$\eqref{x=s}$, $\eqref{e=s}$, \eqref{app.e:empty2},
\eqref{app.e:eset}, \eqref{e:v=e2}, \eqref{special4} and $\eqref{app.e:v=e}$.
That is, these rules apply when the l.h.s. and the r.h.s. are switched.

All other admissible $=$-constraints
are in solved form; hence, they are dealt with as irreducible constraints (see
Section \ref{solved}).

\begin{remark}
\begin{itemize}
\item Rules \eqref{e=e} and \eqref{e=s} include rules
\eqref{e:empty1} and \eqref{e:false} of Figure \ref{f:eq} as special cases.
\item The condition $A \not \equiv \defrisinit$ in rule \eqref{x=s} is
motivated by the fact that the case where $A \equiv \defrisinit$ is dealt with
by rule \eqref{app.e:eset}.
\item The set part of an extensional set can be also a RIS term. All rules listed in
this section still continue to work also in these cases. In particular, rule
\eqref{special2} deals also with constraints of the form $X =
\{t_0,\dots, t_n \plus \{X | \F @ \P\}\}$ where the domain of the RIS is the
same variable occurring in the left-hand side of the equality. This constraint
is rewritten to $X = \{t_0,\dots, t_n \plus \{N | \F @ \P\}\}$ where $N$ is a
fresh variable. For example, the equality $X = \{a,b \plus \{x:X | \true @
x\}\}$ is rewritten to $X = \{a,b \plus \{x:N | \true @ x\}\}$. Note that, if
$N = \e$, then $X = \{a,b\}$, which is clearly a solution of the given
constraint.
\item Rules \eqref{e:v=e2} and \eqref{special4} (resp., \eqref{special3},
\eqref{special5} and \eqref{special6}) are motivated by the observation
that rule \eqref{app.e:v=e} (resp., \eqref{s1=s2}) does not work
satisfactory (it loops forever) whenever the same variable $X$ occurs in both
sides of the equation. As an example, the rewriting of the simple constraint
$\{x : D @ x\} = \{1 \plus D\}$ does not terminate using rule
\eqref{app.e:v=e}, though it has the obvious solution $D = \{1 \plus N\}$, $N$
a fresh variable. Thus, rules \eqref{e:v=e2} and \eqref{special4} (resp.,
\eqref{special3}, \eqref{special5} and \eqref{special6}) are introduced
to deal with these special cases.
Note that these rules are, in a sense, the analogous of rule
\eqref{special2}
which deals with equations of the form $X = \{t_0,\dots, t_n \plus
\textbf{\textit{X}}\}$ where the same variable $X$ occurs in both sides of the
equation. In turn, rule \eqref{special2} is a generalization to RIS of rule (6)
listed in Figure 3 of \cite{Dovier00}.
\qed
\end{itemize}
\end{remark}

\subsection*{Inequality}

\setcounter{equation}{0}
\renewcommand{\theequation}{$\neq_{\arabic{equation}}$}
\begin{gather}
\varnothing \neq \varnothing  \lfun \false \label{e neq e} 
\end{gather}
\begin{gather}
X \neq X  \lfun \false \label{x neq x} 
\end{gather}
\begin{gather}
\text{If $A$ is neither a variable nor a RIS:} \notag\\
A \neq X \lfun X \neq A \label{a neq x} 
\end{gather}
\begin{gather}
\varnothing \neq \{t \plus A\} \lfun \true \label{e_neq_s} 
\end{gather}
%
%
\begin{gather}
\text{If $S \equiv X$ or 
$S \equiv \{D | \F @ \P\}$:} \notag \\
\begin{split}
X \neq & \{t_0, \dots, t_n \plus S\}  \lfun \\
   & (n \in X \land n \notin \{t_0, \dots, t_n \plus S\})
    \lor (n \notin X \land n \in \{t_0, \dots, t_n \plus S\})
\end{split} \label{neq-special1}
\end{gather}
\begin{gather}
\begin{split}
\{t \plus &  A \} \neq \{s \plus B\}  \lfun \\
          & (n \in \{t \plus A\} \land n \notin \{s \plus B\})
                    \lor (n \notin \{t \plus A\} \land n \in \{s \plus B\})
\end{split} \label{s-neq-s}
\end{gather}
\begin{gather}
\text{\emph{(rule \eqref{e:neq} of Fig. \ref{f:eq})}} \notag \\
\begin{split}
\{D | & \F @ \P\} \neq A  \lfun \\
   & (n \in \{D | \F @ \P\} \land n \notin A)
    \lor (n \notin \{D | \F @ \P\} \land n \in A)
\end{split}
\label{app.e:neq}
\end{gather}

There is a symmetric rule for each of the following: $\eqref{e_neq_s}$ and
$\eqref{app.e:neq}$. That is, these rules apply when the l.h.s. and the r.h.s.
are switched.

All other admissible $\neq$-constraints are in solved form (see Section
\ref{solved}).

\begin{remark}
The fact that $X \neq \{D | \F @ \P\}$ (rule \eqref{app.e:neq})
and $X \neq \{t_0, \dots, t_n \plus \{D | \F @ \P\}\}$ (rule
\eqref{neq-special1}) are not considered in solved form
as it is $X \neq S$ when $S$ is not a RIS term is motivated by the observation
that, while determining the satisfiability of $X \neq S$ is immediate, the
satisfiability of of the inequalities
involving RIS depends on $D$, $\F$ and $\P$, and hence requires further
simplification of the constraint. \qed
\end{remark}

\subsection*{Set membership}

\setcounter{equation}{0}
\renewcommand{\theequation}{$\in_{\arabic{equation}}$}
\begin{gather}
t \in \varnothing \lfun \false \label{t_in_e}
\\[2mm]
t \in \{s \plus A\} \lfun t = s \lor t \in A \label{t_in_s}
\\[2mm]
t \in X \lfun X = \{t \plus N\}
\label{t_in_X}
\\[2mm]
\text{\emph{(rule \eqref{in:V} of Fig. \ref{f:in})}} \notag \\
t \in \{ D | \F @ \P\} \lfun
  n \in  D \land \F(n) \land t = \P(n)
  \label{app.in:V}
\end{gather}

No other admissible $\in$-constraint.

\subsection*{Not set membership}

\setcounter{equation}{0}
\renewcommand{\theequation}{$\notin_{\arabic{equation}}$}
\begin{gather}
t \notin \varnothing  \lfun \true \label{t_nin_e} \\[2mm]
t \notin \{s \plus A\} \lfun t \neq s \land t \notin A \label{t_nin_ext}
\\[2mm]
\text{\emph{(rule \eqref{nin:nV} of Fig. \ref{f:in})}} \notag \\
\begin{split}
& t \notin \{\{d \plus D\} | \F @ \P\} \lfun \\
& (\Fdd \land t \neq \Pdd \land t \notin \{D | \F @ \P\}) \lor (\lnot \nFdd
\land t \notin \{D | \F @ \P\})
\end{split}
\label{app.nin:nV}
\end{gather}

All other admissible $\notin$-constraints are in solved form (see Section
\ref{solved}).

\begin{remark}[Membership/not membership]
Rules \eqref{t_in_e} and \eqref{t_nin_e} include rules
\eqref{in:e} and \eqref{nin:e} of Figure \ref{f:in} as special cases.
\qed
\end{remark}

\subsection*{Sort constraints}

\setcounter{equation}{0}
\renewcommand{\theequation}{$\Set_{\arabic{equation}}$}
\begin{gather}
\label{set(e)} \Set(\e) \lfun \true \\[2mm]
\label{set(s)} \Set(\{t \plus A\}) \lfun \Set(A) \\[2mm]
\label{set(ris)} \Set(\{D | \F @ \P\}) \lfun \true \\[2mm]
\label{set(t)} \Set(t) \lfun \false
\end{gather}

\setcounter{equation}{0}
\renewcommand{\theequation}{$\isx_{\arabic{equation}}$}
\begin{gather}
\label{nset(e)} \isx(\e) \lfun \false \\[2mm]
\label{nset(s)} \isx(t \plus A\} \lfun \false \\[2mm]
\label{nset(ris)} \isx(\{D | \F @ \P\}) \lfun \false \\[2mm]
\label{nset(t)} \isx(t) \lfun \true
\\\notag
\end{gather}

\begin{remark}
In Algorithm \ref{glob} (see Section \ref{solver}), $\SATRIS$ calls $\SATX$
only once, at the end of the computation. $\SATX$ is called by passing it the
whole collection of $\Ur$ literals previously accumulated in the current
formula $\Phi$ by the repeated applications of the rewrite rules within
\textsf{STEP}. Alternatively, and more efficiently, $\SATX$ could be called
repeatedly in the inner loop of the solver, just after the \textsf{STEP}
procedure has been called. This would allow possible inconsistencies to be
detected as soon as possible instead of being deferred to the last step of the
decision procedure. For example, if $\Phi$ contains the equation $\{1\} =
\{2\}$, which is rewritten by \textsf{STEP} as $1 =_\mathcal{X} 2$, calling
$\SATX$ just after \textsf{STEP} ends allows the solver to immediately detect
that $\Phi$ is unsatisfiable. Similarly, variable substitutions entailed by
equalities possibly returned by $\SATX$ are propagated to the whole formula
$\Phi$ as soon as possible. \qed
\end{remark}

\setcounter{equation}{0}
\renewcommand{\theequation}{\arabic{equation}}

\subsection{\label{ap:domris}Nested RIS domains}

According to the syntax of $\LRIS$ (see Section \ref{lris}), RIS domains can be
a nested chain of RIS, ending in a variable or an extensional set. On the other
hand, the rewrite rules presented in Section \ref{rules} and Appendix
\ref{ap:rules} apply only to RIS whose domain is not another RIS. We do so
because the rewrite rules for the most general case are more complex, thus they
would hinder understanding of the decision procedure.

These more general rules, however, can be easily generated from the rules for
the simpler case. In this section we show how this generalization can be done
by showing how one of the rules presented in Section \ref{rules}, namely rule
\eqref{e:eset}, is adapted to deal with the more general case. All other
rewrite rules can be generalized in the same way.

Consider a non-variable RIS whose domain is a nested chain of RIS where the
innermost domain is an extensional set. The generalization  of rule
\eqref{e:eset} to a RIS of this form is a follows:
\begin{gather*}
\{ \{ \dots \{\{d \plus D\} | \F_1 @ \P_1\}\dots | \F_{m-1} @ \P_{m-1}\} | \F_m
@
\P_m\} = B  \lfun \\
\begin{split}
& \{\dots \{\{d \plus D\} | \F_1 @ \P_1\}\dots | \F_{m-1} @ \P_{m-1}\} = \{n \plus N\}\\
& \land (\F(n)
  \land \{\P_m(n) \plus \{N | \F_m @ \P_m\}\} = B \\
& \quad \lor \lnot \F(n) \land \{N | \F_m @ \P_m\} = B)
\end{split}
\end{gather*}
where $n$ and $N$ are two new variables. Note that the first element $n$ of the
domain of the outermost RIS is obtained by the recursive application of the
same rules for equality, over the domain itself (possibly another RIS) and the
extensional set $\{n \plus N\}$. Note also that when $m = 1$ this rule boils
down to rule \eqref{e:eset} of Figure \ref{f:eq}.

\subsection{\label{ap:ctrlris}Control Terms}

As with nested RIS domains, we preferred not to show the rewrite rules when the
control term is not a variable, as these rules are somewhat more complex than
the others.

According to Definition \ref{RIS-terms}, when the control term is not a
variable then it is an ordered pair of the form $(x,y)$ where both components
are variables. Consider the following RIS:
\[
\{(x,y):\{(1,2),55\} | \F @ \P\}
\]
The problem with this RIS is that $(x,y)$ does not unify with $55$, for all $x$
and $y$. The semantics of RIS stipulates that $55$ must not be considered as a
possible value on which evaluate $\F$ and $\P$.

As this example shows, it is necessary to consider one more case (i.e., one
more non-deterministic choice) in each rewriting rule. For example, if in rule
\eqref{e:eset} we consider a general control term $c$, and not just a variable,
the rule is split into two rules:
\begin{gather*}
\begin{split}
& \text{If $c \in \Var$ or $d \in \Var$ or }
  \text{$(c \equiv f(x_1,\dots,x_n)$ and $d \equiv f(t_1,\dots,t_n))$: }\\
&        \{c:\{d \plus D\} | \F @ \P\} = B \lfun \\
& c = d \land \F(d) \land \{\P(d) \plus \{c:D | \F @ \P\}\} = B \\
& \lor c = d \land \lnot \F(d) \land \{c:D | \F @ \P\} = B
\end{split}\\[2mm]
\begin{split}
& \text{If $c \equiv f(x_1,\dots,x_n)$ and $d \equiv g(t_1,\dots,t_m)$ and $(f
\not\equiv g$ or $n \not\equiv m)$: } \\
&        \{c:\{d \plus D\} | \F @ \P\} = B \lfun
  \{c:D | \F @ \P\} = B
\end{split}
\end{gather*}

Note how the second rule simply skips $d$.


\section{Source code of the Bell-LaPadula case study}\label{app:blp}

\begin{verbatim}

:- int_solver(clpq).

dominates([L1,C1],[L2,C2]) :-
  subset(C1,C2) & L1 =< L2.

openRead(ScF,ScP,Proctbl,P,F,Proctbl_) :-
  [P,F] nin Proctbl &
  apply(ScF,F,[Lf,Cf]) &
  apply(ScP,P,[Lp,Cp]) &
  dominates([Lf,Cf],[Lp,Cp]) &
  Proctbl_ = {[P,F]/Proctbl}.

seccond(ScF,ScP,Proctbl) :-
  foreach([Pi,Fi] in Proctbl,
          [Lf1,Cf1,Lp1,Cp1],
          subset(Cf1,Cp1) & Lf1 =< Lp1,
          apply(ScF,Fi,[Lf1,Cf1]) & apply(ScP,Pi,[Lp1,Cp1])
         ).

openReadPreservesSeccond(ScF,ScP,Proctbl) :-
  seccond(ScF,ScP,Proctbl) &
  openRead(ScF,ScP,Proctbl,P,F,Proctbl_) &
  nforeach([Pi,Fi] in Proctbl_,
           [Lf1,Cf1,Lp1,Cp1],
           subset(Cf1,Cp1) & Lf1 =< Lp1,
           apply(ScF,Fi,[Lf1,Cf1]) & apply(ScP,Pi,[Lp1,Cp1])
          ).

\end{verbatim}


\allowdisplaybreaks[1]

\section{Detailed Proofs}\label{ap:proofs}
\catcode`\|=\active
\def|{\mid}

This section contains detailed proofs of the theorems stated in the main text,
along with some results that justify some of our claims. We start with the
justification that RIS can be used to encode restricted universal quantifiers.

\begin{proposition}\label{prop:1}
\begin{gather*}
D \subseteq \{x:D | F\} \iff \forall x (x \in D \implies F)
\end{gather*}
\end{proposition}

\begin{proof}

\noindent$\Longrightarrow$)
\begin{gather*}
x \in D \\
\imp x \in \{x:D | F\} \why{H} \\
\imp x \in D \land F(x) \why{RIS def.} \\
\implies F(x)
\end{gather*}

\noindent$\Longleftarrow$)
\begin{gather*}
x \in D \\
\imp F(x) \why{H} \\
\imp x \in D \land F(x) \\
\imp x \in \{x:D | F\} \why{RIS def.}
\end{gather*}
\end{proof}

The following proposition supports the claim that to force a RIS to be
empty it is enough to consider its filter.

\begin{proposition}\label{prop:2}
If $D$ is a non-empty set, then:
\begin{gather*}
\{x:D | \F @ \P\} =  \e \iff \forall x (x \in D \implies \lnot\F(x))
\end{gather*}
\end{proposition}

\begin{proof}
\begin{gather*}
\{x:D | \F @ \P\} =  \e \\
\iff \{y : \exists d (x \in D \land \F \land y = \P)\} = \e \\
\iff \forall y (\lnot\exists d (d \in D \land \F \land y = \P)) \\
\iff \forall y (\forall d (d \notin D \lor \lnot\F(d) \lor y \neq \P(d))) \\
\iff \forall d (d \in D \implies \lnot\F(d))
\end{gather*}
\end{proof}

The following proposition supports the claim that many RIS parameters can
be avoided by a convenient control term.

\begin{proof}[of Proposition \ref{prop:3}]
\mbox{}

$\Longrightarrow$)
\begin{gather*}
a \in S \\
\iff \exists x,\vp
       (x \in D
        \land \vp \in D_2 \land F(x,\vv,\vp) \land P(x,\vv,\vp) = a)
     \why{H; RIS def.} \\
\iff \exists x,\vp
       ((x,\vp) \in D \times D_2
        \land F(x,\vv,\vp) \land P(x,\vv,\vp) = a)
     \why{$\times$ def.} \\
\iff a \in \{(x,\vp):D_1 \times D_2 | F((x,\vp),\vv) @ P((x,\vp),\vv)\}
     \why{RIS def.}
\end{gather*}

$\Longleftarrow$) Similar to the previous case.
\end{proof}

Now we prove Proposition \ref{negparam} which gives the conditions to eliminate
existential quantifiers appearing in RIS filters in relation to functional
predicates.

\begin{proof}[of Proposition \ref{negparam}]
Note that:
\begin{gather}
\forall y (\lnot (p(x_1,\dots,x_{n-1},y) \land \Fpv_r(\vec x_r,y))) \notag\\
\iff \forall y (\lnot p(x_1,\dots,x_{n-1},y)
                \lor \lnot \Fpv_r(\vec x_r,y)) \label{neg1}
\end{gather}

Now we divide the proof into two implications.

\mbox{}

\noindent$\Longrightarrow$) If $\Fpv_q(\vec x_q)$ does not hold the conclusion
is proved. Now assume $\Fpv_q(\vec x_q)$ holds. Then $p(x_1,\dots,x_{n-1},z)$
holds for some $z$ due to the hypothesis and Definition \ref{prefps}. Hence for
\eqref{neg1} to be true, $\Fpv_r(\vec x_r,z)$ must be false (because otherwise
the disjunction would be true for $z$). So in this case we have $\Fpv_q(\vec
x_q) \land p(x_1,\dots,x_{n-1},z) \land \lnot \Fpv_r(\vec x_r,z)$, which proves
the conclusion.

\mbox{}

\noindent$\Longleftarrow$) If $\Fpv_q(\vec x_q)$ does not hold then
$p(x_1,\dots,x_{n-1},y)$ does not hold for all $y$ due to the hypothesis and
Definition \ref{prefps}. Hence the conclusion.

Now assume $z$ is such that $\Fpv_q(\vec x_q) \land p(x_1,\dots,x_{n-1},z) \land
\lnot \Fpv_r(\vec x_r,z)$ is true. Then the conclusion is true for $z$ because
$\Fpv_r(\vec x_r,z)$ is false. Now consider any $y \neq z$. Given that $p$ is a
functional predicate symbol (hypothesis), then we have $\lnot
p(x_1,\dots,x_{n-1},y)$ because $p$ can hold for at most one value of its last
parameter (and we know it holds for $z$). So the conclusion.
\end{proof}

The next subsections provide the proofs of the theorems stated in the main
text. All these proofs concern the base $\LRIS$ language, not the possible
extensions discussed in Section \ref{sec:extension} nor those presented in
Appendixes \ref{ap:domris} and \ref{ap:ctrlris}.

\subsection{\label{ap:solved}Satisfiability of the solved form (Theorem \ref{satisf})}

Basically, the proof of this theorem uses the fact that, given a pure
$\mathcal{RIS}$-formula $\Phi$ verifying the conditions of the theorem, it is
possible to guarantee the existence of a successful assignment of values to all
variables of $\Phi$ using pure sets only, with the only exception of the
variables $X$ occurring in terms of the form $X = u$---which are obviously
already assigned. In particular, the solved forms involving variable RIS verify
the following ($X_i$ are variables; $R_i$ are variables or variable-RIS whose
domain is variable $X_i$):
\begin{itemize}
\item $t \notin \{X_1 | \flt @ \ptt\}$
\item $\{X_1 | \flt @ \ptt\} = \varnothing$
\item $\{X_1 | \flt_1 @ \ptt_1\} = \{X_2 | \flt_2 @ \ptt_2\}$
\item $\Cup(R_3,R_4,R_5)$
\item $R_3 \disj R_4$
\end{itemize}
are solved with $X_i = \e$ and $R_i = \e$ for those $R_i$ that are variables.

In the proof we use the auxiliary function $\mathit{find}$:
\[
  \mathit{find}(x,t) =
\begin{cases}
        \emptyset & \text{if } t =  \emptyset, x \neq \emptyset \\
        \{0\}     & \text{if $t =  x$}\\
        \{1 + n : n \in \mathit{find}(x,y)\} & \text{if } t =  \{y \plus \e\} \\
        \{1 + n : n \in \mathit{find}(x,y)\} \cup \mathit{find}(x,s) &
                \text{if } t =  \{y \plus s\}, s \neq \e
\end{cases}
\]
which returns the set of `depths' at which a given element $x$ occurs in the
set $t$.

\begin{proof}
Consider a pure $\mathcal{RIS}$-formula $\Phi$ in solved form. The proof is
basically the construction of a mapping for the variables of $\Phi$ of sort
$\sSet$
into the interpretation domain $D_\sSet$ (see Section \ref{ap:satisf2} to see
how variables of sort $\Ur$ are managed). The construction is divided into two
parts by dividing $\Phi$ as $\Phi_= \land \Phi_r$, where $\Phi_=$ is a
conjunction of equalities whose l.h.s is a variable, and $\Phi_r$ is the rest
of $\Phi$. In the first part $\Phi_=$ is not considered.  A solution for
$\Phi_r$ is computed by looking for valuations\footnote{A \emph{valuation}
$\sigma$ of a $\Sigma$-formula $\varphi$ is an assignment of values from the
interpretation domain $D_\sU$ to the free variables of $\varphi$ which respects
the sorts of the variables.} of the form:
\begin{equation}\label{e:ne}
X_{i} \mapsto \underbrace{\{\cdots\{}_{n_{i}} \e
  \}\cdots\}\vspace*{-2ex}
\end{equation}
fulfilling all $\neq$ and $\notin$ constraints. We will briefly refer to the
r.h.s. of \eqref{e:ne} as $\{ \e \}^{n_i}$. In particular, RIS domains are
mapped onto $\e$ ($n_i = 0$) and the numbers $n_{i}$ for the other variables
are computed choosing one possible solution of a system of integer equations
and disequations, that trivially admits solutions. Such system is obtained by
analyzing the `depth' of the occurrences of the variables in the terms. Then,
all the variables occurring in $\Phi$ only in r.h.s. of equations of $\Phi_{=}$
are bound to $\e$ and the mappings for the variables of the l.h.s. are bound to
the  uniquely induced valuation.

In  detail, let $X_1,\dots, X_m$ be all the variables occurring in $\Phi$, save
those occurring in the l.h.s. of equalities, and let $X_1,\dots, X_h$, $h \leq
m$, be those variables occurring as domains of RIS terms. Let $n_1,\dots, n_m$
be auxiliary variables ranging over $\mathbb{N}$. We build the system {\it
Syst} as follows:
\begin{itemize}
\item For all $i\leq h$, add the equation $n_i=0$.

\item For all $h<i\leq m$, add the following disequations:
$$\begin{array}{ll}
n_i \neq n_j+c & \forall X_i \neq t \text{ in } \Phi \text{ and } c \in \mathit{find}(X_j,t)\\
n_i \neq c & \forall X_i \neq t \text{ in } \Phi \text{ and } t \equiv \{\e\}^c\\
n_i \neq n_j+c+1 & \forall t \notin X_i \text{ in } \Phi \text{ and } c \in \mathit{find}(X_j,t) \\
n_i \neq c+1 & \forall t \notin X_i \text{ in } \Phi \text{ and } t \equiv
\{\e\}^c
\end{array}$$
\end{itemize}

If $m = h$, then $n_i = 0$ for all $i=1,\dots,m$ is the unique solution of {\it
Syst}. Otherwise, it is easy to observe that {\it Syst} admits
infinitely many solutions. Let:
\begin{itemize}
\item $\{n_1=0,\dots, n_h=0, n_{h+1}=\bar{n}_{h+1},\dots,
n_m=\bar{n}_m\}$  be one arbitrarily chosen solution of {\it Syst}.
\item
$\theta$ be the valuation such that $\theta(X_i)=\{\e\}^{n_i}$ for all $i\leq
m$.
\item
$Y_1,\dots, Y_k$ be all the variables of $\Phi$ which appear only in the l.h.s.
of equalities of the form $Y_i=t_i$.
\item $\sigma$ be the valuation such that $\sigma(Y_i)=\theta(t_i)$.
\end{itemize}

We prove that $\mathcal{R} \models \Phi[\theta\sigma]$ by case analysis on the
form of the atoms in $\Phi$:

\begin{itemize}
\item $Y_i=t_i$ \hspace{2mm} It is satisfied, since $\sigma(Y_i)$ has been
  defined as a ground term and equal to $\theta(t_i)$.

\item $X_i\neq t$ \hspace{2mm} If $t$ is a ground term, then we have two cases:
if $t$ is not of the form $\{\e\}^c$, then it is immediate that
$\theta(X_i)\neq t$; if $t$ is of the form $\{\e\}^c$, for some $c$, then we
have $n_i\neq c$, by construction, and hence $\theta(X_i)\neq t$.

If $t$ is not ground, then if $\theta(X_i)=\theta(t)$, then there exists
 a variable $X_j$ in $t$ such that
$\bar{n}_i=\bar{n}_j+c$ for some $c\in \mathit{find}(X_j,t)$; this cannot be the case
         since we started from a solution of $Syst$.

\item $t \not\in X_i$ \hspace{2mm} Similar to the case above.

\item $\{X_i | \flt @ \ptt\} = \varnothing$ \hspace{2mm} This means that
$\bar{n}_i=0$ and $\theta(X_i)=\theta(X_j)=\theta(X_k)=\e$.

\item $\{X_i | \flt_1 @ \ptt_1\} = \{X_j | \flt_2 @ \ptt_2\}$ \hspace{2mm} This means that
$\bar{n}_i=\bar{n}_j=0$ and $\theta(X_i)=\theta(X_j)=\e$.

\item
$\Cup(R_i,R_j,R_k)$ \hspace{2mm} Recall that $R_i,R_j,R_k$ can be either
variables or variable-RIS whose domain are variables $X_i,X_j,X_k$,
respectively. Now, those $R$ that are variables are replaced by a corresponding
$X$ and for those that are not the valuation is computed for the domains (which
are variables, by construction). For example, if $R_i$ and $R_k$ are
variable-RIS and $R_j$ is a variable, we have $\Cup(R_i,X_j,R_k)$, $X_i$ is the
domain of $R_i$, $X_k$ is the domain of $R_k$ and the valuation is computed for
$X_i$, $X_j$ and $X_k$. Also recall that from item \ref{i:solvedneq} of
Definition \ref{def:solved}, there are no $\neq$ constraints involving any of
the $X$ participating of the $\Cup$ constraint. Then, this means that
$\bar{n}_i=\bar{n}_j=\bar{n}_k=0$ and $\theta(X_i)=\theta(X_j)=\theta(X_k)=\e$.

\item $R_i \disj R_j$ \hspace{2mm} Similar considerations for $R_i$ and $R_j$
to the previous case apply. Then:
\begin{itemize}
\item If $i,j \leq h$, then $\theta(X_i)=\theta(X_j)=\e$
\item If $i > h$ (the same if $j > h$), then $\bar{n}_i \neq \bar{n}_j$ and
so $\theta(X_i)=\{ \e \}^{n_i}$ is disjoint from $\theta(X_j)=\{ \e \}^{n_j}$.
\end{itemize}
\end{itemize}
\end{proof}

\subsection{\label{ap:satisf2}Satisfiability of $\Phi_\mathcal{S} \land \Phi_\Ur$
(Theorem \ref{satisf2})}

The satisfiability of $\Phi_\Ur$ is determined by $\SATX$. Since $\SATX$ is a
decision procedure for $\Ur$-formulas, if $\Phi_\Ur$ is unsatisfiable, then
$\SATX$ returns $\false$; hence, $\Phi$ is unsatisfiable. If $\Phi_\Ur$ is
satisfiable, then $\SATX$ rewrites $\Phi_\Ur$ into an $\Ur$-formula in a
simplified form which is guaranteed to be satisfiable w.r.t. the interpretation
structure of $\LX$. This rewriting, however, may cause non-set variables
in $\Phi_\Ur$, i.e., variables of sort $\sU$, to get values for which the
formula is satisfied. These variables can occur in both $\Phi_\Ur$ and
$\Phi_\mathcal{S}$. Given that, at this point, $\Phi_\mathcal{S}$ is in solved
form, variables of sort $\sU$ can only appear in
constraints that are in solved form. Specifically, if $x$ is a variable of sort $\sU$, the following
are all solved form constraints that may contain $x$:
\begin{enumerate}
\item $X = S(x)$ or $X = \{Y | \flt(x) @ \ptt(x)\}$ (i.e., $x$ is a free
variable in the RIS)
\item $\{X | \flt(x) @ \ptt(x)\} = \varnothing$ or $\varnothing = \{X | \flt(x) @
\ptt(x)\}$
\item $\{X | \flt_1(x) @ \ptt_1(x)\} = \{Y | \flt_2(x) @ \ptt_2(x)\}$.
\item $X \neq S(x)$
\item $t(x) \notin \bar D$
\item $\Cup(\bar C,\bar D,\bar E)$ and $\bar C$ or $\bar D$ or $\bar E$
are of the form $\{X | \flt(x) @ \ptt(x)\}$
\item $\bar C \disj \bar D$, and $\bar C$ or $\bar D$
are of the form $\{X | \flt(x) @ \ptt(x)\}$
\end{enumerate}
All these constraints remains in solved form regardless of the value bound to
$x$. Hence, $\Phi$ is satisfiable.

\subsection{\label{ap:termination}Termination of $\SATRIS$ (Theorem \ref{termination-glob})}

In order to prove termination of $\SATRIS$ we use the rules give in Appendix
\ref{ap:rules} for equality, inequality, set membership, and not set
membership, and those for union and disjointness given in Figures \ref{f:un}
and \ref{f:disj} in the main text.

First of all, it is worth noting that the requirement that the set of variables
ranging on $\mathcal{RIS}$-terms (i.e., variables of sort $\sSet$) and the set
of variables ranging on $\Ur$-terms (i.e., variables of sort $\sU$) are
disjoint sets prevents us from creating recursively defined RIS, which could
compromise the finiteness property of the sets we are dealing with. In fact, a
formula such as $X = \{D | F(X) @ P\}$, where $F$
contains the variable $X$, is not an admissible $\mathcal{RIS}$-constraint,
since the outer $X$ should be of sort $\sSet$ whereas the inner $X$ should be
of sort $\sU$ (recall that the filter is a $\Ur$-formula). Note that, on the
contrary, a formula such as $X = \{D(X) | F @ P\}$ is an admissible formula,
which, in many cases, is suitably handled by our decision procedure.

Let a \emph{rewriting procedure for $\pi$} be the repeated application of the rewrite rules for a specific
$\mathcal{RIS}$-constraint $\pi$ until either the initial formula becomes
$\false$ or no rules for $\pi$ apply. Following \cite{Dovier00}, we begin by
proving that each individual rewriting procedure, applied to an admissible
formula, is \emph{locally terminating}, that is each call to such procedures
will stop in a finite number of steps. For all the rules inherited from \CLPSET
we assume the results in \cite{Dovier00}. Then we prove local termination only
for the new rules dealing with RIS.

The following are non-recursive rules thus they terminate trivially:
\eqref{e=e}, \eqref{x=x}, \eqref{special2}, \eqref{x=s}, \eqref{e=s},
\eqref{special5}, \eqref{e neq e}, \eqref{x neq x}, \eqref{a neq x},
\eqref{e_neq_s}, \eqref{neq-special1}, \eqref{s-neq-s}, \eqref{app.e:neq},
\eqref{t_in_e}, \eqref{t_in_X}, \eqref{app.in:V}, \eqref{t_nin_e},
\eqref{un:equalvars}, \eqref{un:empty}, \eqref{un:empty2}, \eqref{un:empty3},
\eqref{disj:id}, \eqref{disj:rempty} and \eqref{disj:lempty}.

Now we consider rules which contain direct recursive calls or calls to other
rules of the same rewriting procedure and involve at least one RIS.  Figures
\ref{f:eq-calls}-\ref{f:disj-calls} show what rewrite rule calls other rewrite rules of the same rewriting procedure, for $=$,
$\Cup$ and $\disj$, respectively. We only depict these because the other
rewriting procedures are simpler. We will pay special attention to the possible
loops that can be seen in all three graphs.

\begin{figure}
\begin{center}
\tikz [
    nonterminal/.style={rectangle},
    >=Latex, arrows={[round]}] {
\node [nonterminal] (10) at (0,0) {\eqref{s1=s2}}; 
\node [nonterminal,draw] (3) at (1,2) {\eqref{special2}};
\node [nonterminal,draw] (1314) at (-2,2) {\eqref{e:v=e2} \eqref{special4}};
\node [nonterminal] (12) at (-4,0) {\eqref{app.e:eset}};
\node [nonterminal] (15) at (0,-2) {\eqref{app.e:v=e}};
\node [nonterminal,draw] (68) at (-7,0) {\eqref{special3}\eqref{special6}};
\draw (10) edge[->] (3)
      (10) edge[->] (1314)
      (10) edge[<-]  (12)
      (10) edge[->] (15)
      (12) edge[->] (1314)
      (15) edge[<-,bend left] (12)
      (68) edge[->] (12)
      ;
}
\end{center}
\caption{\label{f:eq-calls}Calls relation between rules of the $=$-rewriting
procedure}
\end{figure}

\begin{figure}
\begin{center}
\tikz [
    nonterminal/.style={rectangle},
    >=Latex, arrows={[round]}] {
\node [nonterminal] (19) at (0,0) {\eqref{un:ext}};
\node [nonterminal] (17) at (-1,1) {\eqref{un:ext1}};
\node [nonterminal] (18) at (1,1) {\eqref{un:ext2}};
\node [nonterminal] (20) at (0,3) {\eqref{un:ris}};
\draw (-1.5,-0.5) rectangle (1.5,1.5);
\draw (17) edge[<->] (18)
      (19) edge[->] (17)
      (19) edge[->] (18)
      (20) edge[->,bend left] (0.1,1.5)
      (-0.1,1.5) edge[->,bend left] node[pos=0.5,left] {{\footnotesize ($\Cup$-a)}} (20)
      ;
}
\end{center}
\caption{\label{f:un-calls}Calls relation between rules of the $\Cup$-rewriting procedure}
\end{figure}

\begin{figure}
\begin{center}
\tikz [
    nonterminal/.style={rectangle},
    >=Latex, arrows={[round]}] {
\node [nonterminal] (24) at (0,0) {\eqref{disj:risext}};
\node [nonterminal] (25) at (0,2) {\eqref{disj:extris}};
\node [nonterminal] (26) at (2,0) {\eqref{disj:rris}};
\node [nonterminal] (27) at (2,2) {\eqref{disj:rrisSym}};
\draw (24) edge[<->] (25)
      (24) edge[<->] (26)
      (24) edge[<->] (27)
      (25) edge[<->] (26)
      (25) edge[<->] (27)
      (26) edge[<->] (27)
      ;
}
\end{center}
\caption{\label{f:disj-calls}Calls relation between rules of the $\disj$-rewriting procedure}
\end{figure}

\begin{description}
%
\item[Rule \eqref{special3}] In all branches, the recursive call is
made with at least one argument whose size is strictly smaller than the one of
the input formula. In particular, in the last branch the size of $X$ will not
affect the size of the arguments of the recursive call because $X$ is a
variable so it cannot add elements to $N$, $N$ is a variable itself and $X$ is
removed from the recursive call.
\item[Rule \eqref{special6}] In the first seven branches, the recursive call is
made with at least one argument whose size is strictly smaller than the one of
the input formula. In the last branch, the recursive call is made with
arguments whose size is equal to those of the input formula. Indeed, we remove
$t_0$ from the l.h.s. set but we add $\P(n)$ while variable $X$ is substituted
by the new variable $N$. However, we now know that $\P(n) = s_j$ for some $j
\in [0,k]$. Hence, the recursive call will be processed by one of the first
seven branches which we know reduce the size of at least one of their
arguments.
\item[Rule \eqref{s1=s2}] We will prove termination of this rule by
describing a particular abstract implementation.
\begin{enumerate}
  \item If the set parts of $A$ and $B$ are not RIS, then the results of
  \cite{Dovier00} apply.
  \item\label{i:empty} If the set part of either $A$ or $B$ are RIS, then:
  \begin{enumerate}
    \item `Empty' their domains. This means that if one of the RIS is of the form:
    \[
    \riss{\{d_1,\dots,d_k \plus D\}}{\F}{\P}
    \]
    then rewrite this set into:
    \[
    \{\P(d_1) \plus \riss{\{d_2,\dots,d_k \plus D\}}{\F}{\P}\}
    \]
    or
    \[
    \riss{\{d_2,\dots,d_k \plus D\}}{\F}{\P}
    \]
    depending on whether $\F(d_1)$ holds or not. Continue with this rewriting
    until all the $d_i$ have been processed. This process will transform the
    initial RIS into a number of sets of the form:
    \begin{equation}\label{eq:emptieddomain}
    \{\P(d_{i_1}),\dots,\P(d_{i_m}) \plus \riss{D}{\F}{\P}\}
    \end{equation}
    with $\{d_{i_1},\dots,d_{i_m}\} \subseteq \{d_1,\dots,d_k\}$ and $D$
    variable or the empty set.

    Hence, at the end of this process the domains of the RIS so generated
    are either the empty set or a variable.
    \item If the domain of one of the RIS obtained in the previous step
    is the empty set then substitute the RIS by the empty set. Then RIS of
    the form \eqref{eq:emptieddomain} become
    $\{\P(d_{i_1}),\dots,\P(d_{i_m}) \plus \e\}$ which is equal to
    $\{\P(d_{i_1}),\dots,\P(d_{i_m})\}$.

    Hence, at the end of this process the domains of the RIS so generated
    are variables.
    \item Substitute the RIS by new variables. Formally (when at both sides
    there are RIS):
  \begin{gather*}
  \{t_0,\dots,t_m \plus \riss{D}{\F}{\P}\}
    = \{s_0,\dots,s_k \plus \riss{E}{\G}{\Q}\} \\
  \lfun \\
  \{t_0,\dots,t_m \plus N_1\}
    = \{s_0,\dots,s_k \plus N_2\} \tag{a} \label{eq:a}\\
  {}\land N_1 = \riss{D}{\F}{\P} \land N_2 = \riss{E}{\G}{\Q}
  \end{gather*}
  Note that might be no more RIS at this point, so there is nothing to substitute.
  \end{enumerate}
  \item Now apply rule \eqref{s1=s2} only to the first conjunct of \eqref{eq:a} as in
  \cite{Dovier00} until it terminates, delaying the processing of the remaining
  conjuncts in \eqref{eq:a}. As no RIS are in the first conjunct of \eqref{eq:a}
  the results of \cite{Dovier00} apply.
  \item Once rule \eqref{s1=s2} applied to the first conjunct terminates,
  variables $N_i$ are substituted back by the corresponding RIS, thus obtaining
  equalities of the following forms:
  \begin{gather*}
  \riss{D}{\F}{\P} = \riss{E}{\G}{\Q} \\
  \riss{D}{\F}{\P} = \{\cdot \plus N\}\\
  \riss{D}{\F}{\P} = \{\cdot \plus \riss{E}{\G}{\Q}\}
  \end{gather*}
  where $D$ and $E$ are variables (due to the process described in step \ref{i:empty}).

  Note that, without RIS, the final formulas returned by \eqref{s1=s2} would be
  in solved form as their l.h.s. would be variables.
  \item As can be seen, by the form of the returned equalities and by Figure
  \ref{f:eq-calls}, these equalities are processed by some of the rules added
  to deal with RIS.
\end{enumerate}

So \eqref{s1=s2} either terminates as in \cite{Dovier00}, or calls other rules
for equality (which is an added behavior compared to \cite{Dovier00}).

%
\item[Rule \eqref{app.e:empty2}] The recursive call is made with one
argument whose size is strictly smaller than in the initial formula, since $t$
is removed from the domain of the RIS.
\item[Rule \eqref{app.e:eset}] In the first branch there is no recursion
on the same rule because the l.h.s. of the equality constraint is no longer a
RIS but an extensional set. In particular, the size of the arguments are the
same w.r.t. the initial ones. However, the rules that can be fired in this case
(see Figure \ref{f:eq-calls} and the analysis for the corresponding rules), all
reduces the size of at least one of its arguments. In particular, if rules
\eqref{s1=s2} or \eqref{app.e:v=e} are called, if they call this rule back, it
will be done with strictly smaller arguments. Then, this loop will run a finite
number of times.

In the second branch the recursive call is made with one argument whose size is
strictly smaller than in the initial formula, since $t$ is removed from the
domain of the RIS.
\item[Rule \eqref{e:v=e2}] The size of the r.h.s. argument of the equality
constraint is strictly smaller than the one of the input formula because $t_0$
has been removed. Furthermore, the size of $X$ will not affect the size of the
arguments of the recursive call because $X$ is a variable so it cannot add
elements to $N$, $N$ is a variable itself and $X$ is removed from the recursive
call.
\item[Rule \eqref{special4}] Same arguments of previous rule apply.
%
%
\item[Rule \eqref{app.e:v=e}] The recursive call is made with one
argument whose size is strictly smaller than in the initial formula, since $t$
is removed from the r.h.s. The size of $D$ can affect the size of the arguments
of the recursive call only if $D$ is the set part of $A$. However, this case is
processed by one of the following rules: \eqref{e:v=e2} or \eqref{special4}.

\item[Rule \eqref{t_in_s}] The recursive call is made with a r.h.s.
argument whose size is strictly smaller than the initial one because $s$ is
removed from the set.

\item[Rule \eqref{t_nin_ext}] Same arguments of previous rule apply.
\item[Rule \eqref{app.nin:nV}] In either branch the recursive call
is made with a r.h.s. argument whose size is strictly smaller than the initial
one because $d$ is removed from the domain of the RIS.

\item[Rule \eqref{un:ext1}] In the first branch the recursive call
is made with a first argument whose size is strictly
smaller than the initial one because $t$ is removed from it. In the second
branch, it is both arguments whose sizes are strictly
smaller than the initial ones because $t$ is effectively removed from $A$ given
that we know that $t \in A$.

Note that in Figure \ref{f:un-calls} this rule may call rule \eqref{un:ris},
which in turn may call back this rule or rules \eqref{un:ext2} and
\eqref{un:ext}, thus possibly generating an infinite loop. However, this rule
can call \eqref{un:ris} only when the first or second
argument is a non-variable RIS. In that case, rule \eqref{un:ris} removes
one element of the domain of the RIS but it does not necessarily reduces the
size of the arguments. This may fire the callback to
this rule. But this rule reduces the size of its
arguments, so after a finite number of calls the loop finishes. This
shows how the arrow labeled ($\Cup$-a) in Figure \ref{f:un-calls} can be
traversed only a finite number of times.

See below for an analysis of rule \eqref{un:ris}.

\item[Rule \eqref{un:ext2}] Symmetric arguments to previous rule apply.
\item[Rule \eqref{un:ext}] In all branches the size of the third argument
of the recursive call is strictly smaller than the initial one because $t$ is
removed from it.
\item[\label{i:20}Rule \eqref{un:ris}] Consider the non-variable RIS arguments
of this rule. The rule transforms these RIS into either extensional sets or new RIS.
\begin{enumerate}
  \item If all such arguments are transformed into extensional sets, then
  the rule does not recur on itself (but calls other rules of the $\Cup$
  constraint).
  \item If all such arguments are transformed into variable RIS, then the
  rule does not recur on itself (but calls other rules of the $\Cup$ constraint
  or, possibly, is in solved form).
  \item If at least one of such arguments is still a non-variable RIS, then
  there is a recursive call but the size of that argument is strictly smaller
  than the initial one. In effect, this case occurs when the initial argument in
  question is of the following form $\riss{x:\{a,b \plus D\}}{\F}{\P}$ with $a
  \neq b$ and $\lnot\F(a)$. That is, the initial argument is a non-variable RIS
  with at least two elements in its domain but the ``first'' one does not satisfy
  the filter. Then, $\P(a)$ does not belong to the RIS and so it is equal to
  $\riss{x:\{b \plus D\}}{\F}{\P}$. Precisely, this last RIS is used as the
  argument to the recursive call. Hence, at least one argument of the recursive
  call is strictly smaller than the initial one.

Nevertheless, if $\F(a)$ holds, then $a$ is removed from the domain but the RIS
becomes the set part of an extensional set with $\P(a)$ as its element part. In
this case, the size of this argument is equal to the initial one. However, in
this case, there is no recursive call but a call to one of rules
\eqref{un:ext1}, \eqref{un:ext2} and \eqref{un:ext}, which effectively reduce
the size of one of its.
\end{enumerate}







\item[Rule \eqref{disj:risext}] The recursive call is made with a r.h.s.
argument whose size is strictly smaller than the initial one because $t$ is
removed from the set.
\item[Rule \eqref{disj:extris}] Symmetric arguments of previous rule apply.
\item[Rule \eqref{disj:rris}] In either branch the recursive call is made
with a r.h.s. argument whose size is strictly smaller than the initial one
because $d$ is removed from the domain of the RIS.
\item[Rule \eqref{disj:rrisSym}] Symmetric arguments of previous rule apply.

Therefore, the loops shown in Figure \ref{f:disj-calls} end in a finite number
of times since rules \eqref{disj:risext}-\eqref{disj:rrisSym} reduce the size
of at least one of their arguments every time they are
called.
\end{description}

Local termination of each individual procedure, however, does not guarantee
global termination of $\SATRIS$\/, since the different procedures may be
dependent on each other. However, we observe that rewrite rules not involving
RIS in their left-hand sides do not construct any new RIS term in their
right-hand sides. They simply treat RIS terms as any other term. Hence the
presence of RIS terms do not affect their termination, which has been proved in
\cite{Dovier00}. Hence, it is enough to consider only the new rules involving
RIS terms.

Rules involving RIS generate both $\mathcal{RIS}$- and $\Ur$-formulas. The
$\SATX$ solver solves the $\Ur$-formulas without producing
$\mathcal{RIS}$-formulas. Then, for each constraint $\pi$ and each recursive
rewrite rule for it, we will analyze what $\mathcal{RIS}$-constraints other
than $\pi$ are generated. In this way we will see if there are cycles between
rewriting procedures. We proceed by rule inspection.
\begin{enumerate}
  \item Rules \eqref{neq-special1}, \eqref{s-neq-s} and \eqref{app.e:neq}
  generate $\in$- and $\notin$-constraints.
  \item Rule \eqref{t_in_X} generates $=$-constraints.
  \item Rules \eqref{un:ext1}, \eqref{un:ext2} and \eqref{un:ext} generate
  $=$-constraints.
  \item Rules \eqref{disj:risext}, \eqref{disj:extris}, \eqref{disj:rris}
  and \eqref{disj:rrisSym}, generate $\notin$-constraints.
\end{enumerate}
This is depicted in Figure \ref{f:calls}. Hence, as can be seen, there are no
loops between the rewriting procedures.

\begin{figure}
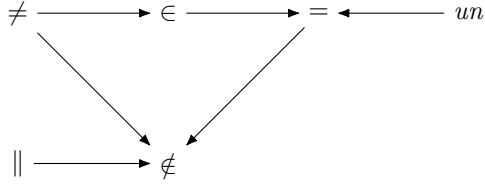

\begin{center}
\tikz [
    nonterminal/.style={rectangle},
    >=Latex, arrows={[round]}] {
\node [nonterminal] (neq) at (0,0) {$\neq$};
\node [nonterminal] (in) at (2,0) {$\in$};
\node [nonterminal] (eq) at (4,0) {$=$};
\node [nonterminal] (disj) at (0,-2) {$\disj$};
\node [nonterminal] (nin) at (2,-2) {$\notin$};
\node [nonterminal] (un) at (6,0) {$\Cup$};
\draw (neq) edge[->] (in)
      (in) edge[->] (eq)
      (neq) edge[->] (nin)
      (disj) edge[->] (nin)
      (un) edge[->] (eq)
      (eq) edge[->] (nin)
      ;
}
\end{center}
\caption{\label{f:calls}Calls relation between rewriting procedures}
\end{figure}

\subsection{\label{ap:equi}Equisatisfiability (Theorem \ref{equisatisfiable})}

This section contains the detailed proofs on the equisatisfiability of the
rewrite rules involving RIS presented in Sect. \ref{rules} and Appendix
\ref{ap:rules}; the
equisatisfiability of the remaining rules has been proved elsewhere
\cite{Dovier00}. Hence, these proofs use the rules considering that the control
expression is a variable and that the domain of RIS are not other RIS. These
proofs can be easily extended to the more general case.

In the following theorems and proofs, $\Fpv$, $\Gpv$, $\Ppv$ and $\Qpv$ are
shorthands for $F(x)$, $P(x)$, $G(x)$ and $Q(x)$, respectively.
Moreover,
note that a set of the form $\{P(x) : F(x)\}$ (where pattern and filter
are separated by a colon ($:$), instead of a bar ($|$), and the pattern is
\emph{before} the colon) is a shorthand for $\{y : \exists x (P(x) = y
\land F(x))\}$. That is, the set is written in the classic notation for
intensional sets used in mathematics. Finally, $H$ denotes the current
hypothesis.

$\P$ and $\Q$ are assumed to be bijective patterns. Recall that all patterns
allowed in $\LRIS$ fulfill this condition. The condition on the bijection of
patterns is necessary to prove equisatisfiability for rule
\eqref{e:v=e}.

The proof of Theorem \ref{equisatisfiable} rests on a series of lemmas each of
which shows that the set of solutions of left and right-hand sides of each
rewrite rule is the same.

\begin{lemma}[Equivalence of rule \eqref{e=e}]
We consider only the following case; all the others covered by this rule
are either symmetric or trivial.
\begin{gather*}
\{x:\e | \Fpv @ \Ppv\} = \e
\end{gather*}
\end{lemma}

\begin{proof}
\begin{gather*}
\{x:\e | \Fpv @ \Ppv\} \\
= \{\Pz : x \in \e \land \Fz\} \\
= \{\Pz : \false \land \Fz\} \\
= \{\Pz : \false\} \\
= \e
\end{gather*}
\end{proof}

\begin{proposition}\label{l:dD}
\begin{gather*}
\begin{split}
\forall d, & D: \\
                & \{x:\{d \plus D\} | \Fpv @ \Ppv\} = \{\Pd | \Fd\}
                \cup \{\Pz | x \in D \land \Fz\}
\end{split}
\end{gather*}
\end{proposition}

\begin{proof}
Taking any $d$ and $D$ we have:
\begin{gather*}
\{x:\{d \plus D\} | \F @ \P\} \\
= \{\Pz : x \in \{d \plus D\} \land \Fz\} \\
= \{\Pz : (x = d \lor x \in D) \land \Fz\} \\
= \{\Pz : (x = d \land \Fz) \lor (x \in D \land \Fz)\} \\
= \{\Pz : x = d \land \Fz\} \cup \{\Pz | x \in D \land \Fz\} \\
= \{\Pd : \Fd\} \cup\{\Pz | x \in D \land \Fz\}
\end{gather*}
\end{proof}

\begin{lemma}[Equivalence of rule \eqref{special2}]
We will consider only the case where $S \equiv \{c:\textbf{\textit{X}} | \F @ \P\}$ with $c \equiv \P$, because the other one is proved in \cite{Dovier00}:
\begin{gather*}
\begin{split}
\forall X, & t_0,\dots,t_k: \\
 & X = \{t_0,\dots, t_k \plus \{c:\textbf{\textit{X}} | \F @ \P\}\} \\
 & \iff X = \{t_0,\dots, t_k \plus \{c:\textbf{\textit{X}}[X/N] | \F @ \P\}\}
\end{split}
\end{gather*}
where $N$ is a new variable.
\end{lemma}

\begin{proof}
Taking any $X, t_0,\dots,t_k$ we have:
\begin{gather*}
X = \{t_0,\dots, t_k \plus \{c:\textbf{\textit{X}} | \F @ \P\}\}
  \why{def. of \textbf{\textit{X}} with $a \geq 0$} \\
\iff X = \{t_0,\dots, t_k\} \cup \{c:\{s_1,\dots,s_a \plus X\} | \F @ \P\}
  \why{def. $\plus$} \\
\iff X = \{t_0,\dots, t_k\} \cup \{c:\{s_1,\dots,s_a\} \cup X | \F @ \P\}
  \why{prop. intensional sets} \\
\iff X = \{t_0,\dots, t_k\} \cup \{c:\{s_1,\dots,s_a\} | \F @ \P\} \cup \{c:X | \F @ \P\}
\end{gather*}

Now note that if $c \equiv \P$, then:
\[
\{c:X | \F @ \P\} \subseteq X
\]

If $\{c:X | \F @ \P\} = X$, then:
\begin{gather*}
X = \{t_0,\dots, t_k\} \cup \{c:\{s_1,\dots,s_a\} | \F @ \P\} \cup \{c:X | \F @ \P\}
  \why{$A = B \cup A \iff B \subseteq A$} \\
\iff \{t_0,\dots, t_k\} \cup \{c:\{s_1,\dots,s_a\} | \F @ \P\} \subseteq \{c:X | \F @ \P\}
  \why{exists $N$ completing $\{c:X | \F @ \P\}$}\\
\iff \{t_0,\dots, t_k\} \cup \{c:\{s_1,\dots,s_a\} | \F @ \P\} \cup \{c:N | \F @ \P\} = \{c:X | \F @ \P\} \\
\iff \{t_0,\dots, t_k\} \cup \{c:\{s_1,\dots,s_a \plus N\} | \F @ \P\} = \{c:X | \F @ \P\} \\
\iff X = \{t_0,\dots, t_k\} \cup \{c:\{s_1,\dots,s_a \plus N\} | \F @ \P\} \\
\iff X = \{t_0,\dots, t_k\} \cup \{c:\textbf{\textit{X}}[N/X] | \F @ \P\}
\end{gather*}

If $\{c:X | \F @ \P\} \subset X$, then let $N$ be the (proper) subset of $X$ such that $\{c:N | \F @ \P\} = \{c:X | \F @ \P\}$. Hence:
\begin{gather*}
X = \{t_0,\dots, t_k\} \cup \{c:\{s_1,\dots,s_a\} | \F @ \P\} \cup \{c:X | \F @ \P\} \\
\iff X = \{t_0,\dots, t_k\} \cup \{c:\{s_1,\dots,s_a\} | \F @ \P\} \cup \{c:N | \F @ \P\} \\
\iff X = \{t_0,\dots, t_k\} \cup \{c:\{s_1,\dots,s_a \plus N\} | \F @ \P\} \\
\iff X = \{t_0,\dots, t_k\} \cup \{c:\textbf{\textit{X}}[N/X] | \F @ \P\}
\end{gather*}
\end{proof}

In the following lemma, the case where $R \equiv X$ and $S \equiv X$ is covered
in \cite{Dovier00}. Besides, given that the case where $R \equiv \{c:\bi{X} |
\F @ \P\}$ and $S \equiv X$ is covered by the case where $R \equiv \{c:\bi{X} |
\F @ \P\}$ and $S \equiv \{d:X | \G @ \Q\}$, we will prove only the last one.
Indeed, given that $c \equiv \P$, $d \equiv \Q$, we have $\{d:X | \G @ \Q\}
\subseteq X$. Then, the last case covers the second case when $\{d:X | \G @
\Q\} = X$. In fact, in this and the following lemma we will write $\{\bi{s}
\plus X\}$ instead of $\bi{X}$, where $\bi{s}$ denotes zero or more elements;
if $\bi{s}$ denotes zero elements then $\{\bi{s} \plus X\}$ is just $X$.

\begin{lemma}[Equivalence of rule \eqref{special3}]
If $c \equiv \P$ and $d \equiv \Q$:
\begin{gather*}
\begin{split}
\forall X, & t_0,\dots,t_m,s_0,\dots,s_k: \\
& \begin{split}
 \{t_0, & \dots, t_m \plus \{c:\{\bi{s} \plus X\} | \F @ \P\}\} =
   \{s_0, \dots, s_k \plus \{d:X | \G @ \Q\}\} \\
 {}\iff{}
  & t_0 = s_j \\
  &\quad{}\land \{t_0 \plus N_1\} = \{c:\{\bi{s} \plus X\} | \F @ \P\}
          \land t_0 \notin N_1 \\
  &\quad{}\land \{t_0 \plus N_2\} = \{d:X | \G @ \Q\}
          \land t_0 \notin N_2 \\
  &\quad{}\land \{t_1, \dots, t_m \plus N_1\}
                = \{s_0, \dots, s_{j-1},s_{j+1}, \dots, s_k
                    \plus N_2\} \\
  &{}\lor t_0 = s_j \\
  &\quad{}\land t_0 \notin \{c:\{\bi{s} \plus X\} | \F @ \P\}
          \land t_0 \notin \{d:X | \G @ \Q\} \\
  &\quad{}\land \{t_1, \dots, t_m \plus \{c:\{\bi{s} \plus X\} | \F @ \P\}\}
                = \{s_0, \dots, s_{j-1},s_{j+1}, \dots, s_k
                    \plus \{d:X | \G @ \Q\}\} \\
  &{}\lor t_0 = s_j \\
  &\quad{}\land \{t_0 \plus N\} = \{c:\{\bi{s} \plus X\} | \F @ \P\}
          \land t_0 \notin N
          \land t_0 \notin \{d:X | \G @ \Q\} \\
  &\quad{}\land \{t_1, \dots, t_m \plus N\}
                = \{s_0, \dots, s_{j-1},s_{j+1}, \dots, s_k
                    \plus \{d:X | \G @ \Q\}\} \\
  &{}\lor t_0 = s_j \\
  &\quad{}\land t_0 \notin \{c:\{\bi{s} \plus X\} | \F @ \P\}
          \land \{s_j \plus N\} = \{d:X | \G @ \Q\}
          \land s_j \notin N \\
  &\quad{}\land \{t_1, \dots, t_m \plus \{c:\{\bi{s} \plus X\} | \F @ \P\}\}
                = \{s_0, \dots, s_{j-1},s_{j+1}, \dots, s_k
                    \plus N\} \\
  &{}\lor t_0 = s_j \\
  &{}\quad\land \{t_1, \dots, t_m \plus \{c:\{\bi{s} \plus X\} | \F @ \P\}\}
                = \{s_0, \dots, s_k \plus \{d:X | \G @ \Q\}\} \\
  &{}\lor t_0 = s_j \\
  &\quad{}\land \{t_0, \dots, t_m \plus \{c:\{\bi{s} \plus X\} | \F @ \P\}\}
                = \{s_0, \dots, s_{j-1},s_{j+1}, \dots, s_k
                    \plus \{d:X | \G @ \Q\}\} \\
  &{}\lor X = \{t_0 \plus N\} \land \G(t_0)  \\
  &{}\quad\land \{t_1, \dots, t_m \plus \{c:\{\bi{s} \plus N\} | \F @ \P\}\}
                = \{s_0, \dots, s_k \plus \{d:N | \G @ \Q\}\}
 \end{split}
 \end{split}
\end{gather*}
\end{lemma}

\begin{proof}
\mbox{}\\
\noindent $\implies)$ \\
By H, $t_0$ must belong the the r.h.s. set. We will divide the proof in the following cases each of which corresponds to one of the branches at the r.h.s. Note that the disjunction of all the preconditions covers all possible cases.
\begin{itemize}
  \item $t_0 \in \{s_0,\dots, s_k\}$, then $t_0 = s_j$, for some $j$
    \begin{itemize}
      \item $t_0 \notin \{s_0, \dots, s_{j-1},s_{j+1}, \dots, s_k\}$
      \begin{itemize}
        \item $s_j \notin \{t_1, \dots, t_m\}$
        \begin{itemize}
        \item First branch: $\{t_0 \plus N_1\} = \{c:\{\bi{s} \plus X\} | \F @ \P\}
        \land \{t_0 \plus N_2\} = \{d:X | \G @ \Q\} \land t_0 \notin N_1 \land t_0 \notin N_2$.
\begin{gather*}
\{t_1, \dots, t_m \plus N_1\}
                = \{s_0, \dots, s_{j-1},s_{j+1}, \dots, s_k
                    \plus N_2\} \\
\iff \why{assumptions $t_0$ does not belong to either set} \\
\{t_0,t_1, \dots, t_m \plus N_1\}
                = \{s_0, \dots, s_{j-1},t_0,s_{j+1}, \dots, s_k
                    \plus N_2\} \\
\iff \why{absorption on the left, semantics of $\plus$ and $s_j = t_0$} \\
\{t_0, \dots, t_m \plus \{t_0 \plus N_1\}\}
                = \{s_0, \dots, s_{j-1},s_j,s_{j+1}, \dots, s_k \plus
                    \{t_0 \plus N_2\}\} \\
\iff \why{assumptions of this branch} \\
\{t_0, \dots, t_m \plus \{c:\{\bi{s} \plus X\} | \F @ \P\}\}
                = \{s_0, \dots, s_k \plus \{d:X | \G @ \Q\}\}
\end{gather*}
Since the last equality is H then the first equality holds.

        \item Second branch: $t_0 \notin \{c:\{\bi{s} \plus X\} | \F @ \P\} \land t_0 \notin \{d:X | \G @ \Q\}$
\begin{gather*}
\{t_1, \dots, t_m \plus \{c:\{\bi{s} \plus X\} | \F @ \P\}\}
                = \{s_0, \dots, s_{j-1},s_{j+1}, \dots, s_k
                    \plus \{d:X | \G @ \Q\}\} \\
\iff \why{assumptions $t_0$ does not belong to either set}\\
\{t_0,t_1, \dots, t_m \plus \{c:\{\bi{s} \plus X\} | \F @ \P\}\}
                = \{s_0, \dots, s_{j-1},t_0,s_{j+1}, \dots, s_k
                    \plus \{d:X | \G @ \Q\}\} \\
\iff \why{$s_j = t_0$} \\
\{t_0, \dots, t_m \plus \{c:\{\bi{s} \plus X\} | \F @ \P\}\}
                = \{s_0, \dots, s_{j-1},s_j,s_{j+1}, \dots, s_k \plus
                    \{d:X | \G @ \Q\}\}
\end{gather*}
Since the last equality is H then the first equality holds.

        \item Third branch: $\{t_0 \plus N\} = \{c:\{\bi{s} \plus X\} | \F @ \P\} \land t_0 \notin \{d:X | \G @ \Q\} \land t_0 \notin N$.
\begin{gather*}
\{t_1, \dots, t_m \plus N\}
                = \{s_0, \dots, s_{j-1},s_{j+1}, \dots, s_k
                    \plus \{d:X | \G @ \Q\}\} \\
\iff \why{assumptions $t_0$ does not belong to either set}\\
\{t_0,t_1, \dots, t_m \plus N\}
                = \{s_0, \dots, s_{j-1},t_0,s_{j+1}, \dots, s_k
                    \plus \{d:X | \G @ \Q\}\} \\
\iff \why{absorption on the left, semantics of $\plus$ and $s_j = t_0$} \\
\{t_0,t_1, \dots, t_m \plus \{t_0 \plus N\}\}
                = \{s_0, \dots, s_{j-1},s_j,s_{j+1}, \dots, s_k
                    \plus \{d:X | \G @ \Q\}\} \\
\iff \why{assumption of this branch} \\
\{t_0, \dots, t_m \plus \{c:\{\bi{s} \plus X\} | \F @ \P\}\}
                = \{s_0, \dots, s_k \plus
                    \{d:X | \G @ \Q\}\}
\end{gather*}
Since the last equality is H then the first equality holds.

        \item Fourth branch: $t_0 \notin \{c:\{\bi{s} \plus X\} | \F @ \P\} \land \{s_j \plus N\} = \{d:X | \G @ \Q\} \land s_j \notin N$.
\begin{gather*}
\{t_1, \dots, t_m \plus \{c:\{\bi{s} \plus X\} | \F @ \P\}\}
                = \{s_0, \dots, s_{j-1},s_{j+1}, \dots, s_k
                    \plus N\} \\
\iff \why{assumptions $t_0$ does not belong to either set}\\
\{t_0,t_1, \dots, t_m \plus \{c:\{\bi{s} \plus X\} | \F @ \P\}\}
                = \{s_0, \dots, s_{j-1},t_0,s_{j+1}, \dots, s_k
                    \plus N\} \\
\iff \why{$s_j = t_0$} \\
\{t_0, \dots, t_m \plus \{c:\{\bi{s} \plus X\} | \F @ \P\}\}
                = \{s_0, \dots, s_{j-1},s_j,s_{j+1}, \dots, s_k \plus
                    N\}
\iff \why{absorption on the left and semantics of $\plus$} \\
\{t_0, \dots, t_m \plus \{c:\{\bi{s} \plus X\} | \F @ \P\}\}
                = \{s_0, \dots, s_{j-1},s_j,s_{j+1}, \dots, s_k \plus
                    \{s_j \plus N\}\} \\
\iff \why{assumption of this branch} \\
\{t_0, \dots, t_m \plus \{c:\{\bi{s} \plus X\} | \F @ \P\}\}
                = \{s_0, \dots, s_k \plus
                    \{d:X | \G @ \Q\}\}
\end{gather*}
Since the last equality is H then the first equality holds.
      \end{itemize}

        \item $s_j \in \{t_1, \dots, t_m\}$

        Now we consider the fifth branch.

        Given that $s_j = t_0$ and that $s_j \in \{t_1, \dots, t_m\}$, then $t_0 \in \{t_1, \dots, t_m\}$.
        \begin{gather*}
        \{t_1, \dots, t_m \plus \{c:\{\bi{s} \plus X\} | \F @ \P\}\}
          \why{$t_0 \in \{t_1, \dots, t_m\}$} \\
        = \{t_0,t_1, \dots, t_m \plus \{c:\{\bi{s} \plus X\} | \F @ \P\}\}
          \why{H} \\
        = \{s_0, \dots, s_k \plus \{d:X | \G @ \Q\}\}
        \end{gather*}
     \end{itemize}

      \item $t_0 \in \{s_0, \dots, s_{j-1},s_{j+1}, \dots, s_k\}$

      Now we consider the sixth branch.

      Given that $s_j = t_0$ and that $t_0 \in \{s_0, \dots, s_{j-1},s_{j+1}, \dots, s_k\}$, then $s_j \in \{s_0, \dots, s_{j-1},s_{j+1}, \dots, s_k\}$.
      \begin{gather*}
      \{s_0, \dots, s_{j-1},s_{j+1}, \dots, s_k
        \plus \{d:X | \G @ \Q\}\}
          \why{$s_j \in
                \{s_0, \dots, s_{j-1},s_{j+1}, \dots, s_k\}$} \\
      = \{s_0, \dots s_k \plus \{d:X | \G @ \Q\}\}
        \why{H} \\
      = \{t_0, \dots, t_m \plus \{c:\{\bi{s} \plus X\} | \F @ \P\}\}
      \end{gather*}
    \end{itemize}
  \item $t_0 \notin \{s_0,\dots, s_k\}$

  Now we consider the last branch.

  Given that $t_0 \notin \{s_0,\dots, s_k\}$, then necessarily $t_0 \in \{d:X | \G @ \Q\}$, then $t_0 \in X$ and $\G(t_0)$ holds. Now, the proof is divided in two cases.
  \begin{itemize}
    \item $t_0 \in \{t_1,\dots, t_m\} \cup \{\bi{s}\} \land \F(t_0)$

    In this case we take $N = X$. Then $ \{t_0 \plus N\} = \{t_0 \plus X\} = X$, where the last equality holds because $t_0 \in X$. Now:
    \begin{gather*}
    \{t_1, \dots, t_m \plus \{d:\{\bi{s} \plus N\} | \F @ \P\}\}
      \why{$t_0 \in \{t_1,\dots, t_m\} \cup \{\bi{s}\} \land \F(t_0)$} \\
    = \{t_0,t_1, \dots, t_m \plus \{d:\{\bi{s} \plus N\} | \F @ \P\}\}
      \why{$N=X$} \\
    = \{t_0,t_1, \dots, t_m \plus \{d:\{\bi{s} \plus X\} | \F @ \P\}\}
      \why{H} \\
    = \{s_0, \dots, s_k \plus \{d:X | \G @ \Q\}\}
      \why{$N=S$} \\
    = \{s_0, \dots, s_k \plus \{d:N | \G @ \Q\}\}
    \end{gather*}

    \item $t_0 \notin \{t_1,\dots, t_m\} \cup \{\bi{s}\} \lor \lnot\F(t_0)$

    In this case we take $N = X \setminus \{t_0\}$, then $X = \{t_0 \plus N\}$ because $t_0 \in X$. Besides $t_0 \notin N$. Now we start from H.
    \begin{gather*}
    \{t_0, \dots, t_m \plus \{c:\{\bi{s} \plus X\} | \F @ \P\}\}
      = \{s_0, \dots, s_k \plus \{d:X | \G @ \Q\}\} \\
    \iff \why{$X = \{t_0 \plus N\}$} \\
    \{t_0, \dots, t_m \plus \{c:\{\bi{s} \plus \{t_0 \plus N\}\} | \F @ \P\}\}
      = \{s_0, \dots, s_k \plus \{d:\{t_0 \plus N\} | \G @ \Q\}\} \\
    \iff \why{properties of $\plus$ and intensional sets} \\
    \{t_0, \dots, t_m \plus \{c:\{\bi{s} \plus N\} | \F @ \P\}\}
        \cup \{c:\{t_0\} | \F @ \P\} \\
    \quad{} = \{s_0, \dots, s_k \plus \{d:N | \G @ \Q\}\}
        \cup \{d:\{t_0\} | \G @ \Q\} \\
    \iff \why{subtract $\{t_0\}$ on both sides} \\
    (\{t_0, \dots, t_m \plus \{c:\{\bi{s} \plus N\} | \F @ \P\}\}
        \cup \{c:\{t_0\} | \F @ \P\}) \setminus \{t_0\} \\
    \quad{} = (\{s_0, \dots, s_k \plus \{d:N | \G @ \Q\}\}
        \cup \{d:\{t_0\} | \G @ \Q\}) \setminus \{t_0\} \\
    \iff \why{$\setminus$ distributes over $\cup$} \\
    (\{t_0, \dots, t_m \plus \{c:\{\bi{s} \plus N\} | \F @ \P\}\} \setminus \{t_0\})
        \cup (\{c:\{t_0\} | \F @ \P\} \setminus \{t_0\}) \\
    \quad{} = (\{s_0, \dots, s_k \plus \{d:N | \G @ \Q\}\}
                 \setminus \{t_0\})
        \cup (\{d:\{t_0\} | \G @ \Q\} \setminus \{t_0\}) \\
    \iff \why{$t_0 \notin \{t_1,\dots, t_m\} \cup \{\bi{s}\}$, $t_0 \notin N$, $t_0 \notin \{s_0,\dots, s_k\}$} \\
    \{t_1, \dots, t_m \plus \{c:\{\bi{s} \plus N\} | \F @ \P\}\}
    = \{s_0, \dots, s_k \plus \{d:N | \G @ \Q\}\}
    \end{gather*}
  \end{itemize}
\end{itemize}

\mbox{}\\
\noindent $\Longleftarrow)$ \\
We will consider each branch and will prove that the equality holds.
\begin{itemize}
  \item First branch.
  \begin{gather*}
  \{t_0, \dots, t_m \plus \{c:\{\bi{s} \plus X\} | \F @ \P\}\}
    \why{H} \\
  =  \{t_0, \dots, t_m \plus \{t_0 \plus N_1\}\}
    \why{properties of $\plus$} \\
  =  \{t_0,t_0, \dots, t_m \plus N_1\}
    \why{absorption on the left} \\
  =  \{t_0, \dots, t_m  \plus N_1\} \\
  =  \{t_0\} \cup \{t_1, \dots, t_m \plus N_1\}
    \why{H} \\
  =  \{t_0\}
     \cup \{s_0, \dots, s_{j-1},s_{j+1}, \dots, s_k \plus N_2\}
     \why{semantics of $\plus$} \\
  =  \{t_0,s_0, \dots, s_{j-1},s_{j+1}, \dots, s_k \plus N_2\}
    \why{absorption on the left} \\
  =  \{t_0,s_0, \dots, s_{j-1},t_0,s_{j+1}, \dots, s_k \plus N_2\}
    \why{H} \\
  =  \{s_0, \dots, s_{j-1},s_j,s_{j+1}, \dots, s_k
       \plus \{t_0 \plus N_2\}\}
    \why{H} \\
  = \{s_0, \dots, s_k \plus \{d:X | \G @ \Q\}\}
  \end{gather*}

  \item Second branch.
  \begin{gather*}
  \{t_0, \dots, t_m \plus \{c:\{\bi{s} \plus X\} | \F @ \P\}\} \why{semantics of $\plus$}\\
  =  \{t_0\} \cup \{t_1, \dots, t_m \plus \{c:\{\bi{s} \plus X\} | \F @ \P\}\}
    \why{H} \\
  =  \{t_0\}
     \cup \{s_0, \dots, s_{j-1},s_{j+1}, \dots, s_k
            \plus \{d:X | \G @ \Q\}\}
     \why{semantics of $\plus$} \\
  =  \{s_0, \dots, s_{j-1},t_0,s_{j+1}, \dots, s_k
       \plus \{d:X | \G @ \Q\}\}
    \why{H} \\
  =  \{s_0, \dots, s_{j-1},s_j,s_{j+1}, \dots, s_k
       \plus \{d:X | \G @ \Q\}\} \\
  = \{s_0, \dots, s_k \plus \{d:X | \G @ \Q\}\}
  \end{gather*}

  \item Third branch.
  \begin{gather*}
  \{t_0, \dots, t_m \plus \{c:\{\bi{s} \plus X\} | \F @ \P\}\}
    \why{H} \\
  =  \{t_0, \dots, t_m \plus \{t_0 \plus N\}\}
    \why{properties of $\plus$} \\
  =  \{t_0,t_0, \dots, t_m \plus N\}
    \why{absorption on the left} \\
  =  \{t_0, \dots, t_m  \plus N\} \\
  =  \{t_0\} \cup \{t_1, \dots, t_m \plus N\}
    \why{H} \\
  =  \{t_0\}
     \cup \{s_0, \dots, s_{j-1},s_{j+1}, \dots, s_k
            \plus \{d:X | \G @ \Q\}\}
     \why{semantics of $\plus$} \\
  =  \{s_0, \dots, s_{j-1},t_0,s_{j+1}, \dots, s_k
            \plus \{d:X | \G @ \Q\}\}
    \why{H} \\
  =  \{s_0, \dots, s_{j-1},s_j,s_{j+1}, \dots, s_k
            \plus \{d:X | \G @ \Q\}\}
    \why{H} \\
  = \{s_0, \dots, s_k \plus \{d:X | \G @ \Q\}\}
  \end{gather*}

  \item Fourth branch.
  \begin{gather*}
  \{s_0, \dots, s_k \plus \{d:X | \G @ \Q\}\}
    \why{H} \\
  = \{s_0, \dots, s_k \plus \{s_j \plus N\}\}
    \why{properties of $\plus$} \\
  = \{s_j,s_0, \dots, s_k \plus N\}
    \why{absorption on the left} \\
  = \{s_0, \dots, s_k \plus N\}
    \why{semantics of $\plus$} \\
  = \{s_j\} \cup \{s_0, \dots, s_{j-1},s_{j+1}, \dots, s_k
                   \plus N\}
    \why{H} \\
  = \{s_j\} \cup \{t_1, \dots, t_m \plus \{c:\{\bi{s} \plus X\} | \F @ \P\}\}
    \why{semantics of $\plus$} \\
  = \{s_j,t_1, \dots, t_m \plus \{c:\{\bi{s} \plus X\} | \F @ \P\}\}
    \why{H} \\
  = \{t_0,t_1, \dots, t_m \plus \{c:\{\bi{s} \plus X\} | \F @ \P\}\}
  \end{gather*}

  \item Fifth branch.
  \begin{gather*}
  \{t_0, \dots, t_m \plus \{c:\{\bi{s} \plus X\} | \F @ \P\}\}
    \why{semantics of $\plus$} \\
  = \{t_0\} \cup \{t_1, \dots, t_m \plus \{c:\{\bi{s} \plus X\} | \F @ \P\}\}
    \why{H} \\
  = \{t_0\}
     \cup \{s_0, \dots, s_k \plus \{d:X | \G @ \Q\}\}
     \why{semantics of $\plus$} \\
  = \{t_0,s_0, \dots, s_k \plus \{d:X | \G @ \Q\}\}
    \why{H} \\
  = \{s_j,s_0, \dots, s_k \plus \{d:X | \G @ \Q\}\}
    \why{absorption on the left} \\
  = \{s_0, \dots, s_k \plus \{d:X | \G @ \Q\}\}
  \end{gather*}

  \item Sixth branch.
  \begin{gather*}
  \{t_0, \dots, t_m \plus \{c:\{\bi{s} \plus X\} | \F @ \P\}\}
    \why{absorption on the left} \\
  = \{t_0,t_0, \dots, t_m \plus \{c:\{\bi{s} \plus X\} | \F @ \P\}\}
    \why{semantics of $\plus$} \\
  = \{t_0\} \cup \{t_0, \dots, t_m \plus \{c:\{\bi{s} \plus X\} | \F @ \P\}\}
     \why{H} \\
  = \{t_0\}
    \cup \{s_0, \dots, s_{j-1},s_{j+1}, \dots, s_k
           \plus \{d:X | \G @ \Q\}\}
    \why{semantics of $\plus$} \\
  = \{s_0, \dots, s_{j-1},t_0,s_{j+1}, \dots, s_k
           \plus \{d:X | \G @ \Q\}\}
    \why{H} \\
  = \{s_0, \dots, s_{j-1},s_j,s_{j+1}, \dots, s_k
           \plus \{d:X | \G @ \Q\}\} \\
  = \{s_0, \dots, s_k \plus \{d:X | \G @ \Q\}\}
  \end{gather*}

  \item Seventh (last) branch.
  \begin{gather*}
  \{t_0, \dots, t_m \plus \{c:\{\bi{s} \plus X\} | \F @ \P\}\}
    \why{H} \\
  = \{t_0, \dots, t_m \plus \{c:\{\bi{s} \plus \{t_0 \plus N\}\} | \F @ \P\}\}
    \why{semantics of $\plus$, property intensional sets} \\
  = \{t_0\}
    \cup \{c:\{t_0\} | \F @ \P\}
    \cup \{t_1, \dots, t_m \plus \{c:\{\bi{s} \plus N\} | \F @ \P\}\}
     \why{H} \\
  = \{t_0\}
    \cup \{c:\{t_0\} | \F @ \P\}
    \cup \{s_0, \dots, s_k \plus \{d:N | \G @ \Q\}\}
    \why{$\{c:\{t_0\} | \F @ \P\} \subseteq \{t_0\}$} \\
  = \{t_0\}
    \cup \{s_0, \dots, s_k \plus \{d:N | \G @ \Q\}\}
    \why{semantics of $\plus$} \\
  = \{t_0,s_0, \dots, s_k \plus \{d:N | \G @ \Q\}\}
    \why{semantics of $\plus$} \\
  = \{s_0, \dots, s_k \plus \{t_0 \plus \{d:N | \G @ \Q\}\}\}
    \why{H, semantics of $\plus$ and property of intensional sets} \\
  = \{s_0, \dots, s_k \plus \{d:\{t_0 \plus N\} | \G @ \Q\}\}
    \why{H} \\
  = \{s_0, \dots, s_k \plus \{d:X | \G @ \Q\}\}
  \end{gather*}

\end{itemize}
\end{proof}


\begin{lemma}[Equivalence of rule \eqref{special6}]
If $c \not\equiv \P$ and $d \not\equiv \Q$:
\begin{gather*}
\begin{split}
\forall X, & t_0,\dots,t_m,s_0,\dots,s_k: \\
& \begin{split}
 \{t_0, & \dots, t_m \plus \{c:\{\bi{s} \plus X\} | \F @ \P\}\} =
   \{s_0, \dots, s_k \plus \{d:X | \G @ \Q\}\} \\
 {}\iff{}
  & t_0 = s_j \\
  &\quad{}\land \{t_0 \plus N_1\} = \{c:\{\bi{s} \plus X\} | \F @ \P\}
          \land t_0 \notin N_1 \\
  &\quad{}\land \{t_0 \plus N_2\} = \{d:X | \G @ \Q\}
          \land t_0 \notin N_2 \\
  &\quad{}\land \{t_1, \dots, t_m \plus N_1\}
                = \{s_0, \dots, s_{j-1},s_{j+1}, \dots, s_k
                    \plus N_2\} \\
  &{}\lor t_0 = s_j \\
  &\quad{}\land t_0 \notin \{c:\{\bi{s} \plus X\} | \F @ \P\}
          \land t_0 \notin \{d:X | \G @ \Q\} \\
  &\quad{}\land \{t_1, \dots, t_m \plus \{c:X | \F @ \P\}\}
                = \{s_0, \dots, s_{j-1},s_{j+1}, \dots, s_k
                    \plus \{d:X | \G @ \Q\}\} \\
  &{}\lor t_0 = s_j \\
  &\quad{}\land \{t_0 \plus N\} = \{c:\{\bi{s} \plus X\} | \F @ \P\}
          \land t_0 \notin N
          \land t_0 \notin \{d:X | \G @ \Q\} \\
  &\quad{}\land \{t_1, \dots, t_m \plus N\}
                = \{s_0, \dots, s_{j-1},s_{j+1}, \dots, s_k
                    \plus \{d:X | \G @ \Q\}\} \\
  &{}\lor t_0 = s_j \\
  &\quad{}\land t_0 \notin \{c:\{\bi{s} \plus X\} | \F @ \P\}
          \land \{s_j \plus N\} = \{d:X | \G @ \Q\}
          \land s_j \notin N \\
  &\quad{}\land \{t_1, \dots, t_m \plus \{c:\{\bi{s} \plus X\} | \F @ \P\}\}
                = \{s_0, \dots, s_{j-1},s_{j+1}, \dots, s_k
                    \plus N\} \\
  &{}\lor t_0 = s_j \\
  &{}\quad\land \{t_1, \dots, t_m \plus \{c:\{\bi{s} \plus X\} | \F @ \P\}\}
                = \{s_0, \dots, s_k \plus \{d:X | \G @ \Q\}\} \\
  &{}\lor t_0 = s_j \\
  &\quad{}\land \{t_0, \dots, t_m \plus \{c:\{\bi{s} \plus X\} | \F @ \P\}\}
                = \{s_0, \dots, s_{j-1},s_{j+1}, \dots, s_k
                    \plus \{d:X | \G @ \Q\}\} \\
  &{}\lor X = \{n \plus N\} \land \Gpv(n) \land t_0 = \Qpv(n) \land
      \F(n) \land \P(n) \notin \{s_0, \dots, s_k\} \\
  &\quad{}\land t_0 = \P(n)
          \land \{t_1, \dots, t_m \plus \{c:\{\bi{s} \plus N\} | \F @ \P\}\}
                = \{s_0, \dots, s_k \plus \{d:N | \G @ \Q\}\} \\
  &{}\lor X = \{n \plus N\} \land \Gpv(n) \land t_0 = \Qpv(n) \land
  \F(n) \land \P(n) = s_j \\
  &\quad{}\land \{\P(n),t_1, \dots, t_m \plus \{c:\{\bi{s} \plus N\} | \F @ \P\}\} =
                \{s_0, \dots,  s_k \plus \{d:N | \G @ \Q\}\} \\
  &{}\lor X = \{n \plus N\} \land \Gpv(n) \land t_0 = \Qpv(n)
          \land \lnot\F(n) \\
  &\quad{}\land \{t_1, \dots, t_m \plus \{c:\{\bi{s} \plus N\} | \F @ \P\}\} =
                \{s_0, \dots, s_k \plus \{d:N | \G @ \Q\}\}
 \end{split}
 \end{split}
\end{gather*}
\end{lemma}

\begin{1}
\hlg{COMMENT.} The eighth branch is a little bit strange. The point is that we
can't remove $\P(n)$ from any set because we don't know whether it's repeated
or not. Then, we leave it but the size of the constraint is strictly smaller
because we removed $t_0$ from the l.h.s. set. We can write that branch as
follows if you like:
\begin{gather*}
X = \{n \plus N\} \land \Gpv(n) \land t_0 = \Qpv(n) \\
  \quad{}\land \{c:\{\bi{s} \plus X\} | \F @ \P\} = \{\P(n) \plus N_1\} \land \P(n) \notin N_1 \land \P(n) = s_j \\
  \quad{}\land \{t_1, \dots, t_m,\P(n) \plus N_1\} =
                \{s_0, \dots,  s_k \plus \{d:N | \G @ \Q\}\}
\end{gather*}
I couldn't use the solution you proposed in the email because the problem was
also in the old seventh branch because it was wrong. So first, I had to change
it and then I had to add two more branches. I guess I thought that branch was
like rule =13 but the added complexity here is that there are elements at both
sides, while in rule =13 there are elements only at the r.h.s. set.
\end{1}

\begin{proof}
The first six branches of this lemma are proved exactly as in the previous
lemma, in both directions, because those branches are equal in both lemmas.
Hence, we will prove only the last three branches, in both directions.
\mbox{}\\
\noindent $\implies)$ \\
In order to prove these branches we need to recall that $\P$ and $\Q$ are
ordered pairs whose first component is the control variable and whose second
component is a $\Ur$-term. Then, from now on we will take $\P(x) = (x,f_\P(x))$
for some term $f_\P$ and $\Q(x) = (x,f_\Q(x))$, for some term $f_\Q$.

Besides, all these branches are proved under the assumption as in the previous
lemma:
\begin{equation}\label{l:special6:as1}
t_0 \notin \{s_0,\dots, s_k\}
\end{equation}

The first step in the proof is that we can conclude that $t_0$ must belong the
the r.h.s. set, by H. Hence, due to \eqref{l:special6:as1}, $t_0$ must belong
to $\{d:X  | \G @ \Q\}$, which means that there exists $n \in X$, $\G(n)$ holds
and $t_0 = \Q(n)$. Observe that $t_0 = \Q(n)$ means that $t_0 = (n,f_\Q(n))$.

Now, note that the following conditions of the branches under consideration:
\begin{gather*}
\F(n) \land  \P(n) \notin \{s_0, \dots, s_k\} \\
\F(n) \land \P(n) = s_j \\
\lnot\F(n)
\end{gather*}
are mutually exclusive. Since $n \in X$ then in the first two cases we have $\P(n) \in \{c:\{\bi{s} \plus X\} | \F @ \P\}$, which implies that $\P(n)$ belongs to the l.h.s. set. If $\P(n)$ belongs to the l.h.s. set of H, then it must belong
to the r.h.s. of H. Then, in the first case we assume $\P(n)$ belongs to $\{d:X
| \G @ \Q\}$ while in the second we assume is one of the $s_i$.

Hence, in order to prove each branch we will assume the corresponding condition
and we will prove the remaining predicates of the conclusion.

\begin{itemize}
\item  Assume $\F(n) \land \P(n) \notin \{s_0, \dots, s_k\}$.

Will we prove that $t_0 = \P(n)$ by contradiction. We know that $\P(n) \in
\{c:\{\bi{s} \plus X\} | \F @ \P\}$. So:
\begin{gather*}
(n,f_\P(n)) \in \{c:\{\bi{s} \plus X\} | \F @ \P\}
         \land t_0 \neq (n,f_\P(n))
         \why{assumption of this branch} \\
\implies (n,f_\P(n)) \in \{d:X  | \G @ \Q\}
         \land f_\Q(n) \neq f_\P(n)
         \why{semant. intensional sets, def. of $\Q$}\\
\implies \G(n) \land (n,f_\P(n)) = (n,f_\Q(n))
         \land f_\Q(n) \neq f_\P(n)
         \why{drop $\G(n)$ and equality of pairs} \\
\implies f_\P(n) = f_\Q(n) \land f_\Q(n) \neq f_\P(n)
\end{gather*}
which is clearly a contradiction. Hence, $t_0 = \P(n)$.

\smallskip

Now we will prove that $\{t_1, \dots, t_m \plus \{c:N | \F @ \P\}\} = \{s_0,
\dots, s_k \plus \{d:N | \G @ \Q\}\}$, by dividing the proof in two cases.
  \begin{itemize}
    \item $t_0 \in \{t_1,\dots, t_m\} \cup \{c:\bi{s} | \F @ \P\}$

    In this case we take $N = X$. Then $ \{n \plus N\} = \{n \plus X\} = X$, where the last equality holds because $n \in X$.

    \begin{itemize}
  \item $t_0 \in \{t_1,\dots,t_m\}$
\begin{gather*}
\{s_0, \dots, s_k \plus \{d:N | \G @ \Q\}\}
  \why{$N = X$} \\
= \{s_0, \dots, s_k \plus \{d:X | \G @ \Q\}\}
  \why{H} \\
= \{t_0,t_1,\dots,t_m \plus \{c:\{\bi{s} \plus X\} | \F @ \P\}\}
  \why{$t_0 \in \{t_1,\dots,t_m\}$, left absorption} \\
= \{t_1,\dots,t_m \plus \{c:\{\bi{s} \plus X\} | \F @ \P\}\}
  \why{$X = N$} \\
= \{t_1,\dots,t_m \plus \{c:\{\bi{s} \plus N\} | \F @ \P\}\}
\end{gather*}
  \item $t_0 \in \{c:\{\bi{s}\} | \F @ \P\}$. Note that in this case $\{c:\{\bi{s} \plus X\} | \F @ \P\} = \{t_0 \plus \{c:\{\bi{s} \plus X\} | \F @ \P\}\}$.
\begin{gather*}
\{s_0, \dots, s_k \plus \{d:N | \G @ \Q\}\}
  \why{$N = X$} \\
= \{s_0, \dots, s_k \plus \{d:X | \G @ \Q\}\}
  \why{H} \\
= \{t_0,t_1,\dots,t_m \plus \{c:\{\bi{s} \plus X\} | \F @ \P\}\}
  \why{semantics of $\plus$} \\
= \{t_1,\dots,t_m \plus \{t_0 \plus \{c:\{\bi{s} \plus X\} | \F @ \P\}\}\}
  \why{$t_0 \in \{c:\{\bi{s}\} | \F @ \P\}$} \\
= \{t_1,\dots,t_m \plus \{c:\{\bi{s} \plus X\} | \F @ \P\}\}
  \why{$X = N$} \\
= \{t_1,\dots,t_m \plus \{c:\{\bi{s} \plus N\} | \F @ \P\}\}
\end{gather*}
\end{itemize}

    \item $t_0 \notin \{t_1,\dots, t_m\} \cup \{c:\bi{s} | \F @ \P\}$

    In this case we take $N = X \setminus \{n\}$, then $X = \{n \plus N\}$ because $n \in X$. Besides $n \notin N$. Now we start from H.
    \begin{gather*}
    \{t_0, \dots, t_m \plus \{c:\{\bi{s} \plus X\} | \F @ \P\}\}
      = \{s_0, \dots, s_k \plus \{d:X | \G @ \Q\}\} \\
    \iff \why{$X = \{n \plus N\}$} \\
    \{t_0, \dots, t_m \plus \{c:\{\bi{s} \plus \{n \plus N\}\} | \F @ \P\}\}
      = \{s_0, \dots, s_k \plus \{d:\{n \plus N\} | \G @ \Q\}\} \\
    \iff \why{properties of $\plus$ and intensional sets} \\
    \{t_0, \dots, t_m \plus \{c:\{\bi{s} \plus N\} | \F @ \P\}\}
        \cup \{c:\{n\} | \F @ \P\} \\
    \quad{} = \{s_0, \dots, s_k \plus \{d:N | \G @ \Q\}\}
        \cup \{d:\{n\} | \G @ \Q\} \\
    \iff \why{subtract $\{t_0\}$ at both sides} \\
    (\{t_0, \dots, t_m \plus \{c:\{\bi{s} \plus N\} | \F @ \P\}\}
        \cup \{c:\{n\} | \F @ \P\}) \setminus \{t_0\} \\
    \quad{} = (\{s_0, \dots, s_k \plus \{d:N | \G @ \Q\}\}
        \cup \{d:\{n\} | \G @ \Q\}) \setminus \{t_0\} \\
    \iff \why{$\setminus$ distributes over $\cup$} \\
    (\{t_0, \dots, t_m \plus \{c:\{\bi{s} \plus N\} | \F @ \P\}\} \setminus \{t_0\})
        \cup (\{c:\{n\} | \F @ \P\} \setminus \{t_0\}) \\
    \quad{} = (\{s_0, \dots, s_k \plus \{d:N | \G @ \Q\}\}
                 \setminus \{t_0\})
        \cup (\{d:\{n\} | \G @ \Q\} \setminus \{t_0\}) \\
    \iff \why{$(\dagger)$} \\
    \{t_1, \dots, t_m \plus \{c:\{\bi{s} \plus N\} | \F @ \P\}\}
    = \{s_0, \dots, s_k \plus \{d:N | \G @ \Q\}\}
    \end{gather*}
  \end{itemize}
The justification of $(\dagger)$ is the following:
\begin{enumerate}
\item Note that $\{c:\{n\} | \F @ \P\} = \{\P(n)\} = \{t_0\}$ \by{assumption of this branch and because in this branch $\P(n) = t_0$}, then $\{c:\{n\} | \F @ \P\} \setminus \{t_0\} = \e$.

\item $\{d:\{n\} | \G @ \Q\} = \{t_0\}$ \by{$\G(n) \land t_0 = \Q(n)$}, then $\{d:\{n\} | \G @ \Q\}  \setminus \{t_0\} = \e$.

\item $t_0 \notin \{c:N | \F @ \P\}$ because if it does then there exists $n' \in N$ such that $\F(n')$ holds and $t_0 = \P(n')$. But, $t_0 = \P(n') \iff t_0 = (n',f_\P(n'))$. On the other hand $t_0 = (n,f_\Q(n))$ and so $(n,f_\Q(n)) = (n',f_\P(n'))$ which implies that $n = n'$. But this is impossible because we have $n' \in N$ and $n \notin N$.

\item $\{t_0, \dots, t_m \plus \{c:\{\bi{s} \plus N\} | \F @ \P\}\} \setminus \{t_0\} = \{t_1, \dots, t_m \plus \{c:\{\bi{s} \plus N\} | \F @ \P\}\}$ \by{assumption: $t_0 \notin \{t_1,\dots, t_m\} \cup \{c:\bi{s} | \F @ \P\}$; and by $t_0 \notin \{c:N | \F @ \P\}$}.

\item $\{s_0, \dots, s_k \plus \{d:N | \G @ \Q\}\} \setminus \{t_0\} = \{s_0, \dots, s_k \plus \{d:N | \G @ \Q\}\}$ \by{$t_0 \notin \{s_0,\dots, s_k\}$ and $n \notin N$}.
\end{enumerate}
\item Assume $\F(n) \land \P(n) = s_j$.

We have to prove that $\{\P(n),t_1, \dots, t_m \plus \{c:\{\bi{s} \plus N\} | \F @
\P\}\} =
                \{s_0, \dots, s_k \plus \{d:N | \G @ \Q\}\}$. We proceed as with the previous branch. In fact, the first case is identical (note that there we used assumptions that are also available in this branch). In the second case, we take the same value for $N$ and start from H as in the previous branch.
    \begin{gather*}
    \{t_0, \dots, t_m \plus \{c:\{\bi{s} \plus X\} | \F @ \P\}\}
      = \{s_0, \dots, s_k \plus \{d:X | \G @ \Q\}\} \\
    \iff \why{$X = \{n \plus N\}$} \\
    \{t_0, \dots, t_m \plus \{c:\{\bi{s},n \plus N\} | \F @ \P\}\}
      = \{s_0, \dots, s_k \plus \{d:\{n \plus N\} | \G @ \Q\}\} \\
    \iff \why{properties of $\plus$ and intensional sets and $\F(n)$ holds} \\
    \{t_0, \dots, t_m \plus \{\P(n) \plus \{c:\{\bi{s} \plus N\} | \F @ \P\}\}\} \\
     \quad{} = \{s_0, \dots, s_k \plus \{d:N | \G @ \Q\}\}
        \cup \{d:\{n\} | \G @ \Q\} \\
    \iff \why{semantics of $\plus$} \\
    \{\P(n),t_0,\dots, t_m \plus \{c:\{\bi{s} \plus N\} | \F @ \P\}\}
      = \{s_0, \dots, s_k \plus \{d:N | \G @ \Q\}\}
        \cup \{d:\{n\} | \G @ \Q\} \\
    \iff \why{subtract $\{t_0\}$ at both sides} \\
    \{\P(n),t_0, \dots, t_m \plus \{c:\{\bi{s} \plus N\} | \F @ \P\}\}
      \setminus \{t_0\} \\
    \quad{} = (\{s_0, \dots, s_k \plus \{d:N | \G @ \Q\}\}
        \cup \{d:\{n\} | \G @ \Q\}) \setminus \{t_0\} \\
    \iff \why{$\setminus$ distributes over $\cup$} \\
    \{\P(n),t_0, \dots, t_m \plus \{c:\{\bi{s} \plus N\} | \F @ \P\}\}
       \setminus \{t_0\} \\
    \quad{} = (\{s_0, \dots, s_k \plus \{d:N | \G @ \Q\}\}
                 \setminus \{t_0\})
        \cup (\{d:\{n\} | \G @ \Q\} \setminus \{t_0\}) \\
    \iff \why{$(\dagger)$} \\
    \{\P(n),t_1, \dots, t_m \plus \{c:\{\bi{s} \plus N\} | \F @ \P\}\}
    = \{s_0, \dots, s_k \plus \{d:N | \G @ \Q\}\}
    \end{gather*}
The justification of $(\dagger)$ is the following:
\begin{enumerate}
\item $\{d:\{n\} | \G @ \Q\} = \{t_0\}$ \by{$\G(n) \land t_0 = \Q(n)$}, then $\{d:\{n\} | \G @ \Q\}  \setminus \{t_0\} = \e$.

\item $t_0 \neq \P(n)$ because if not then $t_0 = s_j$ which is in contraction with \eqref{l:special6:as1}.

\item $t_0 \notin \{c:N | \F @ \P\}$ because if it belongs then there exists $n' \in N$ such that $\F(n')$ holds and $t_0 = \P(n')$. But, $t_0 = \P(n') \iff t_0 = (n',f_\P(n'))$. On the other hand $t_0 = (n,f_\Q(n))$ and so $(n,f_\Q(n)) = (n',f_\P(n'))$ which implies that $n = n'$. But this is impossible because we have $n' \in N$ and $n \notin N$.

\item $\{\P(n),t_0, \dots, t_m \plus \{c:\{\bi{s} \plus N\} | \F @ \P\}\} \setminus \{t_0\} = \{\P(n),t_1, \dots, t_m \plus \{c:\{\bi{s} \plus N\} | \F @ \P\}\}$ \by{assumption: $t_0 \notin \{t_1,\dots, t_m\} \cup \{c:\bi{s} | \F @ \P\}$; by $t_0 \neq \P(n)$; and by $t_0 \notin \{c:N | \F @ \P\}$}.

\item $\{s_0, \dots, s_k \plus \{d:N | \G @ \Q\}\} \setminus \{t_0\} = \{s_0, \dots, s_k \plus \{d:N | \G @ \Q\}\}$ \by{$t_0 \notin \{s_0,\dots, s_k\}$ and $n \notin N$}.
\end{enumerate}
\item Assume $\lnot\F(n)$.

We have to prove that $\{t_1, \dots, t_m \plus \{c:\{\bi{s} \plus N\} | \F @
\P\}\} =
                \{s_0, \dots, s_k \plus \{d:N | \G @ \Q\}\}$.We proceed as in the first branch considered in this lemma. In fact, the first case is identical (note that there we used assumptions that are also available in this branch). In the second case, we take the same value for $N$ and proceed in the same way. The only difference is in the first item of the justification of $(\dagger)$ because now we have $\{c:\{n\} | \F @ \P\} = \e$ because $\lnot\F(n)$.
\end{itemize}

\mbox{}\\
\noindent $\Longleftarrow)$ \\
We will consider only the last three branches and will prove that the equality
holds.

Note that when $\F(n)$ holds then $\{c:\{n\} | \F @ \P\} = \{\P(n)\}$; and when $\F(n)$ does not hold then $\{c:\{n\} | \F @ \P\} = \e$.
\begin{itemize}
\item First branch.
  \begin{gather*}
  \{t_0, \dots, t_m \plus \{c:\{\bi{s} \plus X\} | \F @ \P\}\}
    \why{H} \\
  = \{t_0, \dots, t_m \plus \{c:\{\bi{s} \plus \{n \plus N\}\} | \F @ \P\}\}
    \why{semantics of $\plus$; property of intensional sets} \\
  = \{t_0\}
    \cup \{c:\{n\} | \F @ \P\}
    \cup \{t_1, \dots, t_m \plus \{c:\{\bi{s} \plus N\} | \F @ \P\}\}
     \why{H} \\
  = \{t_0\}
    \cup \{c:\{n\} | \F @ \P\}
    \cup \{s_0, \dots, s_k \plus \{d:N | \G @ \Q\}\}
    \why{$\{c:\{n\} | \F @ \P\} = \{t_0\}$} \\
  = \{t_0\}
    \cup \{s_0, \dots, s_k \plus \{d:N | \G @ \Q\}\}
    \why{semantics of $\plus$} \\
  = \{t_0,s_0, \dots, s_k \plus \{d:N | \G @ \Q\}\}
    \why{semantics of $\plus$} \\
  = \{s_0, \dots, s_k \plus \{t_0 \plus \{d:N | \G @ \Q\}\}\}
    \why{H, semantics of $\plus$ and property of intensional sets} \\
  = \{s_0, \dots, s_k \plus \{d:\{n \plus N\} | \G @ \Q\}\}
    \why{H} \\
  = \{s_0, \dots, s_k \plus \{d:X | \G @ \Q\}\}
  \end{gather*}
\item Second branch.
  \begin{gather*}
  \{t_0, \dots, t_m \plus \{c:\{\bi{s} \plus X\} | \F @ \P\}\}
    \why{semantics of $\plus$ and $X = \{n \plus N\}$} \\
  = \{t_0\}
    \cup \{t_1, \dots, t_m \plus \{c:\{\bi{s},n \plus N\} | \F @ \P\}\}
     \why{$\F(n)$ holds, properties of intensional sets} \\
  = \{t_0\}
    \cup \{t_1, \dots, t_m \plus \{\P(n) \plus \{c:\{\bi{s} \plus N\} | \F @ \P\}\}\}
     \why{semantics of $\plus$} \\
  = \{t_0\}
    \cup \{\P(n),t_1, \dots, t_m \plus \{c:\{\bi{s} \plus N\} | \F @ \P\}\}
     \why{H} \\
  = \{t_0\}
    \cup \{s_0, \dots, s_k \plus \{d:N | \G @ \Q\}\}
    \why{semantics of $\plus$} \\
  = \{t_0,s_0, \dots, s_k \plus \{d:N | \G @ \Q\}\}
    \why{semantics of $\plus$} \\
  = \{s_0, \dots, s_k \plus \{t_0 \plus \{d:N | \G @ \Q\}\}\}
    \why{H, semantics of $\plus$ and property of intensional sets} \\
  = \{s_0, \dots, s_k \plus \{d:\{n \plus N\} | \G @ \Q\}\}
    \why{H} \\
  = \{s_0, \dots, s_k \plus \{d:X | \G @ \Q\}\}
  \end{gather*}
\item Third branch.
  \begin{gather*}
  \{t_0, \dots, t_m \plus \{c:\{\bi{s} \plus X\} | \F @ \P\}\}
    \why{H} \\
  = \{t_0, \dots, t_m \plus \{c:\{\bi{s} \plus \{n \plus N\}\} | \F @ \P\}\}
    \why{semantics of $\plus$; property of intensional sets} \\
  = \{t_0\}
    \cup \{c:\{n\} | \F @ \P\}
    \cup \{t_1, \dots, t_m \plus \{c:\{\bi{s} \plus N\} | \F @ \P\}\}
     \why{H} \\
  = \{t_0\}
    \cup \{c:\{n\} | \F @ \P\}
    \cup \{s_0, \dots, s_k \plus \{d:N | \G @ \Q\}\}
    \why{$\{c:\{n\} | \F @ \P\} = \e$ by $\lnot\F(n)$} \\
  = \{t_0\}
    \cup \{s_0, \dots, s_k \plus \{d:N | \G @ \Q\}\}
    \why{semantics of $\plus$} \\
  = \{t_0,s_0, \dots, s_k \plus \{d:N | \G @ \Q\}\}
    \why{semantics of $\plus$} \\
  = \{s_0, \dots, s_k \plus \{t_0 \plus \{d:N | \G @ \Q\}\}\}
    \why{H, semantics of $\plus$ and property of intensional sets} \\
  = \{s_0, \dots, s_k \plus \{d:\{n \plus N\} | \G @ \Q\}\}
    \why{H} \\
  = \{s_0, \dots, s_k \plus \{d:X | \G @ \Q\}\}
  \end{gather*}
\end{itemize}
\end{proof}

The equivalence of rule \eqref{s1=s2} has been proved previously for any finite
set (extensional or intensional) \cite{Dovier00}.

\begin{lemma}[Equivalence of rule \eqref{app.e:empty2}]
We consider only the following case because the others covered by this rule are
similar.
\begin{gather*}
\begin{split}
\forall d, & D: \\
           & \begin{split}
                & \{x:\{d \plus D\} | \Fpv @ \Ppv\} = \e \\
                & \iff \lnot \Fd \land \{x:D | \Fpv @ \Ppv\} = \e
             \end{split}
\end{split}
\end{gather*}
\end{lemma}

\begin{proof}
Taking any $d$ and $D$ we have:
\begin{gather*}
\{x:\{d \plus D\} | \F @ \P\} = \e  \why{Prop. \ref{l:dD}} \\
\iff \{\Pd : \Fd\} \cup \{\Pz : x \in D \land \Fz\} = \e \\
\iff \{\Pd : \Fd\} = \e \land \{\Pz : x \in D \land \Fz\} = \e \\
\iff \lnot \Fd \land \{x:D | \F @ \P\} = \e
\end{gather*}
\end{proof}

\begin{lemma}[Equivalence of rule \eqref{app.e:eset}]\label{th:eset}
\begin{gather*}
\begin{split}
\forall d, & D, B: \\
           & \begin{split}
             & \{x:\{d \plus D\} | \Fpv @ \Ppv\} = B  \\
             & \begin{split}
               \iff & \Fd \land \{\Pd \plus \{x:D | \Fpv @ \Ppv\}\}
                                = B \\
                    & \lor \lnot \Fd \land \{x:D | \Fpv @ \Ppv\} = B
               \end{split}
\end{split}
\end{split}
\end{gather*}
\end{lemma}

\begin{proof}
Taking any $d$, $D$ and $B$ we have:
\begin{gather*}
\{x:\{d \plus D\} | \Fpv @ \Ppv\} = B \why{Prop. \ref{l:dD}} \\
\iff \{\Pd : \Fd\} \cup \{\Pz : x \in D \land \Fz\} = B \\
\iff \{\Pd : \Fd\} \cup \{x:D | \Fpv @ \Ppv\} = B
\end{gather*}
Now assume $\Fd$, then:
\begin{gather*}
\{\Pd : \Fd\} \cup \{x:D | \Fpv @ \Ppv\} = B \\
\iff \{\Pd\} \cup \{x:D | \Fpv @ \Ppv\} = B \why{semantics of $\plus$} \\
\iff \{\Pd \plus \{x:D | \Fpv @ \Ppv\}\} = B
\end{gather*}
Now assume $\lnot \Fd$, then:
\begin{gather*}
\{\Pd : \Fd\} \cup \{x:D | \Fpv @ \Ppv\} = B \\
\e \cup \{x:D | \Fpv @ \Ppv\} = B \\
\{x:D | \Fpv @ \Ppv\} = B
\end{gather*}
which finishes the proof.
\end{proof}

\begin{remark}
Note that in Theorem \ref{th:eset} when $B$ is:
\begin{itemize}
\item An extensional set of the form $\{y \plus A\}$, then the equality in
the first disjunct becomes an equality between two extensional sets:
\begin{equation*}
\{\Pd \plus \{x:D | \Fpv @ \Ppv\}\} = \{y \plus A\}
\end{equation*}
which is solved by the rules described in \cite{Dovier00}. In turn, the
equality in the second disjunct becomes an equality between a RIS and an
extensional set:
\begin{equation*}
\{x:D | \Fpv @ \Ppv\} = \{y \plus A\}
\end{equation*}
This equality is managed by either rule \eqref{app.e:eset} itself (if $D$ is
not a variable) or by rule \eqref{app.e:v=e} (if $D$ is a variable).

\item A non-variable RIS of the form $\{x:\{e \plus E\} | \Gpv @ \Qpv\}$,
then the equality in the first disjunct becomes an equality between an
extensional set and a non-variable RIS:
\begin{equation*}
\{\Pd \plus \{x:D | \Fpv @ \Ppv\}\} = \{x:\{e \plus E\} | \Gpv @ \Qpv\}
\end{equation*}
which is managed again by rule \eqref{app.e:eset}. In turn, the equality in the
second disjunct becomes an equality between a RIS and a non-variable RIS:
\begin{equation*}
\{x:D | \Fpv @ \Ppv\} = \{x:\{e \plus E\} | \Gpv @ \Qpv\}
\end{equation*}
which is managed by the same rule once more.

Moreover, note that in this case there can be up to four cases (and thus up to
four solutions) considering all the possible combinations of truth values of
$\Fpv$ and $\Gpv$

\item A variable RIS of the form $\{x:E | \Gpv @ \Qpv\}$, then the equality
in the first disjunct becomes an equality between an extensional set and a
variable RIS:
\begin{equation*}
\{\Pd \plus \{x:D | \Fpv @ \Ppv\}\} = \{x:E | \Gpv @ \Qpv\}
\end{equation*}
which is managed by rule \eqref{app.e:v=e}. In turn, the equality in the second
disjunct becomes an equality between a RIS and a variable RIS:
\begin{equation*}
\{x:D | \Fpv @ \Ppv\} = \{x:E | \Gpv @ \Qpv\}
\end{equation*}
which is no further processed (if $D$ is a variable) or is processed by rule
\eqref{app.e:eset} again (if $D$ is not a variable).
\end{itemize}
\qed
\end{remark}

For the following rule we only consider the case where $S \equiv \{c:\bi{X} |
\F @ \P\}$ because the other one covered by the rule corresponds to results
presented in \cite{Dovier00}. In this and the following lemma we will write
$\{\bi{s} \plus X\}$ instead of $\bi{X}$, where $\bi{s}$ denotes zero or more
elements; if $\bi{s}$ denotes zero elements then $\{\bi{s} \plus X\}$ is just
$X$.

\begin{lemma}[Equivalence of rule \eqref{e:v=e2}]
If $c \equiv \P$ and $d \equiv \Q$:
\begin{gather*}
\forall X,t_0,t_1,\dots,t_k: \\
\quad \{d: X  | \G @ \Q\}
      = \{t_0,t_1,\dots,t_k \plus
            \{c:\{\bi{s} \plus X\} | \F @ \P\}\} \\
\quad{}\iff X = \{t_0 \plus N\} \land \Gpv(t_0)
            \land \{d:N | \G @ \Q\}
                  = \{t_1,\dots,t_k \plus
                       \{c:\{\bi{s} \plus N\} | \F @ \P\}\}
\end{gather*}
\end{lemma}

\begin{proof}
In order to simplify the proof we will not write $\P$ and $\Q$ in the
intensional sets because $c \equiv \P$ and $d \equiv \Q$.

\mbox{}\\
\noindent $\implies)$ \\
Clearly, $t_0$ must belong to $\{d: X  | \G\}$, then $t_0 \in X$ and $\G(t_0)$ holds.

If $t_0 \in \{t_1,\dots,t_k\} \cup \{\bi{s}\} \land \F(t_0)$ we take $N = X$.
In this case we have: $\{t_0 \plus N\} = \{t_0 \plus X\} = X$ because $t_0 \in
X$. Now we prove the the last conjunct:
\begin{itemize}
  \item $t_0 \in \{t_1,\dots,t_k\}$
\begin{gather*}
\{t_1,\dots,t_k \plus \{c:\{\bi{s} \plus N\} | \F\}\}
  \why{$\{t_0,t_1,\dots,t_k\} = \{t_1,\dots,t_k\}$ and $X = N$} \\
= \{t_0,t_1,\dots,t_k \plus \{c:\{\bi{s} \plus X\} | \F\}\}
  \why{H} \\
= \{d: X  | \G\}
  \why{$X = N$} \\
= \{d: N  | \G\}
\end{gather*}
\item $t_0 \in \{\bi{s}\}$. Note that in this case $\{c:\{\bi{s} \plus N\} | \F\} =
  \{t_0 \plus \{c:\{\bi{s} \plus N\} | \F\}\}$ because $\F(t_0)$ holds by assumption.
\begin{gather*}
\{t_1,\dots,t_k \plus \{c:\{\bi{s} \plus N\} | \F\}\}
  \why{$\{\bi{s}\} = \{t_0,\bi{s}\}$, $\F(t_0)$ holds and $X = N$} \\
= \{t_0,t_1,\dots,t_k \plus
    \{t_0 \plus \{c:\{\bi{s} \plus N\} | \F\}\}\}
  \why{H} \\
= \{t_0,t_1,\dots,t_k \plus
    \{c:\{\bi{s} \plus N\} | \F\}\}
  \why{semantics of $\plus$} \\
= \{d: X  | \G\}
  \why{$X = N$} \\
= \{d: N  | \G\}
\end{gather*}
\end{itemize}

If $t_0 \notin \{t_1,\dots,t_k\} \cup \{\bi{s}\} \lor \lnot\F(t_0)$ we take $N$
to be $X \setminus \{t_0\}$. Then, $X = \{t_0 \plus N\}$. Besides, $t_0 \notin
N$. Now we will prove the last conjunct by double inclusion.

Take any $x \in \{d:N | \G\}$, then $x \in N$ and $\G(x)$ holds. Hence, $x \in
\{d: X  | \G\}$ because $N \subseteq X$ \by{construction}. Now $x \in
\{t_0,t_1,\dots,t_k \plus \{c:\{\bi{s} \plus X\} | \F\}\}$ \by{H}. If $x = t_0$
then there is a contradiction with the fact that $x \in N$ because $t_0 \notin
N$. So $x$ cannot be $t_0$. If $x = t_i$ ($i \neq 0$), then $x$ trivially
belongs to $\{t_1,\dots,t_k \plus \{c:\{\bi{s} \plus N\} | \F\}$. Finally, if
$x \in \{c:\{\bi{s} \plus X\} | \F\}$, then $\F(x)$ holds. But $x \in N$ so $x
\in \{c:\{\bi{s} \plus N\} | \F\}$.

Now the other inclusion. Take any $x \in \{t_1,\dots,t_k \plus \{c:\{\bi{s}
\plus N\} | \F\}\}$. If $x = t_i$, then $x \in \{t_1,\dots,t_k \plus
\{c:\{\bi{s} \plus X\} | \F\}\}$ and so $x \in \{d: X  | \G\}$ \by{H}, which
implies $x \in X$ and $\G(x)$ holds. Given that $x \neq t_0$ \by{$t_0 \notin
\{t_1,\dots,t_k\}$} and $x \in X$, then $x \in N$ \by{construction}. Since $x
\in N$ and $\G(x)$ holds, then $x \in \{d:N | \G\}$. If $x \in \{c:\{\bi{s}
\plus N\} | \F\}$ then $\F(x)$ holds and $x \neq t_0$ \by{assumption}. Now $x
\in \{t_0,\dots,t_k \plus \{c:\{\bi{s} \plus X\} | \F\}\}$ \by{$N \subseteq
X$}. Then, by H, $x \in \{d: X | \G\}$ which means that $x \in X$ and $\G(x)$
holds. Since $x \neq t_0$, then $x \in X$ implies $x \in N$. Then $x \in \{d:N
| \G\}$.

\bigskip\noindent $\Longleftarrow)$ \\
\begin{gather*}
\{t_0,t_1,\dots,t_k \plus \{c:\{\bi{s} \plus X\} | \F\}\}
  \why{H} \\
= \{t_0,t_1,\dots,t_k \plus \{c:\{\bi{s} \plus \{t_0 \plus N\}\} | \F\}\}
  \why{semantics of $\plus$} \\
= \{t_0,t_1,\dots,t_k \plus \{c:\{\bi{s},t_0 \plus N\} | \F\}\}
  \why{properties of intensional sets and semantics $\plus$} \\
= \{t_0,t_1,\dots,t_k \plus \{c:\{\bi{s} \plus N\} | \F\}\} \cup \{c:\{t_0\} | \F\}
  \why{semantics of $\plus$} \\
= \{t_0\} \cup \{t_1,\dots,t_k \plus \{c:\{\bi{s} \plus N\} | \F\}\} \cup \{c:\{t_0\} | \F\}
  \why{$\{c:\{t_0\} | \F\}  \subseteq \{t_0\}$} \\
= \{t_0\} \cup \{t_1,\dots,t_k \plus \{c:\{\bi{s} \plus N\} | \F\}\}
  \why{H} \\
= \{t_0\} \cup \{d: N | \G\}
  \why{semantics $\plus$} \\
= \{t_0 \plus \{d: N | \G\}
  \why{H, semantics of $\plus$ and properties of intensional sets} \\
= \{d:\{t_0 \plus N\} | \G\}
  \why{H} \\
= \{d: X | \G\}
\end{gather*}
\end{proof}

\begin{lemma}[Equivalence of rule \eqref{special4}]
If $c \not\equiv \P$ and $d \not\equiv \Q$:
\begin{gather*}
\forall X, t_0,t_1,\dots,t_k: \\
\quad \{d:X  | \G @ \Q\}
      = \{t_0,t_1,\dots,t_k \plus
           \{c:\{\bi{s} \plus X\} | \F @ \P\}\} \\
\quad{}\iff X = \{n \plus N\} \land \Gpv(n) \land t_0 = \Qpv(n)
            \land (\F(n) \implies t_0 = \P(n)) \\
            \qquad\quad{}
            \land \{d:N | \G @ \Q\}
                  = \{t_1,\dots,t_k \plus
                       \{c:\{\bi{s} \plus N\} | \F @ \P\}\}
\end{gather*}
\end{lemma}

\begin{proof}
In order to prove this lemma we need to recall that $\P$ and $\Q$ are ordered
pairs whose first component is the control variable and whose second component
is an $\Ur$-term. Then, from now on we will take $\P(x) = (x,f_\P(x))$ for some
term $f_\P$ and $\Q(x) = (x,f_\Q(x))$, for some term $f_\Q$.
\mbox{}\\
\noindent $\implies)$ \\
Clearly, $t_0$ must belong to $\{d:X  | \G @ \Q\}$, which means that there
exists $n \in X$, $\G(n)$ holds and $t_0 = \Q(n)$. Observe that $t_0 = \Q(n)$
means that $t_0 = (n,f_\Q(n))$.

Now will we prove that $\F(n) \implies t_0 = \P(n)$. So assume, $\F(n)$ holds
but $t_0 \neq \P(n)$:
\begin{gather*}
\F(n) \land t_0 \neq \P(n)
  \why{$n \in X$ and definitions of $t_0$ and $\P$} \\
\implies (n,f_\P(n)) \in \{c:\{\bi{s} \plus X\} | \F @ \P\}
         \land (n,f_\Q(n)) \neq (n,f_\P(n))
         \why{H: $\{c:\{\bi{s} \plus X\} | \F @ \P\} \subseteq \{d:X  | \G @ \Q\}$} \\
\implies (n,f_\P(n)) \in \{d:X  | \G @ \Q\}
         \land f_\Q(n) \neq f_\P(n)
         \why{semantics intensional sets and def. of $\Q$}\\
\implies \G(n) \land (n,f_\P(n)) = (n,f_\Q(n))
         \land f_\Q(n) \neq f_\P(n)
         \why{drop $\G(n)$ and equality of pairs} \\
\implies f_\P(n) = f_\Q(n) \land f_\Q(n) \neq f_\P(n)
\end{gather*}
which is clearly a contradiction. Hence, if $\F(n)$ holds then $t_0 = \P(n)$.

Finally, we will prove $\{d:N | \G @ \Q\} = \{t_1,\dots,t_k \plus \{c:\{\bi{s}
\plus N\} | \F @ \P\}\}$ by distinguishing two cases. In the first case, we
assume $t_0 \in \{t_1,\dots,t_k\} \cup \{c:\{\bi{s}\} | \F @ \P\}$, and so we
take $N = X$ which guarantees that $X = \{n \plus N\}$ \by{$X = N$ and $n \in
X$}. Now we prove the set equality (last conjunct):
\begin{itemize}
  \item $t_0 \in \{t_1,\dots,t_k\}$
\begin{gather*}
\{d:N  | \G @ \Q\}
  \why{$N = X$} \\
= \{d:X | \G @ \Q\}
  \why{H} \\
= \{t_0,t_1,\dots,t_k \plus \{c:\{\bi{s} \plus X\} | \F @ \P\}\}
  \why{$t_0 \in \{t_1,\dots,t_k\}$, left absorption} \\
= \{t_1,\dots,t_k \plus \{c:\{\bi{s} \plus X\} | \F @ \P\}\}
  \why{$X = N$} \\
= \{t_1,\dots,t_k \plus \{c:\{\bi{s} \plus N\} | \F @ \P\}\}
\end{gather*}
  \item $t_0 \in \{c:\{\bi{s}\} | \F @ \P\}$. Note that in this case
  $\{c:\{\bi{s} \plus X\} | \F @ \P\} = \{t_0 \plus \{c:\{\bi{s} \plus X\} | \F @ \P\}\}$.
\begin{gather*}
\{d:N  | \G @ \Q\}
  \why{$N = X$} \\
= \{d:X  | \G @ \Q\}
  \why{H} \\
= \{t_0,t_1,\dots,t_k \plus \{c:\{\bi{s} \plus X\} | \F @ \P\}\}
  \why{semantics of $\plus$} \\
= \{t_1,\dots,t_k \plus \{t_0 \plus \{c:\{\bi{s} \plus X\} | \F @ \P\}\}\}
  \why{$t_0 \in \{c:\{\bi{s}\} | \F @ \P\}$} \\
= \{t_1,\dots,t_k \plus \{c:\{\bi{s} \plus X\} | \F @ \P\}\}
  \why{$X = N$} \\
= \{t_1,\dots,t_k \plus \{c:\{\bi{s} \plus N\} | \F @ \P\}\}
\end{gather*}
\end{itemize}

In the second case, we assume $t_0 \notin \{t_1,\dots,t_k\} \cup \{c:\{\bi{s}\}
| \F @ \P\}$, and so we take $N = X \setminus \{n\}$. This means that $X = \{n
\plus N\}$ \by{$n \in X$} and $n \notin N$.

Now we proceed by double inclusion. Take any $(x,f_\Q(x)) \in \{d:N | \G @
\Q\}$, then $x \in N$ and $\G(x)$ holds. As $x \in N$ then $x \neq n$ and so
$(x,f_\Q(x)) \neq (n,f_\Q(n)) = t_0$. Given that $N \subseteq X$, then
$(x,f_\Q(x)) \in \{d:X | \G @ \Q\}$ and so $(x,f_\Q(x)) \in \{t_0,t_1,\dots,t_k
\plus \{c:\{\bi{s} \plus X\} | \F @ \P\}\}$ \by{H}. Given that $(x,f_\Q(x))
\neq t_0$, then $(x,f_\Q(x)) \in \{t_1,\dots,t_k \plus \{c:\{\bi{s} \plus X\} |
\F @ \P\}\}$. But $x \neq n$ and so if $(x,f_\Q(x)) \in \{c:\{\bi{s} \plus X\}
| \F @ \P\}$ then it actually belongs to $\{c:\{\bi{s} \plus N\} | \F @ \P\}$.
Hence $(x,f_\Q(x)) \in \{t_1,\dots,t_k \plus \{c:\{\bi{s} \plus N\} | \F @
\P\}\}$.

Now the other inclusion. Take any $x \in \{t_1,\dots,t_k \plus \{c:\{\bi{s}
\plus N\} | \F @ \P\}\}$. If $x = t_i$ (for some $i \neq 0$) then $x \neq t_0$,
and $x \in \{d:X  | \G @ \Q\}$ \by{H}. So we have $n' \in X$ such that $\G(n')$
holds and $x = (n',f_\Q(n'))$. Now if $n' = n$ then we have $x = (n',f_\Q(n'))
= (n,f_\Q(n)) = t_0$, which is a contradiction with $x \neq t_0$. Hence, $n'$
must be different from $n$ and so $n' \in N$ \by{$n' \in X = \{n \plus N\}$
\by{construction} and $n' \neq n$}. Since $n' \in N$ and $\G(n')$ holds we have
$x = (n',f_\Q(n')) \in \{d:N | \G @ \Q\}$. Finally, if $x \in \{c:\{\bi{s}
\plus N\} | \F @ \P\}$ then $x = (n',f_\P(n'))$ with $n' \in \{\bi{s} \plus
N\}$ and $\F(n')$ holds. If $n' = n$ then $n' \in \{\bi{s}\}$ because $n \notin
N$, and in this case we would have $t_0 \in \{c:\{\bi{s}\} | \F @ \P\}$ because
$\F(n')$ holds, but this a contradiction with the assumption. Hence $n' \neq
n$. On the other hand, $x \in \{c:\{\bi{s} \plus N\} | \F @ \P\}$  implies that
$x = (n',f_\P(n')) \in \{c:\{\bi{s} \plus N\} | \F @ \P\}$ \by{$N \subseteq
X$}, which in turn implies $x = (n',f_\P(n')) \in \{d:X | \G @ \Q\}$ \by{H}.
Then $n' \in X$ and $\G(n')$ holds. Now $n' \in X$ iff $n' \in \{n \plus N\}$
\by{construction}, but we have proved that $n' \neq n$, so actually $n' \in N$.
Therefore, we have $n' \in N$ and $\G(n')$ holds so we have $x = (n',f_\P(n'))
\in \{d:N | \G @ \Q\}$.

\bigskip\noindent $\Longleftarrow)$ \\
\begin{gather*}
\{d:X | \G @ \Q\}
  \why{H} \\
= \{d:\{n \plus N\} | \G @ \Q\}
  \why{semantics of $\plus$ and properties of intensional sets} \\
= \{d:\{n\} | \G @ \Q\} \cup \{d:N | \G @ \Q\}
  \why{H} \\
= \{t_0\} \cup \{d:N | \G @ \Q\}
  \why{H} \\
= \{t_0\} \cup \{t_1,\dots,t_k \plus \{c:\{\bi{s} \plus N\} | \F @ \P\}\}
  \why{semantics of $\plus$} \\
= \{t_0,t_1,\dots,t_k \plus \{c:\{\bi{s} \plus N\} | \F @ \P\}\}
  \why{left absorption} \\
= \{t_0,t_0,t_1,\dots,t_k \plus \{c:\{\bi{s} \plus N\} | \F @ \P\}\}
  \why{semantics of $\plus$} \\
= \{t_0,t_1,\dots,t_k \plus \{t_0 \plus \{c:\{\bi{s} \plus N\} | \F @ \P\}\}\}
  \why{$\Q(x) = (x,f_\Q(x))$ and H: $\Q(n) = t_0$} \\
= \{t_0,t_1,\dots,t_k
    \plus \{(n,f_\Q(n)) \plus \{c:\{\bi{s} \plus N\} | \F @ \P\}\}\}
  \why{H: $\F(n) \implies \P(n) = t_0$; see ($\dagger$) below} \\
= \{t_0,t_1,\dots,t_k \plus \{c:\{n \plus \{\bi{s} \plus N\}\} | \F @ \P\}\}\}
  \why{semantics of $\plus$} \\
= \{t_0,t_1,\dots,t_k \plus \{c:\{n,\bi{s} \plus N\} | \F @ \P\}\}\}
  \why{semantics of $\plus$} \\
= \{t_0,t_1,\dots,t_k \plus \{c:\{\bi{s} \plus \{n \plus N\}\} | \F @ \P\}\}\}
  \why{H} \\
= \{t_0,t_1,\dots,t_k \plus \{c:\{\bi{s} \plus X\} | \F @ \P\}\}\}
\end{gather*}

\noindent ($\dagger$)\hspace{3mm} Note that if $\F(n)$ does not hold then
$\{c:\{n \plus \{\bi{s} \plus N\}\} | \F @ \P\} = \{c:\{\bi{s} \plus N\} | \F @
\P\}$; and if $\F(n)$ holds then $\{c:\{n \plus \{\bi{s} \plus N\}\} | \F @
\P\} = \{(n,f_\P(n))\} \cup \{c:\{\bi{s} \plus N\} | \F @ \P\} = \{t_0\} \cup
\{c:\{\bi{s} \plus N\} | \F @ \P\}$. Besides if $\F(n)$ holds then $f_\P(n) =
f_\Q(n)$ because $t_0 = \Q(n) = \P(n)$.
\end{proof}

\begin{lemma}[Equivalence of rule \eqref{app.e:v=e}]
\begin{gather*}
\begin{split}
\forall & D, y, A: \\
           & \begin{split}
             & \{x:D | \Fpv @ \Ppv\} = \{y \plus A\} \\
             & \iff \exists d, E:
                      D = \{d \plus E\} \land \Fd
                      \land y = \Pd \land \{x:E | \Fpv @ \Ppv\} = A
               \end{split}
\end{split}
\end{gather*}
\end{lemma}

\begin{proof}
\mbox{}\\
\noindent $\implies)$ \\
By H, $y \in \{x:D | \Fpv @ \Ppv\}$; then:
\begin{equation}\label{p:3}
\exists d: d \in D \land \Fd \land \Pd = y \tag{$\exists$1}
\end{equation}
Then we have proved that $\Fd \land \Pd = y$ holds.

It remains to be proved the existence of $E$ and $\{x:E | \Fpv @ \Ppv\} = A$. To this end, the proof is divided in two cases.

In the first case we assume $y \notin A$. Hence, take $E = D \setminus \{d\}$, which verifies $D = \{d \plus E\}$. We will prove $\{x:E | \Fpv @ \Ppv\} = A$ by proving that:
\[
\{x:E | \Fpv @ \Ppv\} \subseteq A \land A \subseteq \{x:E | \Fpv @ \Ppv\}
\]
\begin{itemize}
\item $\{x:E | \Fpv @ \Ppv\} \subseteq A$. Take any $w \in \{x:E | \Fpv @ \Ppv\}$; then
\begin{equation}\label{p:2}
\exists a: a \in E \land \Fpv(a) \land \Ppv(a) = w \tag{$\exists$2}
\end{equation}
As $D = \{d \plus E\}$ then $a \in D$ which implies $w \in \{x:D | \Fpv @
\Ppv\}$ which implies $w \in \{y \plus A\}$, by H. Since $\P$ is a bijective
pattern, then $\Ppv(a) \neq \Ppv(d)$ because $a \neq d$ because $a \in E = D
\setminus \{d\}$. Given that $\Ppv(d) = y$ \by{\eqref{p:3}} and $\Ppv(a) = w$
\by{\eqref{p:2}}, then $w \neq y$, which implies that $w \in A$.

\item $A \subseteq \{x:E | \Fpv @ \Ppv\}$. Take any $w \in A$;
then, by assumption of this case,  $w \neq y$ $(*)$ and $w \in \{x:E | \Fpv @ \Ppv\}$. Hence:
\begin{equation}\label{p:1}
\exists a: a \in D \land \Fpv(a) \land \Ppv(a) = w \tag{$\exists$3}
\end{equation}
So we need to prove that $a \in E$, which by \eqref{p:1} will imply that $w \in
\{x:E | \Fpv @ \Ppv\}$, which will prove this branch.

Given that $D = \{d \plus E\}$ and that $a \in D$, then $a \in E$ iff $a \neq
d$. If $a = d$ then $\Ppv(a) = \Ppv(d)$, then $w =
y$ because $\Ppv(a) = w$ \by{\eqref{p:1}} and $\Ppv(d) = y$
\by{\eqref{p:3}}. And if $w = y$ then there is a contradiction with $(*)$.
Therefore, $a \neq d$ and so $a \in E$.
\end{itemize}

In the second case we assume $y \in A$. Hence, take $E = D$, then $\{d \plus E\} = \{d \plus D\} = D$ because $d \in D$. If  $y \in A$, then $\{y \plus A\} = A$ \by{semantics of $\plus$}. So, by H, we have:
\begin{gather*}
\{x:D | \Fpv @ \Ppv\} = \{y \plus A\} \\
\iff \why{$\{y \plus A\} = A$} \\
\{x:D | \Fpv @ \Ppv\} = A \\
\iff \why{$D = E$} \\
\{x:E | \Fpv @ \Ppv\} = A
\end{gather*}

\noindent $\Longleftarrow)$ \\
By H, let $d$ and $E$ be such that:
\begin{equation}\label{p:4}
D = \{d \plus E\} \land \Fd \land \Pd = y \land \{x:E | \Fpv @ \Ppv\} = A
\end{equation}
Now:
\begin{gather*}
\{x:D | \Fpv @ \Ppv\} = \{y \plus A\} \\
\iff \{x:\{d \plus E\} | \Fpv @ \Ppv\} = \{y \plus A\} \why{$D = \{d \plus E\}$ in \eqref{p:4}} \\
\iff \{\Pd : \Fd\} \cup \{x:E | \Fpv @ \Ppv\} = \{y \plus A\} \why{Lemma \ref{l:dD}} \\
\iff \{\Pd\} \cup \{x:E | \Fpv @ \Ppv\} = \{y \plus A\} \why{$\Fd$ in \eqref{p:4}} \\
\iff \{y\} \cup A  = \{y \plus A\} \why{$\{x:E | \Fpv @ \Ppv\} = A$ in \eqref{p:4}} \\
\iff \true
\end{gather*}
\end{proof}

The equivalence of the rule for $\neq$ (i.e., rule \eqref{e:neq}) is trivial
because this rule applies the definition of set disequality which is given in
terms of $\in$ and $\notin$.

\begin{lemma}[Equivalence of rule \eqref{app.in:V}]
\begin{gather*}
\begin{split}
\forall & D, y: \\
& y \in \{x:D | \F @ \P\} \iff \exists d: d \in D \land \Fd \land y = \Pd
\end{split}
\end{gather*}
\end{lemma}

\begin{proof}
The proof is trivial since this is the definition of set membership w.r.t. an
intensional set.
\end{proof}

\begin{lemma}[Equivalence of rule \eqref{app.nin:nV}]
\begin{gather*}
\begin{split}
\forall & t, d, D: \\
& t \notin \{\{d \plus D\} | \F @ \P\} \\
& \begin{split}
  \iff & \Fdd \land t \neq \Pdd \land t \notin \{D | \F @ \P\} \\
       & \lor \lnot \nFdd \land t \notin \{D | \F @ \P\}
\end{split}
\end{split}
\end{gather*}
\end{lemma}

\begin{proof}
\begin{gather*}
t \notin \{\{d \plus D\} | \F @ \P\} \\
\iff \forall x: x \in \{d \plus D\} \land \Fz \land t \neq \Pz
     \why{RIS semantics} \\
\begin{split}
\iff \forall x: &
       \Fd \land t \neq \Pd \land x \in D \land \Fz \land t \neq \Pz \\
     & \lor \lnot\Fd \land x \in D \land \Fz \land t \neq \Pz
\end{split} \why{single out $d$} \\
\begin{split}
\iff & \Fd \land t \neq \Pd \land \forall x:
       x \in D \land \Fz \land t \neq \Pz \\
     & \lor \lnot\Fd
            \land \forall x: x \in D \land \Fz \land t \neq \Pz
\end{split} \\
\begin{split}
\iff & \Fd \land t \neq \Pd \land t \notin \{D | \F @ \P\} \\
     & \lor \lnot\Fd
            \land t \notin \{D | \F @ \P\}
\end{split} \\
\end{gather*}
\end{proof}

Concerning the $\Cup$ constraint, since rules
\eqref{un:equalvars}-\eqref{un:ext} are extensions of the rules presented by
\cite{Dovier00}, the corresponding proofs are trivial. Rule \eqref{un:ris} is
the only one that truly processes non-variable RIS.

\begin{lemma}[Equivalence of rule \eqref{un:ris}]
If at least one of $A,B,C$ is not a variable nor a variable-RIS:
\begin{gather*}
\Cup(A, B, C) \iff
    \Cup(\svf(A), \svf(B), \svf(C))
    \land \cvf(A) \land \cvf(B) \land \cvf(C)
\end{gather*}
where $\svf$ is a set-valued function
               and $\cvf$ is a constraint-valued function
\begin{gather*}
\svf(\sigma) =
\begin{cases}
N_\sigma & \text{if $\sigma \equiv\defris{\set{d}{D}}$} \notag \\
\sigma & \text{otherwise}
\end{cases} \\[2mm]
\cvf(\sigma) =
\begin{cases}
N_\sigma = (\set{\Pdd}{\defris{D}} \land \Fdd)       & \text{if $\sigma \equiv\defris{\set{d}{D}}$} \notag \\
\quad{}\lor (N_\sigma = \defris{D} \land \lnot \Fdd) & \\
\true & \text{otherwise}
\end{cases}
\end{gather*}
\end{lemma}

\begin{proof}
Actually, this rule covers several cases that could have been written as different rules. Indeed, this rule applies whenever at least one of its arguments is a non variable-RIS. This means that there are seven cases where this rule is applied.

We will prove the equivalence for one of this seven cases because all the other are proved in a similar way. The case to be considered is the following:
\begin{equation}\label{e:risunex}
\Cup(\defris{\set{d}{D}},B,\riss{\set{e}{E}}{\gamma}{v})
\end{equation}
that is, the first and third arguments are non-variable RIS while the second is a variable.

In this case the rule rewrites \eqref{e:risunex} as follows:
\begin{gather*}
\Cup(\svf(\defris{\set{d}{D}}),\svf(B),\svf(\riss{\set{e}{E}}{\gamma}{v})) \\
{}\quad\land \cvf(\defris{\set{d}{D}}) \\
{}\quad\land \cvf(B) \\
{}\quad\land \cvf(\riss{\set{e}{E}}{\gamma}{v}) \iff
    \why{applying the definition of $\svf$ and $\cvf$} \\
\Cup(N_1,B,N_2) \\
{}\quad\land (N_1 = \set{\Pdd}{\defris{D}} \land \Fdd
              \lor N_1 = \defris{D} \land \lnot\Fdd) \\
{}\quad\land \true \\
{}\quad\land (N_2 = \set{v(e)}{\riss{E}{\gamma}{v}} \land \gamma(e)
              \lor N_2 = \riss{E}{\gamma}{v} \land \lnot\gamma(e))
\end{gather*}

Hence, after performing some trivial simplifications we have:
\begin{gather*}
\Fdd
  \land \gamma(e)
  \land \Cup(\set{\Pdd}{\defris{D}},B,\set{v(e)}{\riss{E}{\gamma}{v}}) \\
\lor \\
\Fdd
  \land \lnot\gamma(e)
  \land \Cup(\set{\Pdd}{\defris{D}},B,\riss{E}{\gamma}{v}) \\
\lor \\
\lnot\Fdd
  \land \gamma(e)
  \land \Cup(\defris{D},B,\set{v(e)}{\riss{E}{\gamma}{v}}) \\
\lor \\
\lnot\Fdd
  \land \lnot\gamma(e)
  \land \Cup(\defris{D},B,\riss{E}{\gamma}{v})
\end{gather*}
That is, $\svf$ and $\cvf$ simply transform each non variable-RIS into an extensional set or another RIS depending on whether the `first' element of each domain is true of the corresponding filter or not. In this way, a disjunction covering all the possible combinations is generated.
\end{proof}

As with union, the only rules for $\Disj$ that truly deal with non trivial RIS terms are \eqref{disj:rris} and \eqref{disj:rrisSym}. However, given that they are symmetric, we only prove the equivalence for the first one.

\begin{lemma}[Equivalence of rule \eqref{disj:rris}]
\begin{gather*}
\begin{split}
A \disj & \riss{\set{d}{D}}{\Fpv}{\Ppv} \iff \\
        & \Fdd \land \Pdd \notin A \land A \disj \riss{D}{\Fpv}{\Ppv}
          \lor \lnot\Fdd \land A \disj \riss{D}{\Fpv}{\Ppv}
\end{split}
\end{gather*}
\end{lemma}

\begin{proof}
\begin{gather*}
A \disj \riss{\set{d}{D}}{\Fpv}{\Ppv} \\
\iff A \disj (\{\Pdd | \Fdd\} \cup \{\Pz | x \in D \land \Fz\}) \why{Lemma \ref{l:dD}} \\
\iff A \disj \{\Pdd | \Fdd\} \land A \disj \{\Pz | x \in D \land \Fz\} \why{distribution} \\
\iff (\Fdd \land A \disj \{\Pdd\} \lor \lnot \Fdd \land A \disj \e)
     \land A \disj \{\Pz | x \in D \land \Fz\} \why{singleton comprehension}\\
\iff (\Fdd \land \Pdd \notin A \lor \lnot \Fdd \land A \disj \e)
     \land A \disj \{\Pz | x \in D \land \Fz\} \why{disjointness singleton}\\
\iff \Fdd \land \Pdd \notin A \land A \disj \{\Pz | x \in D \land \Fz\} \why{distribution} \\
{}\quad\lor \lnot\Fdd \land A \disj \{\Pz | x \in D \land \Fz\} \why{RIS semantics} \\
\iff \Fdd \land \Pdd \notin A \land A \disj \riss{D}{\Fpv}{\Ppv}
     \lor \lnot\Fdd \land A \disj \riss{D}{\Fpv}{\Ppv}
\end{gather*}
\end{proof}

The last theorems concern the specialized rules for RUQ given in Figure \ref{f:forall}. As can be seen the only one that does not have a trivial proof is the following.

\begin{lemma}[Equivalence of rule \eqref{forall:iter}]
\begin{gather*}
\forall t, A: t \notin A \implies \\
\quad
\set{t}{A} \cup \risnocp{x:\set{t}{A}}{\Fpv} = \risnocp{x:\set{t}{A}}{\Fpv}
  \iff \Fpv(t) \land A \cup \risnocp{x:A}{\Fpv} = \risnocp{x:A}{\Fpv}
\end{gather*}
\end{lemma}

\begin{proof}
First note that
\begin{equation}\label{eq:disjt}
t \notin A \implies \{t\} \disj A \land \{t\} \disj \risnocp{A}{\Fpv}
\end{equation}
and
\begin{equation}\label{eq:sub1}
\risnocp{x:\{t\}}{\Fpv} \subseteq \{t\}
\end{equation}
and
\begin{equation}\label{eq:teqphit}
\risnocp{x:\{t\}}{\Fpv} = \{t\} \iff \Fpv(t)
\end{equation}

\begin{gather*}
\set{t}{A} \cup \risnocp{x:\set{t}{A}}{\Fpv} = \risnocp{x:\set{t}{A}}{\Fpv} \\
\iff \{t\} \cup A \cup \risnocp{x:\{t\}}{\Fpv} \cup \risnocp{x:A}{\Fpv}
     = \risnocp{x:\{t\}}{\Fpv} \cup \risnocp{x:A}{\Fpv}
       \why{Lemma \ref{l:dD}; semantics $\plus$} \\
\iff \{t\} \cup A \cup \risnocp{x:A}{\Fpv}
     = \risnocp{x:\{t\}}{\Fpv} \cup \risnocp{x:A}{\Fpv}
       \why{\eqref{eq:sub1}; $\{t\}$ in left-hand side} \\
\iff \Fpv(t) \land \{t\} \cup A \cup \risnocp{x:A}{\Fpv}
     = \{t\} \cup \risnocp{x:A}{\Fpv}
       \why{\eqref{eq:teqphit}} \\
\iff \Fpv(t) \land A \cup \risnocp{x:A}{\Fpv}
     = \risnocp{x:A}{\Fpv}
       \why{\eqref{eq:disjt}; basic property disjointedness and union}
\end{gather*}
\end{proof}

\end{document}